%% file: blue.tex
\documentclass{article}
\usepackage[a4paper,margin=1.5in]{geometry}
\usepackage[latin1]{inputenc}
\usepackage{latexsym}
\usepackage{amsmath}
\usepackage{amsthm}
\usepackage{amssymb}
\usepackage{graphicx}
\usepackage{ifthen}
\usepackage{calc}
\usepackage{stmaryrd}
\usepackage{enumitem}
\usepackage{hyperref}
\usepackage{alphalph}
\usepackage{tikz-cd}

\usepackage{tikz}
\usetikzlibrary{arrows,calc}
\usetikzlibrary{shapes.geometric}

\makeatletter
\renewcommand{\section}{\@startsection
  {section}%
  {1}%
  {0mm}%
  {-1\baselineskip}%
  {0.5\baselineskip}%
  {\normalfont\large\bfseries}%
}
\renewcommand{\subsection}{\@startsection
  {subsection}%
  {2}%
  {0mm}%
  {-1\baselineskip}%
  {0.5\baselineskip}%
  {\normalfont\large\itshape}%
}

\renewcommand{\subsubsection}{\@startsection
  {subsubsection}%
  {3}%
  {0mm}%
  {-1\baselineskip}%
  {0.5\baselineskip}%
  {\normalfont\itshape}%
}

\newsavebox{\tempbox}
\renewcommand{\@makecaption}[2]{
  \vspace{10pt}
  \sbox{\tempbox}{\textbf{#1.} #2}
  \ifthenelse{\lengthtest{\wd\tempbox > \linewidth}}{
    \textbf{#1.} #2\par
  }{
    \begin{center}
      \textbf{#1.} #2
    \end{center}
  }
}

\makeatother

\numberwithin{equation}{section}

\numberwithin{figure}{section}

\newtheoremstyle{mythm}%
  {}%
  {}%
  {\itshape}%
  {}%
  {\bfseries}%
  {.}%
  {.5em}%
  {\thmname{#1}~\thmnumber{#2}\ifthenelse{\equal{\thmnote{#3}}{}}{}{~(\thmnote{#3})}}%

\newtheoremstyle{mydefn}%
  {}%
  {}%
  {\upshape}%
  {}%
  {\bfseries}%
  {.}%
  {.5em}%
  {\thmname{#1}~\thmnumber{#2}\ifthenelse{\equal{\thmnote{#3}}{}}{}{~(\thmnote{#3})}}%

\newtheoremstyle{myremark}%
  {}%
  {}%
  {\upshape}%
  {}%
  {\itshape}%
  {.}%
  {.5em}%
  {\thmname{#1}~\thmnumber{#2}\ifthenelse{\equal{\thmnote{#3}}{}}{}{~(\thmnote{#3})}}%

\theoremstyle{mythm}
\newtheorem{theo}{Theorem}[section]
\newtheorem{lem}[theo]{Lemma}
\newtheorem{prop}[theo]{Proposition}
\newtheorem{cor}[theo]{Corollary}

\newtheorem{obs}[theo]{Observation}
\theoremstyle{mydefn}
\newtheorem{defn}[theo]{Definition}
\newtheorem{exa}[theo]{Example}
\newtheorem{exe}[theo]{Exercise}
\theoremstyle{myremark}
\newtheorem{rem}[theo]{Remark}
\theoremstyle{mythm}

\newcommand{\uend}{\hfill$\lrcorner$}
\newcommand{\uende}{\eqno\lrcorner}

\newcounter{claimcounter}
\newenvironment{claim}[1][]{
  \renewcommand{\proof}{\smallskip\par\noindent\textit{Proof. }}
  \medskip\par\noindent%
  \ifthenelse{\equal{#1}{}}{%
    \setcounter{claimcounter}{0}\refstepcounter{claimcounter}\textit{Claim~\arabic{claimcounter}.}
  }{%
    \ifthenelse{\equal{#1}{resume}}{%
      \refstepcounter{claimcounter}\textit{Claim~\arabic{claimcounter}.}
    }{%
      \textit{Claim~#1.}
    }
  }
}{
  \par\medskip
}

\newcommand{\case}[1]{\par\medskip\noindent\textit{Case #1: }}

\newenvironment{cs}{
  \begin{description}
    \renewcommand{\case}[1]{\item[\itshape\mdseries Case ##1:]}
  }{
  \end{description}
}

\newlist{caselist}{description}{10}
\setlist[caselist]{font=\itshape\mdseries}

\setenumerate[1]{label=(\arabic*)}
\newlist{eroman}{enumerate}{2}
\setlist[eroman,1]{label=(\roman*)}
\setlist[eroman,2]{label=(\alph*)}
\newlist{ealph}{enumerate}{1}
\setlist[ealph]{label=(\Alph*)}

\newcounter{nlistcounter}
\newenvironment{nlist}[1]{
  \renewcommand{\thenlistcounter}{\upshape(#1.\arabic{nlistcounter})}
  \begin{list}{\bfseries\thenlistcounter}{%
      \usecounter{nlistcounter}
      \setlength{\labelwidth}{1.5em}%
      \setlength{\leftmargin}{\labelwidth}%
      \addtolength{\leftmargin}{\labelsep}%
      \setlength{\listparindent}{0em}%
      \setlength{\topsep}{5pt}%
      \setlength{\itemsep}{5pt}%
      \setlength{\parsep}{0pt}%
    }
  }{
  \end{list}
}

\definecolor{blau}{RGB}{0,84,159}
\definecolor{hellblau}{RGB}{142,168,229}
\definecolor{petrol}{RGB}{0,97,101}
\definecolor{tuerkis}{RGB}{0,152,161}
\definecolor{gruen}{RGB}{87,171,39}
\definecolor{maigruen}{RGB}{189,205,0}
\definecolor{gelb}{RGB}{255,237,0}
\definecolor{orange}{RGB}{255,128,0}
\definecolor{magenta}{RGB}{227,0,102}
\definecolor{rot}{RGB}{204,7,30}
\definecolor{bordeaux}{RGB}{161,16,53}
\definecolor{violett}{RGB}{97,33,88}
\definecolor{lila}{RGB}{122,111,172}
\definecolor{grey}{gray}{0.7}
\definecolor{mittelblau}{RGB}{0,128,255}
\definecolor{rosa}{RGB}{255,153,204}

\newcommand{\bigmid}{\;\big|\;}
\newcommand{\Bigmid}{\;\Big|\;}
\newcommand{\ceil}[1]{\left\lceil#1\right\rceil}
\newcommand{\floor}[1]{\left\lfloor#1\right\rfloor}
\renewcommand{\mathbf}[1]{\textit{\bfseries #1}}

\renewcommand{\tilde}{\widetilde}
\renewcommand{\hat}{\widehat}
\renewcommand{\bar}{\overline}
\renewcommand{\vec}{\overrightarrow}

\newcommand{\angles}[1]{\left\langle#1\right\rangle}

\renewcommand{\phi}{\varphi}
\renewcommand{\epsilon}{\varepsilon}

\newcommand{\NN}{{\mathbb N}}
\newcommand{\RR}{{\mathbb R}}
\newcommand{\ZZ}{{\mathbb Z}}

\newcommand{\CA}{{\mathcal A}}

\newcommand{\CF}{{\mathcal F}}

\newcommand{\CH}{{\mathcal H}}
\newcommand{\CI}{{\mathcal I}}

\newcommand{\CL}{{\mathcal L}}
\newcommand{\CM}{{\mathcal M}}

\newcommand{\CS}{{\mathcal S}}
\newcommand{\CT}{{\mathcal T}}

\newcommand{\CY}{{\mathcal Y}}

\newcounter{rbcounter}
\usepackage{algorithmic}

\algsetup{linenodelimiter=.}

\usepackage{float}
\newfloat{algorithm}{ht}{alg}
\floatname{algorithm}{Algorithm}

\newcommand{\val}{\operatorname{val}}
\newcommand{\tw}{\operatorname{tw}}

\newcommand{\rk}{\operatorname{rk}}
\newcommand{\bw}{\operatorname{bw}}

\newcommand{\ord}{\operatorname{ord}}
\newcommand{\width}{\operatorname{wd}}
\newcommand{\ad}{\operatorname{ad}}
\newcommand{\Sep}{\operatorname{Sep}}
\newcommand{\At}{\operatorname{At}}
\newcommand{\KT}{\mathfrak T}

\newcommand\contract{\mathord{\downarrow}}
\newcommand\expand{\mathord{\uparrow}}

\newcommand{\comp}{\bar{\raisebox{1ex}{\hspace{0.6em}}}}

\newcommand{\Sing}[1]{{\operatorname{Sing}(#1)}}

\newcommand{\U}{U}

\begin{document}
\title{Tangled up in Blue\\[0.5ex]
  \Large\itshape A Survey on Connectivity, Decompositions, and
  Tangles}
\author{Martin Grohe\\\normalsize RWTH Aachen
  University}
\date{}
\maketitle

\begin{abstract}
  We survey an abstract theory of connectivity, based on symmetric
  submodular set functions. We start by developing Robertson and
  Seymour's \cite{gm19} fundamental duality theory
  between branch decompositions (related to the better-known tree
  decompositions) and so-called tangles, which may be viewed as
  highly connected regions in a connectivity system. We move on to
  studying canonical decompositions of connectivity systems into their
  maximal tangles. Last, but not least, we will discuss algorithmic
  aspect of the theory.
\end{abstract}

\section{Introduction}
Suppose we have some structure, maybe a graph, a hypergraph, or maybe
something entirely different like a set of vectors in Euclidean
space. Let ${\U}$ be the \emph{universe} of our structure. We want to study
partitions, or \emph{separations}, as we prefer to call them, of
${\U}$  (see Figure~\ref{fig:sep}). A \emph{connectivity function} assigns to each separation a nonnegative
integer, which we call the \emph{order} of the separation. For
example, ${\U}$ may be the vertex set of a graph and the order of a
separation (or \emph{cut}) $(X,\bar X)$ could be the number of edges
from $X$ to $\bar X$. This is what is known as ``edge connectivity''
in a graph. Or ${\U}$ could be the edge set of a graph, and the order of a
separation $(X,\bar X)$ the number of vertices incident with an edge
in $X$ and an edge in $\bar X$. We will give precise definitions as
well as many more examples in Section~\ref{sec:cf}.

\begin{figure}
  \centering
  \input{separation}
  \caption{A separation $(X,\bar X)$ of ${\U}$}\label{fig:sep}
\end{figure}

The guiding questions in this survey are the following.
\begin{description}
\item[Question 1:] How can we decompose a connectivity system along low order
  separations?
\item[Question 2:] What are the highly connected regions of a connectivity system?
\end{description}
Obviously, the two questions are complementary: highly connected regions
should be precisely those regions that have no low order
separations. We will see that there is a precise technical duality
that captures this intuition (the Duality Theorem~\ref{theo:duality}).

While it is relatively straightforward to give a satisfactory
definition of decomposition---branch decomposition (see
Section~\ref{sec:dec})---it is less obvious what a ``highly connected
region'' is supposed to be. The fact that we use the unspecific term
``region'' instead of something specific such as ``$k$-connected
component'' already indicates this. Indeed, it is an old and
well-known problem in graph theory, going back to Tutte, to find decompositions of graphs
into $k$-connected components, for any $k\ge 4$. Satisfactory
decompositions of graphs into $k$-connected components only exist for
$k\le 3$. Even for $k=3$ the decomposition is starting to get
elusive, because the 3-connected components of a graph are not
subgraphs: they may contain so-called ``virtual edges'' that are not
present in the graph. The problem (and also a solution to this
problem) can be illustrated on a hexagonal grid (see
Figure~\ref{fig:grid}). To avoid irregularities at the boundary, it is
best to think of the grid as being embedded on a torus. Clearly, such
a grid is not $4$-connected: the three neighbours of any vertex form a
vertex-separator of order $3$ (see
Figure~\ref{fig:hexgrid-sep}(a)). But there is no obvious notion of
``4-connected component'' of such a grid, because the separations of
order $4$ may overlap (see
Figure~\ref{fig:hexgrid-sep}(b)). 

\begin{figure}
  \centering
  \includegraphics[height=5cm]{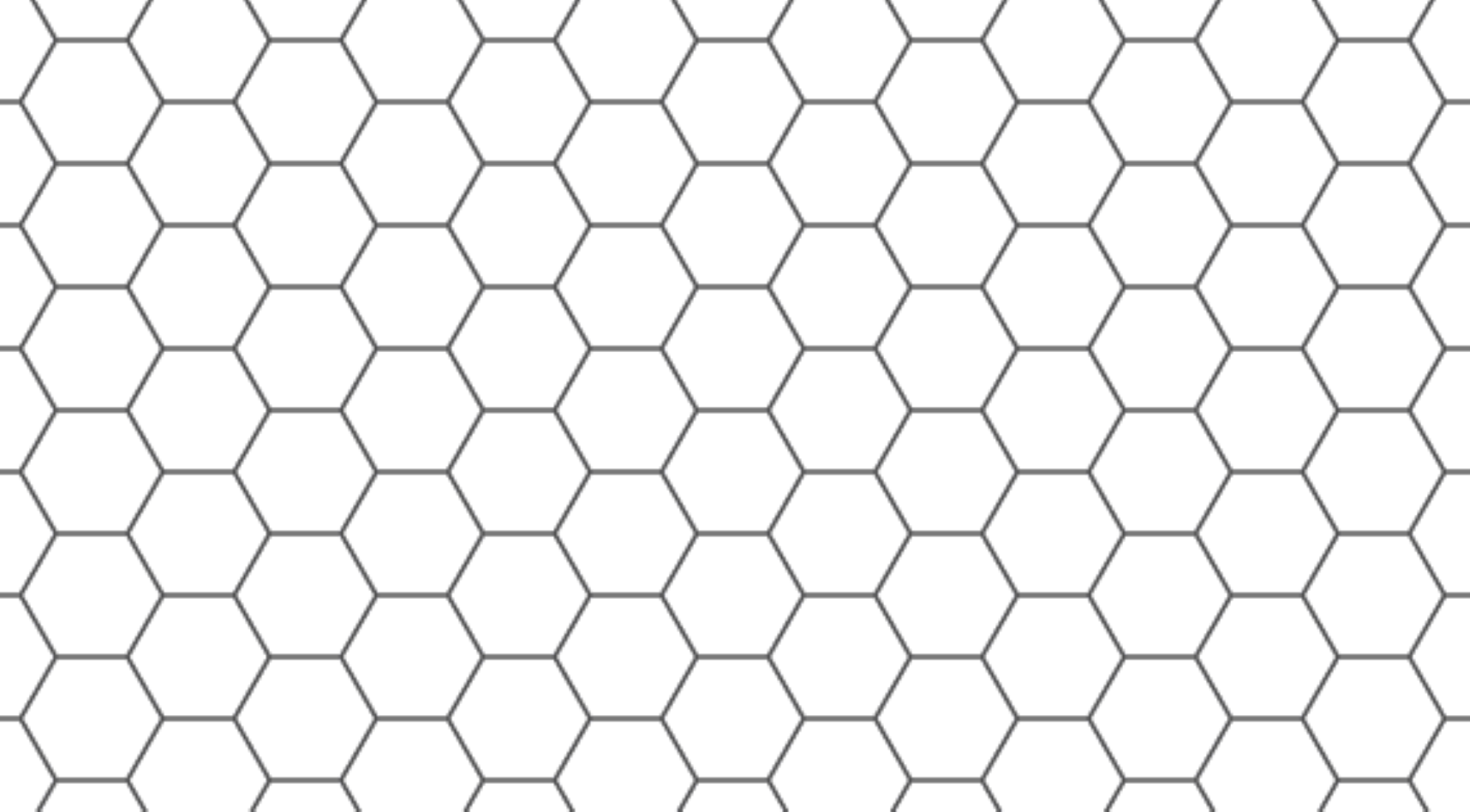}
  \caption{A hexagonal grid}\label{fig:grid}
\end{figure}

\begin{figure}
  \centering
  \input{hexgrid-sep}
  \caption{Separations of a hexagonal grid}\label{fig:hexgrid-sep}
\end{figure}
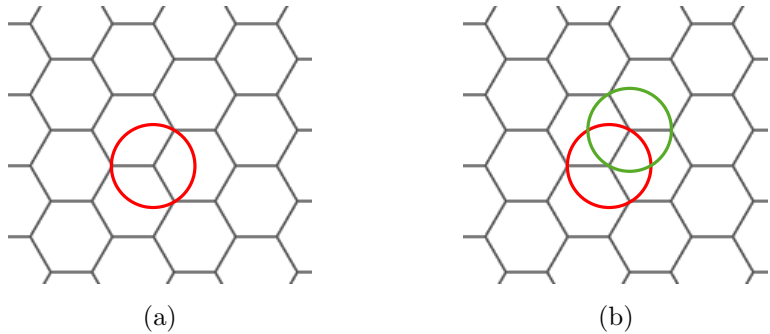

Observe, however, that in every separation of order $3$ of the grid,
one side of separation only consists of a single vertex. If we ignore
separations that are so extremely unbalanced and only look at
separations where both sides have, say, a constant fraction of the
vertices, then the grid suddenly becomes highly connected: if the grid
is square then the smallest balanced separation has order square root
of the number of vertices. This is a well-studied notion of ``high
connectivity''. For example, expander graphs are highly connected in
this sense. It is also the notion of ``high connectivity''
we will be interested in here. It will lead us to \emph{well-linked
  sets} (see Section~\ref{sec:well-linked}) and ultimately \emph{tangles} (see
Section~\ref{sec:tangles}), which describe the highly connected regions in a
connectivity system.  

After establishing the basic duality between decompositions and
tangles (see Section~\ref{sec:duality}), we shall prove that every connectivity
system has a canonical decomposition into its maximal tangles, which
may be viewed as an analogue of the decomposition into $k$-connected
components for $k>3$  (see Section~\ref{sec:can}) . 

The final Section~\ref{sec:alg} is devoted to algorithmic aspects of
the theory.

The goal of this survey is to lay out the basic theory sketched
above. It is not my intention (and beyond my abilities) to
comprehensively cover all results on connectivity functions, 
decompositions, and tangles. %
\section{Connectivity Functions}
\label{sec:cf}

We start by discussing a few basic properties of set functions. Let
${\U}$ be
a finite set, our \emph{universe}. Whenever the
universe ${\U}$ is clear from the context (which it will be most of the
time), we denote the complement ${\U}\setminus X$ of a
set $X\subseteq {\U}$ by $\bar X$. Let $\phi:2^{\U}\to\ZZ$ be an
integer-valued function defined on the subsets of ${\U}$. 
\begin{itemize}
\item $\phi$ is \emph{symmetric} if $\phi(X)=\phi(\bar X)$ for all $X\subseteq {\U}$.
\item $\phi$ is \emph{monotone} if $\phi(X)\subseteq \phi(Y)$ for all $X\subseteq
  Y\subseteq {\U}$.
\item $\phi$ is \emph{submodular} if 
\begin{equation}
  \label{eq:submod}
  \phi(X)+\phi(Y)\ge \phi(X\cap Y)+\phi(X\cup Y)
\end{equation}
for all $X,Y\subseteq {\U}$.
\item $\phi$ is \emph{posimodular} if
\begin{equation}
  \label{eq:posimod}
  \phi(X)+\phi(Y)\ge \phi(X\setminus Y)+\phi(Y\setminus X)
\end{equation}
for all $X,Y\subseteq {\U}$.
\item $\phi$ is \emph{normalised} if $\phi(0)=0$.
\item $\phi$ is \emph{nontrivial} if $\phi(X)\neq 0$ for some
  $X\subseteq {\U}$.
\end{itemize}
Note that for every integer $c$ the function
$\phi-c:X\mapsto\phi(X)-c$ is symmetric, monotone, submodular, posimodular, respectively, if and
only if $\phi$ is. In particular, this is the case for the normalised function
$\phi_0:=\phi-\phi(\emptyset)$. For this reason, we can usually assume,
without loss of generality, that our set functions are normalised.
The \emph{valence} of $\phi$ is 
\[
\val(\phi):=\max\{|\phi(\{u\})-\phi(0)|\mid u\in {\U}\}
\]
if ${\U}\neq\emptyset$ and $\val(\phi):=0$ if ${\U}=\emptyset$.
Of course if $\phi$ is normalised and ${\U}$ is nonempty, $\val(\phi)$ is just the maximum of
the singleton values of $\phi$. We call $\phi$ \emph{univalent} if
$\val(\phi)=1$.

\begin{defn}
  A \emph{connectivity function} on ${\U}$ is a normalised, symmetric, and
  submodular set function $\kappa:2^{\U}\to\ZZ$.
\end{defn}

If $\kappa$ is a connectivity function on ${\U}$, we call the pair
$({\U},\kappa)$ a \emph{connectivity system}. Before we give examples,
we collect a few basic properties of connectivity functions in the
following lemma.

\begin{lem}\label{lem:bp}
  Let $\kappa$ be a connectivity function on ${\U}$.
  \begin{enumerate}
  \item $\kappa$ is posimodular.
  \item $\kappa$ is nonnegative, that is, $\kappa(X)\ge 0$ for all $X\subseteq {\U}$.
  \item $|\kappa(X)-\kappa(Y)|\le\sum_{x\in X\triangle Y}\kappa(x)\le\val(\kappa)\cdot|X\triangle Y|$, for all
    $X,Y\subseteq {\U}$.
  \end{enumerate}
\end{lem}

\noindent
By $X\triangle Y$ we denote the symmetric difference of $X$ and $Y$.

\begin{proof}[Proof of Lemma~\ref{lem:bp}]
  To prove (1), let $X,Y\subseteq {\U}$. Then
  \begin{align*}
    \kappa(X)+\kappa(Y)&=\kappa(X)+\kappa(\bar Y)&\text{by symmetry}\\
    &\ge\kappa(X\cap\bar Y)+\kappa(X\cup\bar Y)&\text{by
                                               submodularity}\\
    &\ge\kappa(X\cap\bar Y)+\kappa(\bar X\cap Y)&\text{by
                                               symmetry}\\
                       &=\kappa(X\setminus Y)+\kappa(Y\setminus X).
  \end{align*}
  To prove (2), let $X\subseteq {\U}$. Then
  \[
  2\cdot\kappa(X)=\kappa(X)+\kappa(\bar
  X)\ge\kappa(\emptyset)+\kappa({\U})=2\cdot\kappa(\emptyset)=0.
  \]
  To prove (3), it clearly suffices to prove that for every $X\subseteq {\U}$ and
  $x\in {\U}\setminus X$ we have 
  \[
  \kappa(X)-\kappa(\{x\})\le\kappa(X\cup\{x\})\le \kappa(X)+\kappa(\{x\}).
  \]
  Indeed, 
  \[
  \kappa(X)+\kappa(\{x\})\ge
  \kappa(\emptyset)+\kappa(X\cup\{x\})=\kappa(X\cup\{x\}),
  \]
  which implies the second inequality, and
  \[
  \kappa(X\cup\{x\})+\kappa(\{x\})=\kappa(X\cup\{x\})+\kappa(\bar{\{x\}})\ge\kappa(X)+\kappa({\U})=\kappa(X),
  \]
  which implies the first inequality.
\end{proof}

\begin{cor}
   A connectivity function $\kappa$ is nontrivial if and only if
   $\val(\kappa)\ge 1$.
\end{cor}

It may also be worth noting that the trivial function $X\mapsto 0$ is
the only connectivity function on a set ${\U}$ of cardinality
$|{\U}|\le 1$. We denote the trivial connectivity function on the empty
set by $\kappa_{\emptyset}$.

Let $\kappa$ be a connectivity function on a set ${\U}$, and let
$X,Y\subseteq {\U}$ be disjoint. An \emph{$(X,Y)$-separation} is a set $Z$ such
that $X\subseteq Z\subseteq\bar Y$. Observe that if $Z$ is an
$(X,Y)$-separation then $\bar Z$ is a $(Y,X)$-separation. An
$(X,Y)$-separation $Z$ is \emph{minimum} if its order $\kappa(Z)$ is minimal.

The following lemma gives a first indication of the value of
submodularity in this context.

\begin{lem}
  Let $\kappa$ be a connectivity function on a set ${\U}$, and let
  $X,Y\subseteq {\U}$ be disjoint. Then there is a (unique) minimum
  $(X,Y)$-separation $Z$ such that
  $Z\subseteq Z'$ for all minimum $(X,Y)$-separations
  $Z'$.

  We call $Z$ the \emph{leftmost minimum $(X,Y)$-separation}.
\end{lem}

\begin{proof}
  Let $Z$ be a minimum $(X,Y)$ separation of minimum cardinality
  $|Z|$, and let $Z'$ be another minimum $(X,Y)$-separation. Then both
  $Z\cap Z'$ and $Z\cup Z'$ are $(X,Y)$-separations, and thus
  $\kappa(Z\cap Z'), \kappa(Z\cup Z')\ge\kappa(Z)=\kappa(Z')$. By
  submodularity, this implies $\kappa(Z\cap Z')=\kappa(Z\cup
  Z')=\kappa(Z)=\kappa(Z')$. By the minimality of $|Z|$, we have
  $|Z|\le|Z\cap Z'|$, and this implies $Z\subseteq Z'$.
\end{proof}

\subsection{Examples}
Before we move on with the theory, we consider a number of examples of
connectivity functions from different domains. We discuss these
examples in great detail; in particular, we often give full (and
tedious) proofs of submodularity. The reader should feel free to skip
these proofs. To get a feeling for how these proofs go, I do recommend
to look at the proof in Example~\ref{exa:ec}, which is the simplest.

Our first two examples capture precisely what is known as
edge-connectivity and vertex-connectivity in a graph.

\begin{exa}[Edge Connectivity]\label{exa:ec}
  Let $G$ be a graph. For all sets $X,Y\subseteq V(G)$ we let
  $E(X,Y)$ be the set of all edges with one endvertex in $X$ and one
  endvertex in $Y$. We define the \emph{edge-connectivity function}
  $\nu_G$ on $V(G)$ by
  \[
  \nu_G(X):=|E(X,\bar X)|.
  \]
  We claim that $\nu_G$ is a connectivity function. We obviously have
  $\nu_G(\emptyset)=0$. The function $\nu_G$ is symmetric, because
  $E(X,\bar X)=E(\bar X,X)$. To see that it is submodular, let
  $X,Y\subseteq V(G)$. We have
  \begin{align*}
    \nu_G(X)=\,&|E(X\cap Y,\bar X\cap Y)|+|E(X\cap Y,\bar X\cap \bar
               Y)|\\&+|E(X\cap \bar Y,\bar X\cap Y)|+|E(X\cap \bar Y,\bar
               X\cap \bar Y)|,\\
    \nu_G(Y)=\,&|E(X\cap Y,X\cap \bar Y)|+|E(X\cap Y,\bar X\cap \bar
               Y)|\\&+|E(\bar X\cap Y,X\cap \bar Y)|+|E(\bar X\cap Y,\bar
               X\cap \bar Y)|,\\
    \nu_G(X\cap Y)=\,&|E(X\cap Y,\bar X\cap Y)|+|E(X\cap Y,\bar X\cap \bar
               Y)|+|E(X\cap Y,X\cap \bar Y)|,\\
    \nu_G(X\cup Y)=\,&|E(X\cap Y,\bar X\cap \bar Y)|+|E(\bar X\cap Y,\bar X\cap \bar
               Y)|+|E(X\cap \bar Y,\bar X\cap \bar Y)|
  \end{align*}                 
  (see Figure~\ref{fig:submod}).
  Comparing the sums of the first two and the last two equations yields the submodularity inequality
  \[
  \nu_G(X)+\nu_G(Y)\ge\nu_G(X\cap Y)+\nu_G(X\cup Y).
  \]
  Hence $\nu_G$ is a connectivity function. Note that $\val(\nu_G)$ is
  the maximum degree of $G$.
  \uend
\end{exa}

\begin{figure}
  \centering
  \begin{tikzpicture}
    \draw (-2,-2) rectangle (2,2) (-2,0)--(2,0) (0,-2)--(0,2);
    \fill[black!50] (0,0) rectangle (-2,2);
    \fill[black!25] (0,0) rectangle (2,2) (0,0) rectangle (-2,-2);

    \path (-2.3,1) node {$X$} (-2.3,-1) node {$\bar X$};
    \path (-1,2.3) node {$Y$} (1,2.3) node {$\bar Y$};
    \path (-1,1) node {$\color{white}X\cap Y$} (-1,-1) node {$\bar X\cap Y$} 
          (1,-1) node {$\bar X\cap\bar Y$} (1,1) node {$X\cap \bar Y$};
  \end{tikzpicture}
  \caption{Crossing separations}
  \label{fig:submod}
\end{figure}
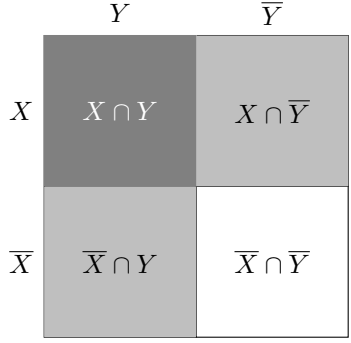

\begin{exa}[Vertex Connectivity, \cite{gm10}]\label{exa:conn1}
  Let $G$ be a graph. We define the \emph{boundary} $\partial(X)$ of an edge set
  $X\subseteq E(G)$ to be the set of vertices incident with both an edge in $X$
  and an edge in $E(G)\setminus X$. We define the
  \emph{vertex-connectivity function} $\kappa_G$ on the edge set
  $E(G)$ by 
  \[
  \kappa_G(X):=|\partial(X)|.
  \]
  for all $X\subseteq E(G)$. 

  We claim that $\kappa_G$ is a connectivity function. Obviously,
  $\kappa_G(\emptyset)=\emptyset$ and $\kappa_G$ is symmetric. To
  prove that it is submodular, let $X,Y\subseteq E(G)$. We need to prove
  \begin{equation}
    \label{eq:submod1}
  \kappa_G(X)+\kappa_G(Y)\ge \kappa_G(X\cap Y)+\kappa_G(X\cup Y).
  \end{equation}
  On the right-hand side of the
  inequality \eqref{eq:submod1},
  \begin{eroman}
  \item all vertices incident with edges in both $X\cap Y$, $\bar
    X\cap \bar Y$ are counted twice, and 
  \item of the remaining vertices all vertices incident with edges in
    both $X\cap Y$, $\bar X\cap Y$ or both $X\cap Y$, $X\cap \bar Y$
    are counted once in $\kappa_G(X\cap Y)$, and all vertices incident with edges in
    both $\bar X\cap Y$, $\bar X\cap \bar Y$ or both $X\cap \bar Y$,
    $\bar X\cap \bar Y$
    are counted once in $\kappa_G(X\cup Y)$
  \end{eroman}
  (again, Figure~\ref{fig:submod} may be helpful).

  On the left-hand side, the vertices in (i) are counted twice as
  well, and the vertices in (ii) are counted at least once, those
  incident with edges in both $X\cap Y$, $\bar X\cap Y$ and those
  incident with edges in both $X\cap \bar Y$, $\bar X\cap \bar Y$ in $\kappa_G(X)$
  and those incident with edges in both $X\cap Y$, $X\cap \bar Y$ and those
  incident with edges in both $\bar X\cap Y$, $\bar X\cap \bar Y$ in
  $\kappa_G(Y)$. This proves the inequality.

  Note that $\val(\kappa_G)\le 2$, where equality holds if and only
  if $G$ contains a triangle or a path of length $3$.
\uend
\end{exa}

\begin{exa}[Hypergraph Connectivity]\label{exa:hc}
  We can easily generalise the edge-con\-nectivity and
  vertex-connectivity functions from graphs to hypergraphs. Let $H$ be a
  hypergraph with vertex set $V(H)$ and edge set $E(H)\subseteq
  2^{V(H)}$. We define the \emph{edge-connectivity} function
  $\nu_H:2^{V(H)}\to\NN$ by 
  \[
  \nu_H(X):=\big|\{e\in E(H)\bigmid e\cap X\neq\emptyset\text{ and
  }e\cap\bar X\neq\emptyset\}\big|
  \]
  and the \emph{vertex-connectivity} function
  $\kappa_H:2^{E(H)}\to\NN$ by 
  \[
  \kappa_H(Y):=|\partial(Y)|=\big|\{v\in V(H)\bigmid \exists e\in Y,e'\in \bar
  Y:\;v\in e\cap e'\}\big|.
  \]
  We leave it as an exercise to the reader to verify that these are
  indeed connectivity functions. 

  Note the duality between the
  two functions: if by $\tilde H$ we denote the dual hypergraph with
  vertex set $E(H)$ and edges $e_v:=\{e\in E(H)\mid v\in e\}$ for all
  $v\in V$, which is actually a multi-hypergraph, we have $\nu_{\tilde
    H}=\kappa_H$ and, identifying each
  vertex $v\in V(H)$ with the edge $e_v\in E(\tilde H)$, also
  $\kappa_{\tilde H}=\nu_{H}$.
  \uend
\end{exa}

\begin{exa}[Matching Connectivity, \cite{vat12} (also see \cite{jeohortel15})]\label{exa:mc}
  There is an alternative connectivity function capturing
  vertex connectivity in a graph $G$. As opposed to the function
  $\kappa_G$, it is defined on the vertex
  set of $G$. For disjoint subsets $X,Y\subseteq V(G)$, an
  \emph{$(X,Y)$-matching} is a set $M\subseteq E(G)$ of mutually
  disjoint edges that all
  have an endvertex in $X$ and an endvertex in $Y$. Here we call two
  edges \emph{disjoint} if they do not have an endvertex in common. A
  \emph{maximum} $(X,Y)$-matching is an $(X,Y)$-matching of maximum
  cardinality. By König's theorem, the maximum cardinality of an
  $(X,Y)$-matching is equal to the minimum cardinality of a vertex cover for
  the set $E(X,Y)$ of edges from $X$ to $Y$, where a \emph{vertex cover} for a
  set $F$ of edges is a set $S$ of vertices such that each edge in $F$
  has at least one endvertex in $S$.

  We define the \emph{matching connectivity function} $\mu_G$ to be the set
  function on $V(G)$ defined by
  \[
  \mu_G(X)=\text{maximum cardinality of an $(X,\bar X)$-matching.}
  \]
  We claim that $\mu_G$ is a connectivity function. Obviously,
  $\mu_G(\emptyset)=\emptyset$ and $\mu_G$ is symmetric. To prove that
  it is submodular, let $X,Y\subseteq V(G)$.  
  Let $M_{\cap}\subseteq E(G)$
  be a maximum $(X\cap Y,\bar{X\cap Y})$-matching, and let $M_{\cup}\subseteq E(G)$
  be a maximum $(X\cup Y,\bar{X\cup Y})$-matching. We define subsets
  $M_1,\ldots, M_6$ of $M_\cap\cup M_\cup$ as 
  follows.
  \begin{itemize}
  \item $M_1$ is the set of all edges in $M_{\cap}$ from $X\cap Y$ to
    $\bar X\cap Y$.
  \item $M_2$ is the set of all edges in $M_{\cap}$ from $X\cap Y$ to
    $X\cap \bar Y$.
  \item $M_3$ is the set of all edges in $M_{\cup}$ from $\bar X\cap Y$ to
    $\bar X\cap \bar Y$.
  \item $M_4$ is the set of all edges in $M_{\cup}$ from $X\cap \bar Y$ to
    $\bar X\cap \bar Y$.
  \item $M_5$ is the set of all edges in $M_\cap\cup M_\cup$ from
    $X\cap Y$ to $\bar X\cap\bar Y$ that are disjoint from all edges in $M_1\cup M_4$
  \item $M_6$ is the set of all edges in $M_\cap\cup M_\cup$ from
    $X\cap Y$ to $\bar X\cap\bar Y$ that are disjoint from all edges in $M_2\cup M_3$
  \end{itemize}
  We claim that $M_1\cup\ldots\cup M_6=M_\cap\cup M_\cup$. An easy
  inspection (Figure~\ref{fig:submod} may help again) shows that it
  suffices to prove that all edges in $M_\cap\cup M_\cup$ from
  $X\cap Y$ to $\bar X\cap\bar Y$ are either in $M_5$ or in
  $M_6$. Suppose for contradiction that $e=vw\in M_\cap\cup M_\cup$ is
  an edge with $v\in X\cap Y$ and $w\in\bar X\cap\bar Y$ that is
  neither in $M_5$ nor $M_6$. If $e\in M_{\cap}$, then $e$ is disjoint
  from all edges in $M_1\cup M_2\subseteq M_{\cap}$ and thus there are
  edges $e'=v'w'\in M_4$ and $e''=v''w''\in M_3$ that share an
  endvertex with $e$. Say, $v'\in X\cap\bar Y$ and
  $w'\in \bar X\cap \bar Y$ and $v''\in \bar X\cap Y$ and
  $w''\in \bar X\cap \bar Y$. As $v\in X\cap Y$ and thus
  $v\neq v',v''$, we have $w=w'=w''$. However, as
  $e',e''\in M_{\cup}$, this contradicts $M_{\cup}$ being a
  matching. The case $e\in M_{\cup}$ is symmetric. This proves
  $M_1\cup\ldots\cup M_6=M_\cap\cup M_\cup$.

  Observe next that $M_\cap\cap M_\cup\subseteq M_5\cap M_6$. To see
  this, let $e\in M_\cap\cap M_\cup$. Then $e$ is an edge from 
  $X\cap Y$ to $\bar X\cap\bar Y$. As $e\in M_\cap$, it is disjoint from
  all edges in $M_1\cup M_2\subseteq M_\cap$, and as $e\in M_\cup$, it is disjoint from
  all edges in $M_3\cup M_4\subseteq M_\cup$. Thus $e\in M_5\cap M_6$.

  Finally, observe that 
  $M_1\cup M_4\cup M_5$ is
  an $(X,\bar X)$-matching and and $M_2\cup M_3\cup M_6$ is a $(Y,\bar
  Y)$-matching. Thus we have
  \begin{align*}
    \mu_G(X)+\mu_G(Y)&\ge |M_1\cup M_4\cup M_5|+|M_2\cup M_3\cup
                       M_6|\\
    &\ge|M_1\cup\ldots\cup M_6|+|M_5\cap M_6|\\
    &\ge|M_\cap\cup M_\cup|+|M_\cap\cap M_\cup|\\
    &=|M_\cap|+|M_\cup|\\
    &=\mu_G(X\cap Y)+\mu_G(X\cup Y).
  \end{align*}
  Note that $\mu_G$ is either trivial (if $E(G)=\emptyset$) or univalent.\uend
\end{exa}

Let us show that $\mu_G$ is closely related to $\kappa_G$ and hence
that it also captures vertex-connectivity in $G$. For a set
$Y\subseteq E(G)$ of edges, we let $V(Y)$ be the set of all
endvertices of edges in $Y$. Then $\partial(Y)=V(Y)\cap V(\bar
Y)$. For every $X\subseteq V(G)$, we let $E(X):=E(X,X)$ be the set of
all edges with both endvertices in $X$. Then $E(G)=E(X)\cup E(\bar
X)\cup E(X,\bar X)$.

\begin{lem}\label{lem:mu-vs-kappa}
  Let $G$ be a graph.
  \begin{enumerate}
  \item Let $Y\subseteq E(G)$ and $X\subseteq V(G)$ such
    that 
    $
    V(Y)\setminus\partial(Y)\subseteq X\subseteq V(Y).
    $
    Then $\mu_G(X)\le\kappa_G(Y)$.
  \item Let $X\subseteq V(G)$. Then there is a
  $Y\subseteq E(G)$ such that $E(X)\subseteq Y \subseteq E(X)\cup
  E(X,\bar X)$ and $\kappa_G(Y)\le\mu_G(X)$. 
  \end{enumerate}
\end{lem}

\begin{proof}
  To prove (1), note that every edge in $E(X,\bar X)$ has at least one
  endvertex in $\partial(Y)$. In other words: $\partial(Y)$ is a
  vertex cover of $E(X,\bar X)$. Thus
  $\mu_G(X)\le|\partial(Y)|=\kappa_G(Y)$.

  To prove (2), let $S$ be a minimum vertex cover of $E(X,\bar X)$. 
  Then every edge in $E(G)$ either has
  both endvertices in $X\cup S$ or both endvertices in $\bar X\cup
  S$. Let $Y\subseteq E(G)$ such that all edges with one endvertex in
  $X\setminus S$ are in $Y$ and all edges with one endvertex in
  $\bar X\setminus S$ are in $\bar Y$. Then $\partial(Y)\subseteq S$
  and thus $\kappa_G(Y)\le |S|=\mu_G(X)$.
\end{proof}

While the connectivity functions of the previous examples capture
natural notions of connectivity, the function defined in the next
example, introduces an entirely different notion of ``connectivity''
on graphs. Instead of the ``flow'' that can be send across a separation, it
measures how ``complicated'' a separation is. 

\begin{figure}
  \centering
  \input{cutrank}
  \caption{The cut-rank function}
  \label{fig:cutrank}
\end{figure}
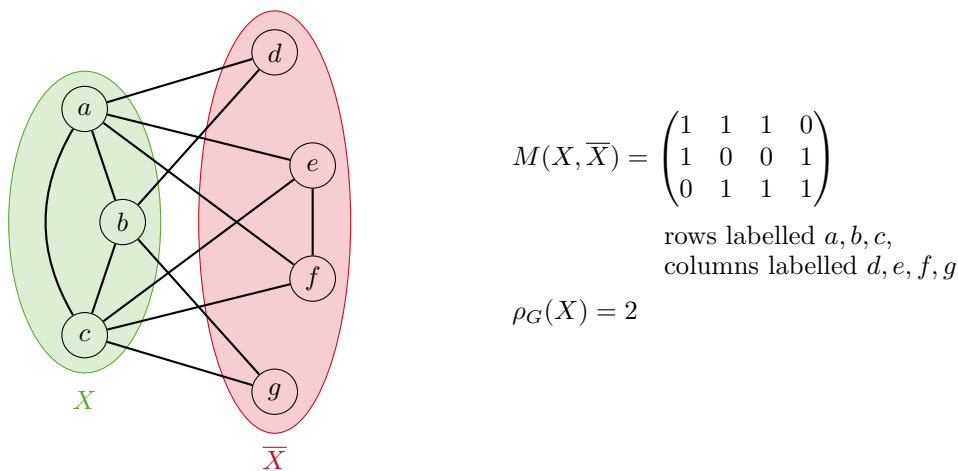

\begin{exa}[Cut Rank, \cite{oumsey06}]\label{exa:cr}
  Let $G$ be a graph. For all subsets $X,Y\subseteq V(G)$, we let
  $M=M(X,Y)$ be the $X\times Y$-matrix with entries 
  \[
  M_{xy}:=
  \begin{cases}
    1&\text{if }xy\in E(G),\\
    0&\text{otherwise}
  \end{cases}
  \]
  for all $x\in X,y\in Y$. That is, $M(X,Y)$ is the submatrix of the
  adjacency matrix of $G$ with rows indexed by vertices in $X$ and
  columns indexed by vertices in $Y$.
  We view $M$ as a matrix over the 2-element field $\mathbb F_2$ and
  denote its row rank by $\rk_2(M)$. 

  We define the \emph{cut rank function} of $G$ to be the set function
  $\rho_G$ on $V(G)$ defined by
  \[
  \rho_G(X):=\rk_2(M_G(X,\bar X)).
  \]
  Figure~\ref{fig:cutrank} shows an example.

  We claim that $\rho_G$ is a connectivity function. We have
  $\rho_G(\emptyset)=0$ because, by definition, the empty matrix has
  rank $0$. The function $\rho_G$ is symmetric,
  because the row rank and the column rank of a matrix coincide. We
  prove that $\rho_G$ is submodular by induction on $|V(G)|$.

  The base step $|V(G)|=0$ is trivial. For the inductive step, suppose that
  $|V(G)|>0$. Let $X,Y\subseteq V(G)$ and
  \begin{align*}
    M_X&:=M(X,\bar X),&M_Y&:= M(Y,\bar Y),\\
    M_{X\cap Y}&:=M(X\cap Y,\bar{X\cap Y}),&M_{X\cup Y}&:=M(X\cup Y,\bar{X\cup Y})
  \end{align*}
  We shall prove that
  \begin{equation}
    \label{eq:submod2}
    \rk_2(M_X)+\rk_2(M_Y)\ge \rk_2(M_{X\cap Y})+\rk_2(M_{X\cup Y}).
  \end{equation}
  If $X\subseteq Y$, then $X\cap Y=X$ and $X\cup Y=Y$, and
  \eqref{eq:submod2} holds trivially. So suppose that $X\not\subseteq
  Y$ and let $x\in X\setminus Y$. We let $X':=X\setminus\{x\}$ and
  $G':=G\setminus\{x\}$. We define further matrices:
  \begin{align*}
    A&:=M(X\cap Y,\bar X\cap Y),&B&:=M(X\cap Y,\bar X\cap\bar
    Y),\\
    C&:=M(X\cap \bar Y,\bar X\cap Y),&D&:=M(X\cap \bar Y,\bar X\cap\bar
    Y),\\
    E&:=M(X\cap Y,X\cap \bar Y),&F&:=M(\bar X\cap Y,\bar X\cap\bar
    Y).
  \end{align*}
  (Figure~\ref{fig:submod} may be helpful again to sort out the sets.)
  Then
  \begin{equation}\label{eq:submod2z}
  \begin{array}{r@{\,}l@{\hspace{2cm}}r@{\,}l}
 M_X&=
  \begin{pmatrix}
    A&B\\C&D
  \end{pmatrix},&
 M_Y&=
  \begin{pmatrix}
    E&B\\C^t&F
  \end{pmatrix},\\[4ex]
  M_{X\cap Y}&=
  \begin{pmatrix}
    A&E&B
  \end{pmatrix},&
  M_{X\cup Y}&=
  \begin{pmatrix}
    B\\F\\D
  \end{pmatrix}
  \end{array}
  \end{equation}
  For all these matrices $M$, we denote the corresponding matrix in
  the graph $G'=G\setminus\{x\}$ by $M'$. For example, $M_X'=M_{G'}(X',\bar
  X)$ and $A'=M_{G'}(X'\cap
  Y,\bar X\cap Y)$. Note that the matrix equalities
  \eqref{eq:submod2z} also hold for the
  ``primed'' versions of the matrices, and that for all the matrices
  $M$ we have 
  \begin{equation}
    \label{eq:submod2a}
\rk_2(M)-1\le\rk_2(M')\le \rk_2(M).
\end{equation}
By the
  inductive hypothesis, we have
  \begin{equation}
    \label{eq:submod2b}
    \rk_2(M'_X)+\rk_2(M'_Y)\ge \rk_2(M'_{X\cap Y})+\rk_2(M'_{X\cup Y}).
  \end{equation}
  Observe next that
  \begin{equation}\label{eq:submod2c}
  \rk_2(M_X)=\rk_2(M'_X)\quad\implies\quad \rk_2(M_{X\cup
    Y})=\rk_2(M'_{X\cup Y}).
  \end{equation}
  Indeed, if $\rk_2(M_X)=\rk_2(M'_X)$, then the $x$-row
  of $M_X$, which is a row of $(C\;\;D)$, is a linear combination of the
  other rows of $M_X$. This implies that the $x$-row of $D$ is a
  linear combination of the remaining rows of $B$ and $D$. Thus the
  $x$-row of $M_{X\cup Y}$, which is a row of $D$, is a linear
  combination of the remaining rows of $M_{X\cup Y}$. Hence $rk_2(M_{X\cup
    Y})=\rk_2(M'_{X\cup Y})$.

  Similarly, arguing with columns instead of rows, we see that
  \begin{equation}\label{eq:submod2d}
  \rk_2(M_Y)=\rk_2(M'_Y)\quad\implies\quad \rk_2(M_{X\cap
    Y})=\rk_2(M'_{X\cap Y}).
  \end{equation}
  Clearly, \eqref{eq:submod2a}--\eqref{eq:submod2d} imply
  \eqref{eq:submod2}. This completes the proof that $\rho_G$ is a
  connectivity function.
    \uend
\end{exa}

  Let us briefly discuss how the cut-rank function $\rho_G$ relates to
  the edge-connectivity function $\nu_G$ and the matching connectivity
  function $\mu_G$, which are also defined on the vertex set of a
  graph $G$. As the number of $1$ entries of a matrix over $\mathbb
  F_2$ is always an upper bound for its row rank, we have
  $\rho_G(X)\le\nu_G(X)$ for all $X\subseteq V(G)$. We also have
  $\rho_G(X)\le\mu_G(X)$, because in a matrix of row rank $k$ we can
  always find $k$ distinct rows $i_1,\ldots,i_k$ and $k$ distinct
  columns $j_1,\ldots,j_k$ such that for all $p$ the entry in row
  $i_p$ and columns $j_p$ is $1$. (This can be proved by induction on $k$.)
  In the matrix $M(X,\bar X)$, the edges corresponding to these entries
  form an $(X,\bar X)$-matching of order $k$. As we trivially have
  $\mu_G(X)\le\nu_G(X)$ for all $X$, altogether we get
  \[
  \rho_G(X)\le\mu_G(X)\le\nu_G(X).
  \]
  In general, all the equations can be strict, and the in fact the gaps
  can be arbitrarily large.
 For example, for the complete graph $K_n$ we
  have 
  \begin{align*}
    \rho_{K_n}(X)&=1,\\
    \mu_{K_n}(X)&=\min\{|X|,|\bar X|\},\\
    \nu_{K_n}(X)&=|X|\cdot|\bar X|.
  \end{align*}
  for all $X\subseteq V(K_n)$. 

It is also worth noting that
  $\kappa_G,\nu_G,\mu_G$ are all \emph{subgraph monotone}. For
  example, for $\nu_G$ this means that is, for all
  $G'\subseteq G$ and $X\subseteq V(G')$ we have $\nu_{G'}(X)\le
  \nu_G(X)$. The cut-rank function $\rho_G$ is not subgraph
  monotone. It is only \emph{induced-subgraph monotone}.

  \medskip
Let us now turn to examples from a different domain: vector spaces and matroids.

\begin{exa}[Vector Spaces]\label{exa:vsc}
  Let $\mathbb V$ be a vector space. For set $X\subseteq\mathbb V$, by
  $\angles X$ we denote the subspace of $\mathbb V$ generated by $X$,
  and for a subspace $\mathbb W\subseteq\mathbb V$, by $\dim(\mathbb
  W)$ we denote its dimension.

  Now let ${\U}$ be a finite subset of $\mathbb V$. We define a set
  function $\kappa_{\mathbb V,{\U}}$ on ${\U}$ by 
  \[
  \kappa_{\mathbb V,{\U}}:=\dim(\angles X\cap\angles{\bar X}).
  \]
  We leave it to the reader to verify that $\lambda_A$ is a connectivity
  function. In view of the following example, observe that 
  \[
  \kappa_{\mathbb V,{\U}}(X)=\dim(\angles X)+\dim(\angles{\bar
    X})-\dim(\angles {\U}).\uende
  \]
\end{exa}

For our next example, which is a direct generalisation of the previous
one, we review a few basics of
matroid theory. A \emph{matroid} is a pair $\CM=({\U},\CI)$, where ${\U}$ is a
finite set and $\CI\subseteq 2^{\U}$ a nonempty set that is closed under
taking subsets and has the following \emph{augmentation property}: if
$I,J\in\CI$ such that $|I|<|J|$, then there is an $u\in J$ such that
$I\cup\{u\}\in\CI$. The elements of $\CI$ are called \emph{independent
  sets}. We define a set function $\rho_{\CM}$ on ${\U}$ by letting
$\rho_{\CM}(X)$ be the maximum cardinality of an independent set
$I\subseteq X$. It is easy to see that $\rho_{\CM}$ is normalised, monotone, 
submodular, and univalent. We call $\rho_{\CM}$ the
\emph{rank function} of the matroid $\CM$. The \emph{rank} of the
matroid $\CM$ is $\rho_{\CM}({\U})$, that is, the maximum cardinality of
an independent set.

It can be shown that if $\rho:2^{\U}\to\NN$ is a normalised, monotone,
submodular, and univalent, then there is a matroid
$\CM_\rho$ on ${\U}$ with rank function $\rho$. The independent sets of
this matroid are the sets $I\subseteq {\U}$ with $\rho(I)=|I|$.

\begin{exa}[Matroid Connectivity]\label{exa:conn6}
  Let $\CM=({\U},\CI)$ be a matroid. Then the set function
  $\kappa_{\CM}$ on ${\U}$ defined by
  \[
  \kappa_{\CM}(X):=\rho_{\CM}(X)+\rho_{\CM}(\bar X)-\rho_{\CM}({\U})
  \]
  is a connectivity function, known as the \emph{connectivity function
    of the matroid $\CM$}. It is obviously symmetric, and we have
  $\kappa_\CM(\emptyset)=\rho_\CM(\emptyset)+\rho_{\CM}({\U})-\rho_{\CM}({\U})=0$. The
  submodularity follows directly from the submodularity of
  $\rho_{\CM}$:
  \begin{align*}
    \kappa_{\CM}(X)+\kappa_{\CM}(Y)=\,&\rho_{\CM}(X)+\rho_{\CM}(Y)\\
    &+\rho_{\CM}(\bar X)+\rho_{\CM}(\bar Y)-2\rho_{\CM}({\U})\\
    \ge\,&\rho_{\CM}(X\cap Y)+\rho_{\CM}(X\cup Y)\\
    &+\rho_{\CM}(\bar X\cap\bar Y)+\rho_{\CM}(\bar X\cup \bar Y)-2\rho_{\CM}({\U})\\
    =\,&\rho_{\CM}(X\cap Y)+\rho_{\CM}(\bar{X\cap Y})-\rho_{\CM}({\U})\\
    &+\rho_{\CM}(X\cup Y)+\rho_{\CM}(\bar{X\cup Y})-\rho_{\CM}({\U})\\
    =\,&\kappa_{\CM}(X\cap Y)+\kappa_{\CM}(X\cup Y).
  \end{align*}
  The connectivity functions of the vector space
  Example~\ref{exa:vsc} is a special case; in fact, it is precisely
  the case of a \emph{representable matroid}. Let $\mathbb V$ be a
  vector space and ${\U}\subseteq \mathbb V$ a finite subset. We define a
  matroid $\CM_{\U}$ on ${\U}$ by letting $I\subseteq {\U}$ be independent in
  $\CM_{\U}$ if $I$ is a linearly independent set of vectors. It is easy
  to verify that this is indeed a matroid and that its rank function
  of is $\rho_{\CM_{\U}}$ is defined by $\rho_{\CM_{\U}}(X)=\dim(\angles X)$. 
  \uend
\end{exa}

\begin{exa}[Integer Polymatroids]
  Observing that we never used the univalence of the
  rank function of a matroid in the previous example, we give a
  further generalisation by simply dropping the condition that the
  rank function be univalent. An
  \emph{integer polymatroid} is a normalised monotone and submodular set function
  $\pi$ on a finite set ${\U}$. With
  each such $\pi$ we can associate a connectivity function $\kappa_\pi$
  on ${\U}$ defined by 
    \[
  \kappa_\pi(X):=\pi(X)+\pi(\bar X)-\pi({\U}).
  \]
  \uend
\end{exa}

In fact, the previous example is about as general as it gets. Jowett,
Mo, and Whittle~\cite{jowmowhi15} have observed that up to a
factor of 2 every connectivity function is the connectivity function
of a polymatroid. To see this, let $\kappa$ be a connectivity function
on ${\U}$. We define a set function $\pi$ on ${\U}$ by
\[
\pi(X):=\kappa(X)+\sum_{x\in X}\kappa(\{x\}).
\]
It is easy to see that $\pi$ is an integer
polymatroid. Furthermore, the connectivity function $\kappa_{\pi}$
associated with $\pi$ is equal to $2\kappa$. Indeed,
\begin{align*}
  \kappa_\pi(X)&=\kappa(X)+\sum_{x\in X}\kappa(\{x\})
                  +\kappa(\bar X)+\sum_{x\in \bar X}\kappa(\{x\})
                  -\kappa({\U})-\sum_{x\in {\U}}\kappa(\{x\})\\
                &=\kappa(X)+\kappa(\bar X)-\kappa({\U})=2\kappa(X).
\end{align*}
Another way of phrasing this result is that every connectivity
function is the connectivity function of a half-integral
polymatroid. Jowett et al.~\cite{jowmowhi15} actually proved a
stronger result characterising the polymatroids associated with
connectivity functions this way as \emph{self-dual} (see
\cite{jowmowhi15} for details).

\medskip
We close this section with a ``non-example''. Somewhat surprisingly, a
natural extension of the matching connectivity function to hypergraphs
is not submodular.

  \begin{figure}
    \centering
    \input{htw.tex}
    \caption{The hypergraph $H_0$ of Example~\ref{exa:htw} and its dual
    $\tilde H_0$}
    \label{fig:hw}
  \end{figure}
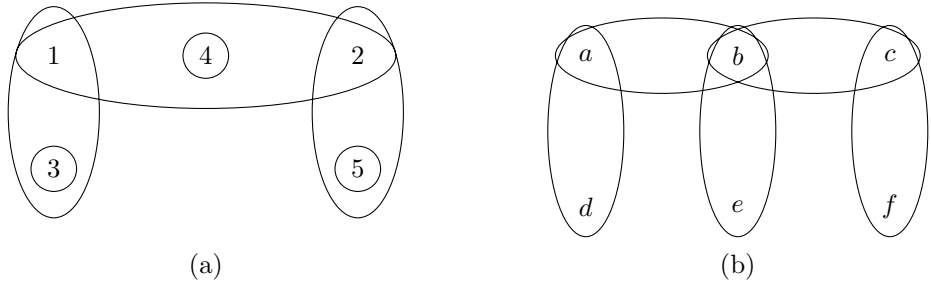

\begin{exa}\label{exa:htw}
  Let $H$ be a hypergraph. For subsets $X,Y\subseteq V(H)$, we define
  $E(X,Y)$ to be the set of all $e\in E(H)$ such that $e\cap
  X\neq\emptyset$ and $e\cap Y\neq\emptyset$. A \emph{vertex cover}
  for a set $F\subseteq E(H)$ is a set $S$ of vertices such that
  $S\cap e\neq\emptyset$ for all $e\in F$. 

  We define $\mu_H: 2^{V(H)}\to\ZZ$ by letting $\mu_H(X)$ be the
  minimum cardinality of a vertex cover of $E(X,\bar X)$. This
  function is obviously normalised and symmetric, but it is \emph{not}
  always submodular.

  As an example, consider the hypergraph $H_0$ with
  \begin{align*}
    V(H_0)&:=\{1,2,3,4,5\},\\
    E(H_0)&:=\big\{\underbrace{\{1,3\}}_{=:a},\underbrace{\{1,2,4\}}_{=:b},\underbrace{\{2,5\}}_{=:c},\underbrace{\{3\}}_{=:d},\underbrace{\{4\}}_{=:e},\underbrace{\{5\}}_{=:f}\big\}
  \end{align*}
  (see Figure~\ref{fig:hw}(a)). Let
  $
  X:=\{1,4\}
  $
  and $Y:=\{2,4\}$. Then $\mu_{H_0}(X)=1$, because $E(X,\bar X)=\{a,b\}$ and
  vertex $1$ covers the edges $a$ and $b$. Similarly, $\mu_{H_0}(Y)=1$ and
  $\mu_{H_0}(X\cap Y)=1$ and $\mu_{H_0}(X\cup Y)=2$. This contradicts
  submodularity. 

  Recall (from Example~\ref{exa:hc}) the definition of the dual
  $\tilde H$ of a hypergraph $H$. We define the ``dual'' set function $\tilde\mu_H:2^{E(H)}\to\ZZ$
  by 
  \[
  \tilde\mu_H(Y):=\mu_{\tilde H}(Y).
  \]
  For a set $Y$ of hyperedges, it measures the minimum cardinality of
  an edge cover of the boundary $\partial(Y)$.
  
  As $\mu_H$, the dual function $\tilde\mu_H$ is normalised and
  symmetric, but not submodular, not even on simple graphs. The dual hypergraph $\tilde H_0$,
  shown in Figure~\ref{fig:hw}(b), witnesses the latter. In this dual
  form, the example is due to \cite{adlgotgro07}.

  Remarkably, the function $\tilde\mu_H$ has been used in \cite{adlgotgro07} 
  to define the \emph{hyper branch width} of a
  hypergraph in the same way as we shall define the branch width of
  connectivity function in the next section.  Hyper branch width is
  a constant factor approximation of the more familiar hyper tree width
  (see, for example, \cite{gotleosca02,gotleosca03,ggmss05}). However,
  as $\tilde\mu_H$ is not a connectivity function, the general theory we
  shall develop in the following sections does not apply to it, and
  the nice results that nevertheless hold for $\tilde\mu_H$ have to be
  proved in an ad-hoc fashion.

  It is an interesting open
  question if there is a connectivity function whose branch width also
  approximates hyper tree width to a constant factor.
  \uend  
\end{exa}

\subsection{Connectivity Functions on Separation Systems}
It is sometimes useful to define connectivity systems in an even more
abstract setting of \emph{separation systems}, which are lattices with
``complementation''. The following example may serve as motivation.

\begin{exa}\label{exa:vertex-separation}
  When thinking about vertex connectivity, it is sometimes more convenient to define
  separations of a graph on the vertex set rather than on the edge set (as we did in
  Example~\ref{exa:conn1}). 

Let $G$ be a graph. A \emph{vertex separation} of $G$ is a pair
  $(Y,Z)$ of subsets of $V(G)$ such that $Y\cup Z=V(G)$ and there
  is no edge from $Y\setminus Z$ to $Z\setminus Y$. Alternatively, we
  may view a vertex separation as a partition
  $(Y',S,Z')$ of $V(G)$ (with possibly empty parts) such that there is
  no edge from $Y'$ to $Z'$. The \emph{order} of a
  vertex separation $(Y,Z)$ is $\kappa(Y,Z):=|Y\cap Z|$. Intuitively,
  $\kappa$ is a connectivity function on the set of all vertex
  separations, which is closely related to the
  vertex-connectivity function $\kappa_G$ of
  Example~\ref{exa:conn1}. But $\kappa$ does not fit into our
  framework, because it is not defined on the power set on some set
  ${\U}$, but on the set $\CS$ of all vertex separations of $G$.

  Note, however, that $\CS$ has a natural lattice structure. We define
  the \emph{join} and \emph{meet} of two separations $(Y,Z)$ and
  $(Y',Z')$ to be the vertex separation
  $(Y,Z)\vee(Y',Z'):=(Y\cap Y',Z\cup Z')$ and
  $(Y,Z)\wedge(Y',Z'):=(Y\cup Y',Z\cap Z')$. It is easy to check that
  $\kappa$ is submodular with respect to these lattice
  operations. Furthermore, $\kappa$ is symmetric with respect to the
  natural \emph{complementation} $\bar{(Y,Z)}:=(Z,Y)$.
  \uend
\end{exa}

Recall that a lattice is a set $\CL$ with two binary operations
\emph{join} $\vee$ and
\emph{meet} $\wedge$ satisfying the commutative laws, the associative
laws, the idempotent laws ($x\vee x=x\wedge x=x$), and the absorption
laws ($x\vee(x\wedge y)=x$ and $x\wedge(x\vee y)=x$). Associated with
the lattice is a partial order defined by $x\le y:\Leftrightarrow
x=x\wedge y$. The join and meet operation correspond to supremum and
infimum of two elements with respect to this partial order. In fact,
lattices are in one-to-one correspondence to partial orders in which
any two elements have a unique supremum and infimum.

A \emph{separation system} \cite{diehunlem16} is is a tuple
$(\CS,\vee,\wedge,\comp)$, where $(\CS,\vee,\wedge)$ is a lattice and
$\comp:\CS\to\CS$ is an order-reversing involution, that is, $x\le
y\implies \bar y\le\bar x$ and $\bar{\bar x}=x$. We call $\comp$
\emph{complementation}. We call a function $\phi:\CS\to\ZZ$
\emph{symmetric}, \emph{monotone}, \emph{submodular},
\emph{normalised}, \emph{nontrivial} by modifying the respective standard
definitions for set function $\phi$ in the obvious way, replacing
$\cup$ by $\vee$, $\cap$ by $\wedge$, $\subseteq$ by $\le$, and
$x\setminus y$ by $x\wedge\bar y$.

A \emph{connectivity function} on a separation system
$(\CS,\vee,\wedge,\comp)$ is a function $\kappa:\CS\to\ZZ$ that is
normalised, symmetric, and submodular. Most of the theory of
connectivity functions we shall develop in the following sections can
be generalised to connectivity functions on separation systems without
much effort.

\begin{exa}\label{exa:vertex-separation2}
  The set of vertex separations of a graph together with the meet, join,
  and complementation operations defined in
  Example~\ref{exa:vertex-separation} is a separation systems, and the
  function $\kappa$ defined by $\kappa(X,Y):=|X\cap Y|$ is a
  connectivity function on this separation system.
  \uend
\end{exa}

\begin{exa}
  For every set ${\U}$, the \emph{natural separation system on $2^{\U}$} is
  the system we obtain by interpreting meet as union, join as
  intersection, and complementation as complementation in ${\U}$. Of
  course $\kappa$ is a connectivity function on ${\U}$ in the usual sense
  if and only it is a connectivity function on this separation system.
  \uend
\end{exa}

\begin{exa}\label{exa:contraction}
  Let $\kappa$ be a connectivity function on a set ${\U}$, and let
  $\CA\subseteq 2^{\U}$. 
  Let
  \[
  \CS:=\{X\subseteq {\U}\mid X\cap A=\emptyset\text{
    or }A\subseteq X\text{ for all }A\in\CA\}
  \]
  and note that $\CS$ is closed under union, intersection, and
  complementation. Hence if we let $\vee,\wedge,\comp$ be the
  restrictions of union, intersection, and complementation to $\CS$,
  then we obtain a separation system $(\CS,\vee,\wedge,\comp)$. We
  may think of this separation system as the system we obtain from the
  natural separation system on $2^{\U}$ if we
  declare the sets in $\CA$ to be \emph{inseparable}, or \emph{atoms}.

  We usually assume the atoms to be mutually disjoint, because if
  $A,A'$ are atoms with a nonempty intersection, then their union
  $A\cup A'$ becomes inseparable as well, and replacing $A,A'$ by $A\cup
  A'$ in $\CA$ yields the same separation system. If the atoms in
  $\CA$ are
  mutually disjoint, we can think of $(\CS,\vee,\wedge,\comp)$ as the
  connectivity system we obtain by \emph{contracting} the atoms to
  single points. For every set $A\in\CA$ we introduce a fresh element
  $a$, and we let
  \[
  {\U}\contract_\CA:=\Big({\U}\setminus\bigcup_{A\in\CA}A\Big)\cup\{a\mid
  A\in\CA\}.
  \]
  Then our separation system $(\CS,\vee,\wedge,\comp)$ is isomorphic to the
  natural separation system on $2^{{\U}\contract_\CA}$. The image of the
  connectivity function $\kappa$, restricted to $\CS$, under the
  natural isomorphism is the connectivity function
  $\kappa\contract_\CA:2^{{\U}\contract_\CA}\to\ZZ$ defined by 
  \[
  \kappa\contract_\CA(X)=\kappa(X\expand_\CA),
  \]
  where $X\expand_\CA:={\U}\setminus\bigcup_{\substack{A\in\CA\\a\not\in
    X}}A$.
\uend
\end{exa}

\section{Branch Decompositions and Branch Width}
\label{sec:dec}

In this section we define decompositions of connectivity
functions. Our decompositions are based on the branch decompositions
introduced in \cite{gm10}, but it will be useful to introduce several
generalisations and variants of branch decompositions as well.

We start with some  terminology and notation. As usual, a \emph{tree} is a
connected acyclic graph.
A \emph{directed tree} is obtained from a tree by orienting all edges
away from a distinguished \emph{root}. The \emph{leaves} of a directed
tree $T$ are the nodes of out-degree $0$; we denote the set of all
leaves by $L(T)$. We call all  non-leaf nodes \emph{internal
  nodes}. In a directed tree, we can
speak of the \emph{children} of an internal node and the \emph{parent}
of a non-root node.  We also speak of \emph{descendants}, that is,
children, children of the children, et cetera, and \emph{ancestors}
of a node. A directed tree is \emph{binary} if every
internal node has exactly two children. 

We typically denote tree nodes by $s,t,u$ and elements
of the universe ${\U}$ of a connectivity function by $x,y,z$.

\subsection{Directed Decompositions}

Directed decompositions of connectivity systems are defined in a most
straightforward way: we split the universe into two disjoint parts,
possibly split each of these parts again, and so on. This gives us a
decomposition naturally structured as a binary tree, and with the pieces
of the decomposition labelling the leaves of the tree. We call the
decomposition \emph{complete} if all these pieces (at the leaves) are
just single elements. The \emph{width} of the decomposition is the
maximum order of the separations appearing at any stage of the
decomposition (with respect to the connectivity function we
decompose). For technical reasons, it will be convenient to also introduce
a more relaxed form of decomposition, which we call
\emph{pre-decomposition}. 
Later, we will also introduce an
undirected version of our decompositions, which will be based on cubic
(undirected) trees rather than binary directed trees.

\begin{defn}
  Let ${\U}$ be a finite set.
  \begin{enumerate}
  \item A \emph{directed pre-decomposition} of ${\U}$ is a pair $(T,\gamma)$
    consisting of a binary directed tree $T$ and a mapping
    $\gamma:V(T)\to 2^{\U}$ such that $\gamma(r)={\U}$ for the root $r$ of $T$ and
    $\gamma(t)\subseteq\gamma(u_1)\cup\gamma(u_2)$ for all internal nodes $t$
    with children $u_1,u_2$.
  \item A directed pre-decomposition $(T,\gamma)$ is \emph{complete} if
    $|\gamma(t)|=1$ for all leaves $t\in L(T)$.
  \item A directed pre-decomposition $(T,\gamma)$ is \emph{exact} at a
    node $t\in V(T)$ with children $u_1,u_2$ if
    $\gamma(t)=\gamma(u_1)\cup\gamma(u_2)$ and
    $\gamma(u_1)\cap\gamma(u_2)=\emptyset$.
  \item A \emph{directed decomposition} is a directed
    pre-decomposition that is exact at all internal nodes.
  \item A \emph{directed branch decomposition} is a complete directed
    decomposition.
  \end{enumerate}
  If $\kappa$ is a connectivity function on ${\U}$, then a
\emph{(complete) directed (branch, pre-)decomposition} of $\kappa$ is a
(complete) directed (branch, pre-)decomposition of ${\U}$.
\uend
\end{defn}

Let $(T,\gamma)$ be a directed pre-decomposition.
We call the sets $\gamma(t)$ the \emph{cones} 
and the cones $\gamma(t)$ for the leaves $t\in L(T)$ the \emph{atoms}
of the pre-decompositions. We denote the set of all atoms of
$(T,\gamma)$ by $\At(T,\gamma)$. In view of
Example~\ref{exa:contraction}, it is worth noting that a
decomposition $(T,\gamma)$ of $\kappa$ may be viewed as a branch
decomposition of $\kappa\contract_{\At(T,\gamma)}$.

Observe that if $(T,\gamma)$ is a directed decomposition (not just
a pre-decomposition), then the restriction
of $\gamma$ to the leaves, determines $\gamma$: for an internal node $s$, 
the cone $\gamma(s)$ is the union of the atoms $\gamma(t)$
for all leaves $t$ that are
descendants of $s$ in $T$. In a complete directed
pre-decomposition, the mapping $\gamma$ specifies a mapping
from leaves of $T$ onto ${\U}$: each leaf $t$ is mapped to the unique element
of $\gamma(t)$. In a directed branch decomposition, this mapping is bijective.

We call the cones $\gamma(t)$ and their complements $\bar{\gamma(t)}$
for non-root nodes $t\in V(T)\setminus\{r\}$ the \emph{separations} of
the decomposition and denote the set of all separations of
$(T,\gamma)$ by $\Sep(T,\gamma)$.

Let us call a directed pre-decomposition \emph{proper} if $|V(T)|>1$ and
\emph{non-degenerate} if all atoms are nonempty.  Observe that if
$(T,\gamma)$ is a proper and nondegenerate directed decomposition then
the atoms are the inclusionwise minimal separations. Thus
$\At(T,\gamma)$ is determined by $\Sep(T,\gamma)$.

\begin{defn}
  Let $\kappa$ be a connectivity function on ${\U}$.
  \begin{enumerate}
  \item The \emph{width} of a directed pre-decomposition $(T,\gamma)$ of
    $\kappa$ is
    \[
    \width(T,\gamma):=\max\big\{\kappa(\gamma(t))\bigmid t\in V(T)\big\}.
    \]
  \item The \emph{branch width} of $\kappa$ is
    \[
    \bw(\kappa):=\min\big\{\width(T,\gamma)\bigmid(T,\gamma)\text{ directed
      branch decomposition of }\kappa\big\}.\uende
    \]
  \end{enumerate}
\end{defn}

Note that the unique connectivity function $\kappa_\emptyset$ on the empty
universe has no branch decomposition. Nevertheless, its is convenient
to define the branch width of $\kappa_\emptyset$ to be $0$.

\begin{exa}\label{exa:conn-bd1}
  Let
    \[
  {\U}=\left\{
    \begin{pmatrix}
      1\\0\\0\\0
    \end{pmatrix},
    \begin{pmatrix}
      0\\1\\0\\0
    \end{pmatrix},
    \begin{pmatrix}
      0\\0\\1\\0
    \end{pmatrix},
    \begin{pmatrix}
      0\\0\\0\\1
    \end{pmatrix},
    \begin{pmatrix}
      1\\1\\0\\0
    \end{pmatrix},
    \begin{pmatrix}
      1\\1\\1\\1
    \end{pmatrix}    
  \right\}
  \subseteq\mathbb R^4,
  \]
  and let $\kappa:=\kappa_{\RR^4,{\U}}$ be the connectivity function
  define on ${\U}$ as in Example~\ref{exa:vsc}.
  
  Figure~\ref{fig:dbranchdec1} shows two directed branch decompositions
  of $\kappa$. As discussed above, to specify the mapping $\gamma$ in a
  directed decomposition $(T,\gamma)$, we only need to specify the
  values $\gamma(t)$ for the leaves $t$. In the Figure, we do this by
  displaying the unique element of $\gamma(t)$ at every leaf $t$.

  The width of the first decomposition is
  $1$, and the width of the second decomposition is $2$. To verify
  this, we have to compute the values $\kappa(\gamma(s))$ for all nodes
  $s$. For example, for the node $t$ in the first decomposition
  (Figure~\ref{fig:dbranchdec1}(a)) we have
  \[
  \gamma(t)=\left\{\begin{pmatrix}
      1\\0\\0\\0
    \end{pmatrix},
    \begin{pmatrix}
      0\\1\\0\\0
    \end{pmatrix},
    \begin{pmatrix}
      1\\1\\0\\0
    \end{pmatrix}
  \right\}.
  \]
  Thus
  \[
  \angles{\gamma(t)}
  =
  \angles{\begin{pmatrix}
      1\\0\\0\\0
    \end{pmatrix},
    \begin{pmatrix}
      1\\1\\0\\0
    \end{pmatrix}
  }.
  \]
  Furthermore, 
  \[
  \angles{\bar{\gamma(t)}}=\angles{\begin{pmatrix}
      1\\1\\1\\1
    \end{pmatrix},
    \begin{pmatrix}
      0\\0\\1\\0
    \end{pmatrix},
    \begin{pmatrix}
      0\\0\\0\\1
    \end{pmatrix}
  }=  \angles{\begin{pmatrix}
      1\\1\\0\\0
    \end{pmatrix},
    \begin{pmatrix}
      0\\0\\1\\0
    \end{pmatrix},
    \begin{pmatrix}
      0\\0\\0\\1
    \end{pmatrix}
  }.
  \]
  Thus
  \[
  \Big\langle\gamma(t)\Big\rangle\cap \angles{\bar{\gamma(t)}}=\angles{\begin{pmatrix}
      1\\1\\0\\0
    \end{pmatrix}}.
  \]
  and hence
  \[
  \kappa(\gamma(t))=\dim\Big(\Big\langle\gamma(t)\Big\rangle\cap
  \angles{\bar{\gamma(t)}}\Big)=1.
  \]
  Finally, observe that $\bw(\kappa)=1$. The first of our
  decompositions witnesses $\bw(\kappa)\le 1$. The following
  Exercise implies that $\bw(\kappa)\ge 1$.
  \uend
\end{exa}

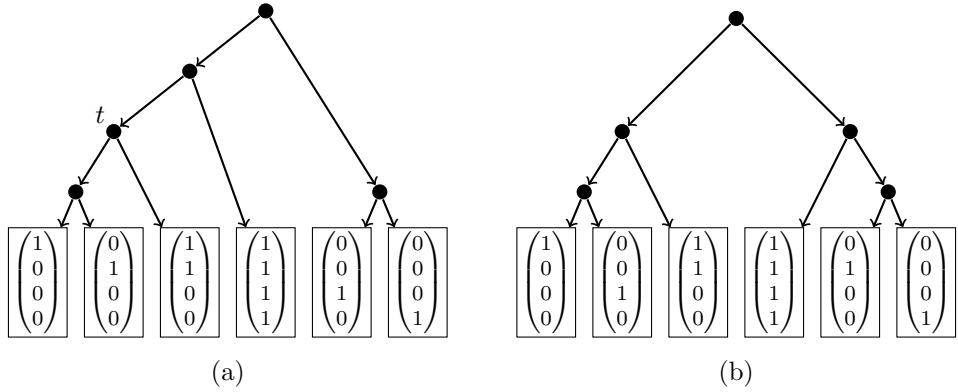
\begin{figure}
  \centering
  \input{dbranchdec1}
  \caption{Two directed branch decompositions of the connectivity function of a
    set of vectors}
  \label{fig:dbranchdec1}
\end{figure}

\begin{exe}
  Let ${\U}$ be a finite subset of some vector space
  $\mathbb V$. Prove that $\bw(\kappa_{\mathbb V,{\U}})=0$ if and only if the
  vectors in ${\U}$ are linearly independent.
  \uend
\end{exe}

\begin{figure}
  \centering
  \input{rankdec}
  \caption{A directed branch decomposition of the vertex set of a
    graph}
  \label{fig:rankdec}
\end{figure}
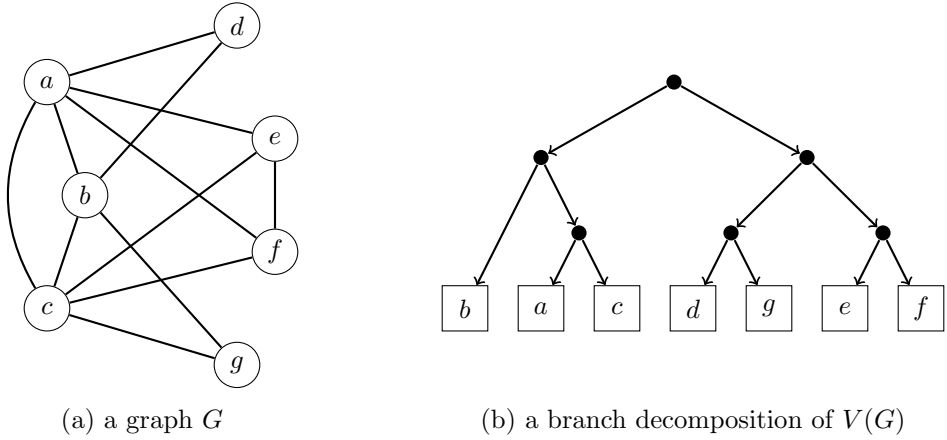

\begin{exa}
  Let $G$ be the graph shown in
  Figure~\ref{fig:rankdec}(a). Figure~\ref{fig:rankdec}(b) shows a
  branch decomposition of $V(G)$.
  The width of this decomposition is
  \begin{itemize}
  \item $8$ if it is viewed as a decomposition of the
    edge-connectivity function $\nu_G$;
  \item $3$ if it is viewed as a decomposition of the
    matching connectivity function $\mu_G$;
 \item $2$ if it is viewed as a decomposition of the
    cut-rank function $\rho_G$.\uend
   \end{itemize}
 \end{exa}

We now prove a first nontrivial result about our decompositions,
showing that every pre-decomposition can be turned into a
decomposition of the the same width. The proof is a nice
application of submodularity.

\begin{lem}[Exactness Lemma \cite{gm10}]\label{lem:exact}
  Let $\kappa$ be a connectivity function on ${\U}$ and $(T,\gamma)$ a
  directed pre-decomposition of $\kappa$. Then there is a function
  $\gamma':V(T)\to 2^{\U}$ such that $(T,\gamma')$ is a directed
  decomposition of $\kappa$ satisfying
  $\kappa(\gamma'(t))\le \kappa(\gamma(t))$ for all nodes $t\in V(T)$ and
  $\gamma'(t)\subseteq\gamma(t)$ for all leaves $t\in L(T)$.
\end{lem}

\begin{proof}
  We will
  iteratively construct a sequence $\gamma_1,\ldots,\gamma_m$ of mappings
  from $V(T)$ to $2^{\U}$ such that
  $(T,\gamma_1),\ldots,(T,\gamma_m)$ are pre-decompositions satisfying the following invariants for all $i\in[m-1]$ and nodes $t\in V(T)$:
  \begin{eroman}
  \item $\kappa(\gamma_{i+1}(t))\le \kappa(\gamma_{i}(t))$;
  \item either $\gamma_{i+1}(t)\subseteq\gamma_{i}(t)$ or
    $\kappa(\gamma_{i+1}(t))< \kappa(\gamma_{i}(t))$;
  \item if $t\in L(T)$ then
    $\gamma_{i+1}(t)\subseteq\gamma_{i}(t)$.
  \end{eroman}
  Furthermore, $(T,\gamma_m)$ will be a decomposition, that is, exact at
  all internal nodes. Clearly, this will prove the lemma.
  
  We let $\gamma_1:=\gamma$.  In the inductive step, we assume that we have
  defined $\gamma_i$. If $(T,\gamma_i)$ is exact at all internal nodes, we
  let $m:=i$ and stop the construction. Otherwise, we pick an
  arbitrary node $s\in V(T)$ with children $t_1,t_2$ such that
  $(T,\gamma_i)$ is not exact at $s$. Then either
  $\gamma_i(s)\subset\gamma_i(t_1)\cup\gamma_i(t_2)$ or
  $\gamma_i(t_1)\cap\gamma_i(t_2)\neq\emptyset$ (possibly both). We let
  $X:=\gamma_i(s)$ and $Y_p:=\gamma_i(t_p)$ for $p=1,2$.

  In each of
  the following cases, we only modify $\gamma_i$ at the nodes
  $s,t_1,t_2$ and let $\gamma_{i+1}(u):=\gamma_i(u)$ for all
  $u\in V(T)\setminus\{s,t_1,t_2\}$.
  \begin{cs}
    \case1
    $X\subset Y_1\cup Y_2$.
    \begin{cs}
      \case{1a}
      $\kappa(X\cap Y_p)\le\kappa(Y_p)$ for $p=1,2$.\\
      We let $\gamma_{i+1}(s):=\gamma_i(s)$ and $\gamma_{i+1}(t_p):=X\cap Y_p$
      for $p=1,2$.

      Note that in this case we have
      $\kappa(\gamma_{i+1}(u))\le \kappa(\gamma_{i}(u))$ and
      $\gamma_{i+1}(u)\subseteq\gamma_{i}(u)$ for all nodes $u$ and either
      $\gamma_{i+1}(t_1)\subset\gamma_{i}(t_1)$ or
      $\gamma_{i+1}(t_2)\subset\gamma_{i}(t_2)$.

      \case{1b}
      $\kappa(X\cap Y_p)>\kappa(Y_p)$ for some $p\in\{1,2\}$.\\
      For $p=1,2$, we let $\gamma_{i+1}(t_p):=Y_p$.

      By submodularity we have $\kappa(X\cup Y_p)<\kappa(X)$ for some
      $p\in\{1,2\}$.  If $\kappa(X\cup Y_1)<\kappa(X)$ we let
      $\gamma_{i+1}(s):=X\cup Y_1$, and otherwise we let
      $\gamma_{i+1}(s):=X\cup Y_2$.

      Note that in this case we have
      $\kappa(\gamma_{i+1}(u))\le \kappa(\gamma_{i}(u))$ for all nodes $u$
      and $\kappa( \gamma_{i+1}(s))<\kappa( \gamma_{i}(s))$ and
      $\gamma_{i+1}(u)=\gamma_{i}(u)$ for all nodes $u\neq s$.

      Also note that invariant (iii) is preserved, because $s$ is not
      a leaf of the tree.
    \end{cs}

    \case2
    $X= Y_1\cup Y_2$ and $Y_1\cap Y_2\neq\emptyset$.\\
    We let $\gamma_{i+1}(s):=\gamma_i(s)$.

    By posimodularity, either $\kappa(Y_1\setminus Y_2)\le\kappa(Y_1)$
    or $\kappa(Y_2\setminus Y_1)\le\kappa(Y_2)$. If
    $\kappa(Y_1\setminus Y_2)\le\kappa(Y_1)$, we let
    $\gamma_{i+1}(t_1):=Y_1\setminus Y_2$ and
    $\gamma_{i+1}(t_2):=Y_2$. Otherwise, we let $\gamma_{i+1}(t_1):=Y_1$ and
    $\gamma_{i+1}(t_2):=Y_2\setminus Y_1$.
  
    Note that in this case we have
    $\kappa(\gamma_{i+1}(u))\le \kappa(\gamma_{i}(u))$ and
    $\gamma_{i+1}(u)\subseteq\gamma_{i}(u)$ for all nodes $u$ and either
    $\gamma_{i+1}(t_1)\subset\gamma_{i}(t_1)$ or
    $\gamma_{i+1}(t_2)\subset\gamma_{i}(t_2)$.
  \end{cs}
  This completes the description of the construction. To see that it
  terminates, we say that the \emph{total weight} of $\gamma_i$ is
  $\sum_{t\in V(T)}\kappa(\gamma_i(t))$ and the \emph{total size} of
  $\gamma_i$ is $\sum_{t\in V(T)}|\gamma_i(t)|$.  Now observe that in each
  step of the construction either the total weight decreases or the
  total weight stays the same and the total size decreases. This
  proves termination.

  To see that $(T,\gamma_i)$ is indeed a pre-decomposition, observe
  first that $\gamma_i(r)={\U}$ for the root $r$ of $T$, because the root
  can only occur as the parent node $s$ in the construction above, and
  the set at the parent node either stays the same (in Cases~1a and 2)
  or increases (in Case~1b). Moreover, it is easy to check that for
  all nodes $s'$ with children $t_1',t_2'$ we have
  $\gamma_i(s')\subseteq\gamma_i(t_1')\cup\gamma_i(t_2')$. This follows
  immediately from the construction if $s=s'$. If $s'$ is the parent
  of $s=t_i'$, it follows because the set at $s$ can only increase. If
  $s'=t_i$, it follows because the set at $t_i$ can only
  decrease. Otherwise, all the sets at $s',t_1',t_2'$ remain
  unchanged. Note that the invariant (iii) is preserved, because
  leaves can only occur as the child nodes $t_i$ in the construction
  above, and the sets at the child nodes either decrease (in Cases~1a
  and 2) or stay the same (in Case~1b).
\end{proof}

The Exactness Lemma may yield a degenerate decomposition where some atoms are
empty. While not ruled out by the definitions, empty  atoms are not
making much sense in a decomposition. The following lemma shows that
we can easily get rid of them.

\begin{lem}\label{lem:empty-leaves}
  Suppose that ${\U}\neq\emptyset$.
  Let $(T,\gamma)$ be a directed decomposition of a set ${\U}$. Then there is a
  directed decomposition $(T',\gamma')$ of ${\U}$ such that 
  \begin{eroman}
  \item $T'$ is a subtree of $T$ (with the same root),
  \item $\gamma'(t)=\gamma(t)\neq\emptyset$ for all $t\in V(T')$ and
    $\gamma(t)=\emptyset$ for all $t\in V(T)\setminus V(T')$.
  \end{eroman}
\end{lem}

\begin{proof}
  We simply delete all nodes $t$ with $\gamma(t)=\emptyset$ and all
  their siblings. This works because if $t,t'$ are children of a node
  $s$ and $\gamma(t)=\emptyset$ then $\gamma(s)=\gamma(t')$ by the exactness of
  the decomposition at $s$.
\end{proof}

\subsection{Undirected Decompositions}

We now introduce an undirected version of our decompositions and show
that it is ``equivalent'' to the directed version. More precisely, we
shall give simple constructions turning a directed decomposition into
an undirected one and vice-versa. Despite this equivalence, it will be
convenient to have both versions, because in some proofs it is easier
to work with the directed version (for example, the proof of the
Exactness Lemma) and in some proofs it is easier to work with the undirected
version (for example, the proof of Lemma~\ref{lem:bw-upperbound} in
this section and, more importantly, the proof of the Duality Theorem
in Section~\ref{sec:duality}).

Trees and graphs are undirected by default, so we omit the qualifier
``undirected'' in the following. We denote the set of all neighbours
of a node $t$ of a tree or graph $T$ by $N^T(t)$ or just $N(t)$ if $T$
is clear from the context.
We denote the set of leaves of a
tree $T$, that is, nodes of degree at most $1$, by $L(T)$ and call all
non-leaf nodes \emph{internal nodes}. A tree is
\emph{cubic} if all internal nodes have degree $3$. We refer to pairs
$(t,u)$ where $tu\in E(T)$ as \emph{oriented edges} and denote the set
    of all oriented edges of $T$ by $\vec E(T)$.

\begin{defn}
  Let ${\U}$ be a finite set.
  \begin{enumerate}
  \item A \emph{pre-decomposition} of ${\U}$ is a pair $(T,\gamma)$
    consisting of a cubic tree $T$ and a mapping
    $\gamma:\vec E(T)\to 2^{\U}$ such that 
    \begin{eroman}
      \item
        $\gamma(t,u)=\bar{\gamma(u,t)}$ for all $(t,u)\in\vec
        E(T)$,
      \item
        $\gamma(t,u_1)\cup\gamma(t,u_2)\cup\gamma(t,u_3)={\U}$
        for all internal nodes $t\in V(T)$ with $N(t)=\{u_1,u_2,u_3\}$.
    \end{eroman}
  \item A pre-decomposition $(T,\gamma)$ is \emph{complete} if
    $|\gamma(t,u)|=1$ for all leaves $u\in L(T)$ with $N(u)=\{t\}$.
  \item A pre-decomposition $(T,\gamma)$ is \emph{exact} at a
    node $t\in V(T)$ with $N(t)=\{u_1,u_2,u_3\}$ if the sets
    $\gamma(t,u_i)$ are mutually disjoint.
  \item A \emph{decomposition} is a 
    pre-decomposition that is exact at all internal nodes.
  \item A \emph{branch decomposition} is a complete
    decomposition.
  \end{enumerate}
  If $\kappa$ is a connectivity function on ${\U}$, then a
\emph{(complete, branch, pre-)decomposition} of $\kappa$ is a
(complete, branch, pre-)decomposition of ${\U}$.
\begin{enumerate}[start=6]
\item The \emph{width} of a pre-decomposition $(T,\gamma)$ of
    $\kappa$ is
    \[
    \width(T,\gamma):=\max\big\{\kappa(\gamma(t,u))\bigmid
    (t,u)\in\vec E(T)\big\}.\uende
    \]
\end{enumerate}
\end{defn}

Let $(T,\gamma)$ be a pre-decomposition ${\U}$. For every
$(t,u)\in\vec E(T)$, we call $\gamma(t,u)$ the \emph{cone} of the
pre-decomposition at $(t,u)$. For undirected pre-decompositions, the
cones coincide with the \emph{separations}, and we let 
\[
\Sep(T,\gamma):=\big\{\gamma(t,u)\bigmid
(t,u)\in\vec E(T)\big\}.
\]
It will be convenient to let $\gamma(t):=\gamma(s,t)$ for leaves $t\in
L(T)$ with $N(t)=\{s\}$. We call the sets $\gamma(t)$ for $t\in L(T)$
the \emph{atoms} of the decomposition and denote the set of all atoms
by $\At(T,\gamma)$. If $T$ is a one-node
tree with $V(T)=\{t\}$, we let $\gamma(t):={\U}$. Note that if
$(T,\gamma)$ is a decomposition then the restriction of 
$\gamma$ to the leaves determines $\gamma$. 

\begin{lem}\label{lem:dir-vs-undir}
  \begin{enumerate}
  \item For every pre-decomposition $(T,\gamma)$ of ${\U}$
    there is a directed pre-decomposition $(T_{\to},\gamma_{\to})$ of $\kappa$ such
    that $\Sep(T,\gamma)=\Sep(T_{\to},\gamma_{\to})$ and $\At(T,\gamma)=\At(T_{\to},\gamma_{\to})$.
  \item For every directed pre-decomposition $(T_{\to},\gamma_{\to})$ of ${\U}$
    that is exact at the root of $T_{\to}$ there is a pre-decomposition $(T,\gamma)$ of $\kappa$ such
    that $\Sep(T,\gamma)=\Sep(T_{\to},\gamma_{\to})$ and $\At(T,\gamma)=\At(T_{\to},\gamma_{\to})$.
  \end{enumerate}
\end{lem}

Note that (2) applies to all directed decompositions $(T_\to,\gamma_\to)$, because directed
decompositions are exact at every node.

\begin{proof}
  To prove the implication (1), let $(T,\gamma)$ be a
  pre-decomposition of ${\U}$. If
  $E(T)=\emptyset$, we let $T_{\to}$ be the one-node directed tree and
  define $\gamma_{\to}$ by $\gamma(r):={\U}$. Then
  $\Sep(T,\gamma)=\Sep(T_{\to},\gamma_{\to})=\emptyset$ and
  $\At(T,\gamma)=\At(T_{\to},\gamma_{\to})=\{{\U}\}$. 

  Otherwise, let $e=s_0s_1\in E(T)$ be an arbitrary edge. To construct
  $T_{\to}$, we subdivide the edge $e$, inserting a new node $r$. Then we
  orient all edges away from $r$ and obtain a directed tree $T_{\to}$ with
  root $r$. We define $\gamma_{\to}:V(T_{\to})\to2^{\U}$ by
  \[
  \gamma_{\to}(t):=
  \begin{cases}
    {\U}&\text{if }t=r,\\
    \gamma(s_{1-i},s_i)&\text{if }t=s_i,\\
    \gamma(s,t)&\text{if }t\neq r,s_0,s_1\text{ and $s$ is the
      parent of $t$ in $T_{\to}$}.
  \end{cases}
  \]
  It is straightforward to verify that $(T_{\to},\gamma_{\to})$ is a directed
  decomposition with $\Sep(T,\gamma)=\Sep(T_{\to},\gamma_{\to})$ and
  $\At(T,\gamma)=\At(T_{\to},\gamma_{\to})$. 

  \medskip
  To prove (2), we simply revert this
  construction. Let $(T_{\to},\gamma_{\to})$ be a
  directed decomposition of ${\U}$. Without loss of generality we assume
  that $|V(T_{\to})|>1$. Let $r$ be the root of $T_{\to}$ and $s_0,s_1$ its
  children. Let $T$ be the tree obtained from the undirected tree
  underlying $T_{\to}$ suppressing $r$, that is, deleting $r$ and adding an
  edge $s_0s_1$. We define $\gamma:V(T)\to2^{\U}$ by
  \[
  \gamma(t,u):=
  \begin{cases}
    \gamma_{\to}(s_i)&\text{if }(t,u)=(s_{1-i},s_i),\\
    \gamma_{\to}(u)&\text{if }(t,u)\in E(T_{\to}),\\
    \gamma_{\to}(t)&\text{if }(u,t)\in E(T_{\to}).
  \end{cases}
  \]
  Again, it is straightforward to verify that this construction works.
\end{proof}

\begin{cor}
  Let $\kappa$ be a connectivity function on a set ${\U}$. Then
    \[
    \bw(\kappa)=\min\big\{\width(T,\gamma)\bigmid(T,\gamma)\text{ branch decomposition of }\kappa\big\}.\uende
    \]
\end{cor}

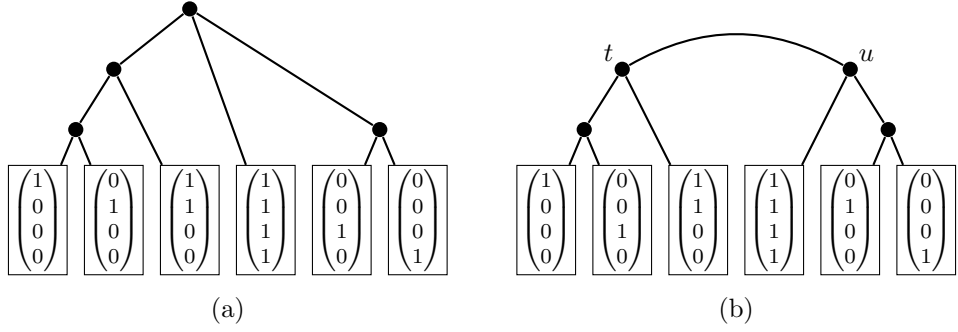
\begin{figure}
  \centering
  \input{branchdec1}
  \caption{Two branch decompositions corresponding to the directed
    branch decompositions in Figure~\ref{fig:dbranchdec1}}
  \label{fig:branchdec1}
\end{figure}

\begin{exa}
  Figure~\ref{fig:branchdec1} shows the two branch decompositions
  obtained by applying the construction of the proof of
  Lemma~\ref{lem:dir-vs-undir} to the directed
    branch decompositions in Figure~\ref{fig:dbranchdec1}.

    The separation at the oriented edge $(t,u)$ of the second decomposition
    is
    \[
    \left\{\begin{pmatrix}
      1\\1\\1\\1
    \end{pmatrix},
\begin{pmatrix}
      0\\1\\0\\0
    \end{pmatrix},
\begin{pmatrix}
      0\\0\\0\\1
    \end{pmatrix}\right\}
    \]
    \uend
\end{exa}

\begin{rem}
  This is a good place to introduce the concept of \emph{canonicity}
  of a construction or algorithm. In general a construction (or
  algorithm) is \emph{canonical} if every isomorphism between its
  input objects commutes with an isomorphism between the output
  objects.  More formally, suppose we have a construction (or
  algorithm) $A$ that associates an output $A(I)$ with every input
  $I$. Then the construction is \emph{canonical} if for any two inputs
  $I_1$ and $I_2$ and every isomorphism $f$ from $I_1$ to $I_2$ there
  is an isomorphism $g$ from $A(I_1)$ to $A(I_2)$ such that
  $g(A(I_1))=A(I_2)$, that is, the following diagram commutes:
\[
\begin{tikzcd}
I_1 \arrow{r}{f} \arrow{d}{A} & I_2 \arrow{d}{A} \\
A(I_1)\arrow{r}{g} & A(I_2)
\end{tikzcd}.
\]
For example, the construction of a decomposition from a directed
decomposition in the proof of the implication (2)$\implies$(1) of
Lemma~\ref{lem:dir-vs-undir} is canonical, but the construction of a directed
decomposition from a
decomposition in the proof of the implication (1)$\implies$(2) is not, because
it depends on the choice of the edge $s_0s_1$ to be subdivided.
\uend
\end{rem}

The following corollary is an immediate consequence of
of Lemma~\ref{lem:dir-vs-undir} and the Exactness Lemma
(Lemma~\ref{lem:exact}).

\begin{cor}[Exactness Lemma for Undirected Decompositions]\label{cor:exact}
  Let $(T,\gamma)$
  a pre-decomposition of $\kappa$. Then there is a decomposition
  $(T',\gamma')$ of $\kappa$ such that
  \begin{eroman}
  \item $\width(T',\gamma')\le\width(T,\gamma)$;
  \item The atoms of $(T',\gamma')$ are subsets of the
    atoms of $(T,\gamma')$, that is, for every $t'\in
    L(T')$ there is a $t\in L(T)$ such that $\gamma'(t')\subseteq\gamma(t)$.
  \end{eroman}
\end{cor}

\begin{lem}\label{lem:bw-upperbound}
  Let $\kappa$ be a connectivity function on a set ${\U}$. Then
  \begin{equation}\label{eq:bw1}
  \bw(\kappa)\le \val(\kappa)\cdot\ceil{\frac{|{\U}|}{3}}.
  \end{equation}
\end{lem}

\begin{proof}
  We may assume without loss of generality that $|{\U}|\ge
  3$, because if $|{\U}|\le 2$ then $\bw(\kappa)=\val(\kappa)$. We partition ${\U}$
  into three nonempty sets ${\U}_1,{\U}_2,{\U}_3$ of size
  \begin{equation}
    \label{eq:1}
  |{\U}_i|\le\ceil{\frac{|{\U}|}{3}}.
  \end{equation}
  For $i=1,2,3$, let $T_i$ be a cubic tree with $|{\U}_i|$
  leaves and $f_i$ a bijection from $L(T_i)$ to ${\U}_i$. Without loss of
  generality we assume that the trees $T_1,T_2,T_3$ are mutually node-disjoint.
  We form a new tree $T$ by joining the $T_i$
  at a
  new node $s$. More precisely, for
  $i=1,2,3$ we pick an arbitrary edge $e_i\in E(T_i)$ and subdivide
  it, inserting a new node $s_i$. If $E(T_i)=\emptyset$, we let $s_i$
  be the unique node of $T_i$. Then we add an edge between $s_i$ and
  $t$. This yields a new cubic tree $T$ with $L(T)=L(T_1)\cup
  L(T_2)\cup L(T_3)$. We define a bijection $f:L(T)\to {\U}$ by
  $f(t):=f_i(t)$ for all $t\in L(T_i)$. We define $\gamma:\vec E(T)\to
  2^{\U}$ by letting $\gamma(s,t)$ to be the set of all $f(u)$ for
  leaves $u\in L(T)$ in the connected component of $T-st$ that
  contains $t$. Clearly, $(T,\gamma)$ is a branch decomposition.

  It remains to verify that this branch decomposition has width at
  most $\val(\kappa)\cdot\ceil{|{\U}|/3}$. Let $(t,u)\in\vec E(T)$. If
  $(t,u)=(s,s_i)$ then $\gamma(t,u)={\U}_i$ and thus, by
  Lemma~\ref{lem:bp}(3) and \eqref{eq:1},
  $\kappa(\gamma(t,u))\le \val(\kappa)\cdot\ceil{|{\U}|/3}$.
If
  $(t,u)=(s_i,s)$ we just use the fact that
  $\gamma(t,u)=\gamma(u,t)$. Otherwise, $(t,u)\in \vec E(T_i)$
  for some $i$. Without loss of generality we assume that $(t,u)$ is
  pointing away from $s$. Then $\gamma(t,u)\subseteq {\U}_i$, and
  again by
  Lemma~\ref{lem:bp}(3) and \eqref{eq:1},
  $\kappa(\gamma(t,u))\le \val(\kappa)\cdot\ceil{|{\U}|/3}$.
\end{proof}

\subsection{Branch Decomposition of Graphs}
\label{sec:graphdec}

In this section, we study branch decompositions and branch width of the
four connectivity functions 
that we defined for graphs $G$:
\begin{itemize}
\item the edge-connectivity function $\nu_G$ defined on $V(G)$,
\item the vertex-connectivity function $\kappa_G$ defined on $E(G)$,
\item the matching connectivity function $\mu_G$ defined on $V(G)$,
\item the cut-rank function $\rho_G$ defined on $V(G)$.
\end{itemize}
Branch decompositions of $\kappa_G$ are usually called \emph{branch
decompositions of $G$}, and the branch width of $\kappa_G$ is known as
the \emph{branch width of $G$}. Branch
decompositions of $\rho_G$ are called \emph{rank decompositions} of $G$, and the branch width of $\rho_G$ is known as
the \emph{rank width of $G$}.

As for all graphs $G$ we have $\val(\rho_G),\val(\mu_G)\le 1$, it is an
immediate consequence of
Lemma~\ref{lem:bw-upperbound} that
\begin{equation}
  \label{eq:2}
   \bw(\mu_G)\le\ceil{\frac{|V(G)|}{3}}
  \quad\text{and}\quad
   \bw(\rho_G)\le\ceil{\frac{|V(G)|}{3}}.
\end{equation}
\begin{exa}
  It follows from \eqref{eq:2} that for the complete graph
  $K_n$ we have \[\bw(\mu_{K_n}),\bw(\rho_{K_n})\le\ceil{n/3}.\] This
  bound is tight for $\mu_{K_n}$, but not for $\rho_{K_n}$. As a
  matter of fact, we have
  \[
  \bw(\rho_{K_n})\le1,
  \]
  with equality for every $n\ge 2$. This follows from the simple
  observation that for every $X\subseteq V(K_n)$ the matrix
  $M_{K_n}(X,\bar X)$ has only $1$-entries and hence row rank
  $1$. Thus every branch decomposition of $\rho_{K_n}$ has width $1$.
  \uend
\end{exa}

As $\val(\kappa_G)\le 2$, by Lemma~\ref{lem:bw-upperbound} we also get
\[
  \bw(\kappa_G)\le\ceil{\frac{2\cdot|E(G)|}{3}}.
\]
The following exercise shows that we get the
same bound in terms of $|V(G)|$.

\begin{exe}
  Prove that 
  \begin{equation}
    \label{eq:3}
    \bw(\kappa_G)\le\ceil{\frac{2|V(G)|}{3}}
  \end{equation}
  for every graph
  $G$.

  \textit{Hint:} Partition $V(G)$ into three sets
  $V_1,V_2,V_3$ of the same size. Then partition $E(G)$ into
  sets $E_1,E_2,E_3$, where $E_i$ contains all edges with both
  endvertices in $V_i$ and all edges with one endvertex in $E_i$ and
  the other endvertex in $E_j$ for $j=i+1\bmod3$.
\end{exe}

\begin{exe}
  Prove that $\bw(\kappa_G),\bw(\mu_G),\bw(\rho_G)\le 1$ for
  all forests $G$. Furthermore, prove that $\bw(\kappa_G)\le 1$ if any
  only if $G$ is a forest, and give an example of a graph $G$ that is
  not a forest, but still satisfies $\bw(\mu_G)=\bw(\rho_G)=1$.
\end{exe}

Recall that for the edge-connectivity function $\nu_G$ we have
$\val(\nu_G)=\Delta(G)$, where $\Delta(G)$ denotes the maximum degree of $G$. Thus
\[
\Delta(G)\le \bw(\nu_G)\le \Delta(G)\cdot\ceil{\frac{|G|}{3}}.
\]
The lower bound
is trivial and the upper bound follows from Lemma~\ref{lem:bw-upperbound}.

\begin{theo}\label{theo:mu_vs_kappa}
  Let $G$ be a graph with at least one vertex of degree $2$. Then
  $\bw(\mu_G)\le\bw(\kappa_G)\le 2\bw(\mu_G)$.
\end{theo}

Note that if $G$ is a graph of maximum degree
$1$, then $\bw(\kappa_G)=0$ and $\bw(\mu_G)=1$. The theorem is a
variant of a theorem due to Vatshelle~\cite{vat12} asserting that
$\bw(\mu_G)$ is linearly bounded in terms of the tree width of $G$
(see below) and vice-versa.

\begin{proof}[Proof of the first inequality of Theorem~\ref{theo:mu_vs_kappa}]
  Let $G$ be a graph  with at least one vertex of degree $2$. Then
  $\bw(\kappa_G)\ge 1$. Without loss of generality we may assume that
  $G$ has no isolated vertices, because adding isolated vertices does
  not increase $\bw(\mu_G)$.
 Let $(T,\gamma)$ be a directed branch
  decomposition of $\kappa_G$. In a first step, we define a
  pre-decomposition $(T,\gamma')$ of $\mu_G$ of the same width. We define
  $\gamma':V(T)\to 2^{V(G)}$ by letting
  \[
  \gamma'(t):=V\big(\gamma(t)\big).
  \]
  It is easy to verify that $(T,\gamma')$ is indeed a pre-decomposition
  of $\mu_G$. In particular, $\gamma'(r)=V(\gamma(r))=V(E(G))=V(G)$ for the
  root $r$ of $T$ by our assumption that $G$ have no isolated
  vertices. It follows from Lemma~\ref{lem:mu-vs-kappa}(1) that
  $\mu_G(\gamma'(t))=\mu_G(V(\gamma(t)))\le\kappa_G(\gamma(t))$ for all
  $t\in V(T)$ and thus $\width(T,\gamma')\le\width(T,\gamma)$. We apply the
  Exactness Lemma to $(T,\gamma')$ and obtain a decomposition $(T,\gamma'')$
  of $\mu_G$ such that $\width(T,\gamma'')\le \width(T,\gamma')$ and
  $\gamma''(t)\subseteq\gamma'(t)$ for all leaves $t\in L(T)$. Note that
  $|\gamma''(t)|\le|\gamma'(t)|=|V(\gamma(t))|=2$ for all leaves $t\in L(T)$,
  because $\gamma(t)$ consists of a single edge. We define a binary
  directed tree $T'$ by
  attaching two new leaves to all $t$ of $T$ such
  that $|\gamma''(t)|=2$, and we define $\gamma''':V(T')\to 2^{V(G)}$ by
  $\gamma'''(t):=\gamma''(t)$ for all $t\in V(T)$ and $\gamma'''(u_i)=v_i$ if
  $u_1,u_2$ are the new children of a leaf $t\in L(T)$ with
  $\gamma''(t)=\{v_1,v_2\}$. Then $(T',\gamma''')$ is a decomposition
 of $\mu_G$ of width $\max\{1,\width(T,\gamma'')\}$. The only thing that keeps
 $(T',\gamma''')$ from being a branch decomposition is that there may be
 leaves $t\in L(T')$ with $\gamma'''(t)=\emptyset$. We can fix this with
 an application of Lemma~\ref{lem:empty-leaves}.
\end{proof}

The second inequality will be proved in Section~\ref{sec:duality}.

\begin{exe}
  Prove that both inequalities in Theorem~\ref{theo:mu_vs_kappa} are tight.
\end{exe}

Better known than branch decompositions and branch width are tree
decompositions and tree width of graphs. We will show that they are closely
related.  Most readers will be familiar with tree decompositions and
tree width, but
let us recall the definitions anyway.

A \emph{tree decomposition} of a graph $G$ is a pair
$(T,\beta)$, where $T$ is a tree and $\beta:V(T)\to 2^{V(G)}$ such
that for all $v\in V(G)$ the set $\{t\in V(T)\mid v\in\beta(t)\}$ is
connected in $T$ and for all $e=vw\in E(G)$ there is a $t\in V(T)$
such that $v,w\in\beta(t)$.

The \emph{width} of a tree decomposition $(T,\beta)$ is 
  \[
  \width(T,\beta):=\max\big\{|\beta(t)|\bigmid t\in V(T)\big\}-1.
  \]
The \emph{tree width} of $G$, denoted by $\tw(G)$, is the minimum of
the widths of all tree decompositions of $G$.

\begin{theo}[\cite{gm10}]\label{theo:bw_vs_tw}
  For every graph $G$, 
  \[
  \bw(\kappa_G)\le\tw(G)+1\le\max\left\{\frac{3}{2}\bw(\kappa_G),\;2\right\}.
  \]
\end{theo}

\begin{proof}
  To prove the first inequality, let $(T,\beta)$ be a tree
  decomposition of $G$ of width $k$.

  In a first step, we transform $(T,\beta)$ into a tree decomposition
  $(T_1,\beta_1)$ of $G$ of width $k$ such that for every edge
  $e=vw\in E(G)$ there is a leaf $t_e\in L(T_1)$ with
  $\beta_1(t_e)=\{v,w\}$. To form $T_1$, for every $e=vw\in E(G)$, we
  pick a node $t\in V(T)$ such that $v,w\in\beta(t)$ and attach
  a new leaf $t_e$ to it. Then we let $\beta_1(t_e):=\{v,w\}$ and
  $\beta_1(t):=\beta(t)$ for all $t\in V(T)$.

  In a second step, we transform $(T_1,\beta_1)$ into a tree decomposition
  $(T_2,\beta_2)$ of $G$ of width $k$ such that for every leaf $t\in
  L(T_2)$ there is an edge
  $e=vw\in E(G)$ such that $\beta_2(t)=\{v,w\}$ and for every edge
  $e=vw\in E(G)$ there is exactly one leaf $t\in L(T_2)$ such that  $\beta_2(t)=\{v,w\}$. To form $T_2$, we
  repeatedly delete leaves $t$ such that there is no edge $e=vw\in
  E(G)$ with $\beta_1(t)=\{v,w\}$. If an edge appears at
  several leaves, we delete all but one of these leaves, and then
  possibly repeat the construction if the deletions generate new leaves.

   In a third step, we transform $(T_2,\beta_2)$ into a tree decomposition
  $(T_3,\beta_3)$ of $G$ of width $k$ such that $T_3$ is a tree of
  maximum degree at most $3$. We replace every node $t$ with $\ell>3$
  neighbours by a cubic tree $S_t$ with exactly $\ell$ leaves, identifying
  the leaves of $S_t$ with the neighbours of $t$ in $T_2$, and let
  $\beta_3(s):=\beta_2(t)$ for every internal node $s$ of $S_t$. 

  Let $T_4$ be the cubic tree obtained from $T_3$ by suppressing all nodes
  of degree $2$. (\emph{Suppressing} an node $v$ of degree $2$ means
  deleting the node and adding an edge between its two
  neighbours.) Then $L(T_4)=L(T_3)$. Let $\gamma:L(T_4)\to 2^{E(G)}$
  be defined by $\gamma(t):=\{e\}$ for the unique edge $e=vw$ such that
  $\beta_3(t)=\{v,w\}$. Now we define $\gamma:\vec
  E(T_4)\to 2^{E(G)}$ in the usual way: we let $\gamma(s,t)$ be the
  union of all $\gamma(u)$ for leaves $u$ in the connected component of
  $T-\{st\}$ that contains $t$. Then $(T_4,\gamma)$ is a branch
  decomposition of $G$. 

  To see that $\width(T_4,\gamma)\le k+1$, we observe that
  $\partial_G(\gamma(s,t))\subseteq\beta_3(t)$ for all $(s,t)\in\vec
  E(T_4)$. Indeed, if $v\in \partial_G(\gamma(s,t))$ there are
  leafs $u,u'$ in different components of $T_4-\{st\}$ and edges
  $e=vw,e'=vw'\in E(G)$ such that $\beta_3(u)=\{v,w\}$ and
  $\beta_3(u')=\{v,w'\}$. As $t$ appears on the path from $u$ to $u'$
  in $T_3$, we have $v\in\beta_3(t)$.

  Hence $(T_4,\gamma)$ is a branch
  decomposition of $G$ of width at most $k+1$. This proves the first
  inequality.

  \medskip
  To prove the second inequality, let $(T,\gamma)$ be a branch
  decomposition of $G$ of width $k$. We define $\beta:V(T)\to
  2^{V(G)}$ as follows:
  \begin{itemize}
  \item For all leaves $t\in L(T)$, we let $\beta(t)$ be the set of
    endvertices of the unique edge $e$ such that $\gamma(t)=\{e\}$.
  \item For all internal nodes $s$ with neighbours $t_1,t_2,t_3$ we
    let
    \[
    \beta(s)=\partial(\gamma(s,t_1))\cup \partial(\gamma(s,t_2))\cup \partial(\gamma(s,t_2)).
    \]
  \end{itemize}
  We leave it to the reader to prove that $(T,\beta)$ is a tree
  decomposition of $G$. The size of the bags at the leaves is $2$, and
  the size of the bags at the internal nodes is at most $(3/2)k$, because if $s$ is an internal node of
  $T$ with neighbours $t_1,t_2,t_3$, then every vertex $v\in \beta(s)$
  is contained in a at least $2$ of the sets $\gamma(s,t_i)$. Thus
  \[
  3k\ge
  |\partial(\gamma(s,t_1))|+|\partial(\gamma(s,t_2))|+|\partial(\gamma(s,t_2))|\ge
  2|\beta(s)|.\qedhere
  \]
\end{proof}

\section{Well-Linked Sets}
\label{sec:well-linked}

In the previous section, we have seen several examples establishing
upper bounds on the branch width of a connectivity function. In this
and the next section, we will be concerned with lower bounds. We will develop
\emph{obstructions} to small branch width.

Intuitively, branch width is a measure for the ``global connectivity''
of a connectivity system: if $\bw(\kappa)$ is small then ${\U}$
can be decomposed along separations of small order, and
therefore the ``global connectivity'' of $\kappa$
may be viewed as being low. Thus the most obvious obstruction to small branch
width is a ``highly connected set.'' There are various views on what
might constitute a highly connected set with respect to a connectivity
function; the following is quite natural.

\begin{defn}
  Let $\kappa$ be a connectivity function on a set ${\U}$. A set
  $W\subseteq {\U}$ is \emph{well-linked} if for every $X\subseteq {\U}$,
  \[
  \kappa(X)\ge\min\{|W\cap X|,|W\setminus X|\}.
  \uende
  \]
\end{defn}

\begin{exa}
  Let $G$ be a a graph and $C\subseteq G$ a cycle of length at least
  $4$. Then any set $W\subseteq V(C)$ of size $|W|\le 4$ is
  well-linked with respect to $\nu_G$ and $\mu_G$.
  If $C$ is chordless, that is, an induced
  subgraph of $G$, then  $W$ is also
  well-linked with respect to $\rho_G$. Furthermore, every set $X\subseteq E(C)$ of size $|X|\le 4$ is
  well-linked with respect to $\kappa_G$.
  \uend
\end{exa}

In the next example, we characterise the well-linked sets of the
matching connectivity function of a graph.

\begin{exa}
  Let $G$ be a graph. We claim that a set $W\subseteq V(G)$ is
  well-linked with respect to $\mu_G$ if and only if for all
  disjoint sets $Y,Z\subseteq W$ of the same size $|Y|=|Z|=:\ell$ there is a
  family of $\ell$ mutually disjoint paths from $Y$ to $Z$.

  To prove the forward direction of this claim, let $W\subseteq V(G)$
  be well-linked with respect to $\mu_G$. Let $Y,Z\subseteq W$ be
  disjoint sets of the same size $\ell$. Suppose for contradiction
  that there are no $\ell$ mutually disjoint paths from $Y$ to $Z$. By
  Menger's Theorem, there is a set $S\subseteq V(G)$ of size
  $|S|<\ell$ that separates $Y$ from $Z$. Let $X$ be the union of
  $S\cap Y$ and the vertex sets of all connected components of
  $G\setminus S$ that contain a vertex of $Y$. Then $S$ is a vertex
  cover for $E(X,\bar X)$, and thus $\mu_G(X)\le|S|<\ell$. However, we
  have $|W\cap X|\ge|Y|=\ell$ and $|W\setminus X|\ge|Z|=\ell$. This
  contradicts $W$ being well-linked for $\mu_G$. 

  To prove the backward direction, suppose that for all disjoint sets
  $Y,Z\subseteq W$ of the same size $|Y|=|Z|=:\ell$ there is a family
  of $\ell$ mutually disjoint paths from $Y$ to $Z$. Let
  $X\subseteq V(G)$, and let $S$ be a minimum vertex cover of
  $E(X,\bar X)$. Without loss of generality we assume that $|W\cap
  X|\le|W\setminus X|$. Let $Y:=W\cap X$ and $Z\subseteq W\setminus X$ with
  $|Y|=|Z|=:\ell$. Then there are $\ell$ mutually vertex disjoint
  paths from $Y$ to $Z$. Each of these paths must contain an edge in
  $E(X,\bar X)$ and thus a vertex of $S$. Hence
  $\mu_G(X)=|S|\ge\ell=\min\{|W\cap X|,|W\setminus X|\}.$ This
  proves that $W$ is well-linked with respect to $\mu_G$.
  \uend
\end{exa}

The following theorem relates well-linkedness and branch width.

\begin{theo}\label{theo:conn-wl}
  Let $\kappa$ be a nontrivial connectivity function and $k\ge 1$.
  \begin{enumerate}
  \item If $\bw(\kappa)\le k$, then there is no well-linked set of size
    greater than $3k$.
  \item If there is no well-linked set of size greater than $\frac{k}{\val(\kappa)}-1$ then
    $\bw(\kappa)\le k$.
  \end{enumerate}
\end{theo}

In the proof of part (1) of this theorem, we use the following
observation for the first time. Despite its simplicity, I find it
worthwhile to highlight this, because it is a standard argument that
we shall apply several more times.

\begin{obs}
  Let $T$ be an (undirected) tree, and let $\omega$ be an arbitrary
  orientation of the edges of $T$, that is, $\omega:E(T)\to\vec E(T)$
  such that $\omega(st)=(s,t)$ or $\omega(st)=(t,s)$ for all $st\in
  E(T)$. Then there is a node $s\in V(T)$ such that all edges incident
  with $s$ are oriented towards $s$, that is, $\omega(st)=(t,s)$ for
  all $t\in N^T(s)$.
\end{obs}

To see this, we just follow an oriented path in the tree starting at
an arbitrary node until we
reach a node with no outgoing edges. This will happen eventually,
because there are no cycles in $T$.

\begin{proof}[Proof of Theorem~\ref{theo:conn-wl}(1)]
  As $\kappa$ is nontrivial, we have $\val(\kappa)\ge 1$.  Let ${\U}$ be
  the universe of $\kappa$ and $W\subseteq {\U}$ such that $|W|>3k$. We
  shall prove that $W$ is not well-linked. Let $(T,\gamma)$ be a branch
  decomposition of $\kappa$ of width at most $k$. We orient the edges
  of $T$ towards the bigger part of $W$. That is, we orient $st\in E(T)$
  towards $t$ if
  $|\gamma(s,t)\cap W|> |\gamma(t,s)\cap W|$ and towards $s$ if
  $|\gamma(s,t)\cap W|< |\gamma(t,s)\cap W|$. We break ties
  arbitrarily.

  Then there is a node $s\in V(T)$ such that all edges incident with
  $s$ are oriented towards $s$. As $|W|\ge 3$ and
  $|\gamma(t,s)|=1$ if $s$ is a leaf and $t$ its
  neighbour, the node $s$ is not a leaf. Thus $s$ has three
  neighbours, say, $t_1,t_2,t_3$. Let $W_i:=\gamma(s,t_i)\cap
  W$. Then $|W_i|\le |W|/2$, and $W_1,W_2,W_3$ form a partition of
  $W$.  Without loss of generality we assume that
  $|W_1|\ge|W_2|\ge|W_3|$. Then $|W_1|\ge|W|/3$ and $|W_2\cup W_3|\ge
  |W|/2$. Let $X:=\gamma(s,t_1)$. Then $\kappa(X)\le k$, because
  the width of the decomposition $(T,\gamma)$ is at most $k$, and $|W\cap
  X|=|W_1|\ge|W|/3>k$ and $|W\setminus X|=|W_2\cup
  W_3|\ge|W|/2>k$. Thus $W$ is not well-linked.
\end{proof}

The proof of part (2) requires more preparation. Let $X\subseteq
{\U}$. We define a function $\pi_X:2^X\to\ZZ$ by 
\[
\pi_X(Y):=\min\{ \kappa(Y')\mid Y\subseteq Y'\subseteq X\}.
\]
It is easy to verify that $\pi_X$ is an integer polymatroid on $X$, that is, normalised,
monotone, and submodular. Let us call a set $Y\subseteq X$ \emph{free}
in $X$ if $|Y'|\le \pi_X(Y')$ for all $Y'\subseteq Y$. It can be shown
that the free subsets of $X$ are the independent sets of a matroid on
$X$ of rank at least $\kappa(X)/\val(\kappa)$ (see \cite[Chapter
12]{oxl11}), but we do not use this here. %

\begin{lem}\label{lem:wl1}
  There is a set $Y\subseteq X$ such that 
  $Y$ is free in $X$ and $|Y|\ge \kappa(X)/\val(\kappa)$ and $\pi_X(Y)=\kappa(X)$.
\end{lem}

\begin{proof}
  Let $Y\subseteq X$ be an inclusionwise maximal free set. Then for
  all $x\in
  X\setminus Y$ the set $Y\cup\{x\}$ is not free in $X$. Thus there is
  a $Y_x\subseteq Y$ such that $|Y_x|+1>\pi_X(Y_x\cup\{x\})$. By
  the monotonicity of $\pi_X$ and since $Y$ is free, we have
  \[
  |Y_x|+1>\pi_X(Y_x\cup\{x\})\ge\pi_X(Y_x)\ge|Y_x|,
  \]
  which implies $\pi_X(Y_x\cup\{x\})=\pi_X(Y_x)$. Thus
  \[
  \pi_X(Y_x)+\pi_X(Y)=\pi_X(Y_x\cup\{x\})+\pi_X(Y)\ge\pi_X(Y_x)+\pi_X(Y\cup\{x\}),
  \]
  where the second inequality holds by submodularity. Hence
  \[
  \pi_X(Y)=\pi_X(Y\cup\{x\}). 
  \]
  As this holds for all $x\in X\setminus Y$, an easy induction based
  on the submodularity and monotonicity of $\pi_X$ implies
  $\pi_X(Y)=\pi_X(Y\cup (X\setminus Y))=\pi_X(X)$.

  By Lemma~\ref{lem:bp}(3) and the definition of $\pi_X$, we have
  \[
  \val(\kappa)\cdot|Y|\ge\kappa(Y)\ge\pi_X(Y)=\pi_X(X)=\kappa(X),
  \]
  which implies the lemma.
\end{proof}

To understand the following examples, we observe that if
$\val(\kappa)=1$ then a set $Y$ is free in a set $X$ if and only if
$\pi_X(Y)=|Y|$. This follows from the fact that $|Y|\ge\kappa(Y)$ for
all $Y$ if $\val(\kappa)=1$.

\begin{exa}
  Let $G$ be a graph and $X\subseteq V(G)$. A set $Y\subseteq X$ is
  free in $X$ with respect to $\mu_G$ if and only if there is a a minimum vertex
  cover $S$ for $E(X,\bar X)$ and a family of $|Y|$ mutually
  disjoint paths from $Y$ to $S$. 
  \uend
\end{exa}

\begin{exa}
  Let $G$ be a graph and $X\subseteq V(G)$. Let $Y\subseteq X$ such
  that the rows in the matrix
  $M(X,\bar X)$ corresponding to
  the elements of $Y$ are linearly independent. Then $Y$ is free in
  $X$. (The converse does not necessarily hold.)
\uend
\end{exa}

\begin{lem}\label{lem:wl2}
  Let $\kappa$ be a connectivity function on ${\U}$. Let $X\subseteq {\U}$ and $Y$ free in $\bar X$, and let
  $Z\subseteq {\U}$ such that
  \begin{equation}
    \label{eq:10}
    \kappa(Z)<\min\{|Y\cap Z|,|Y\setminus Z|\}. 
  \end{equation}
  Then $X\cap Z$ and $X\setminus Z$ are both nonempty with
  $\kappa(X\cap Z),\kappa(X\setminus Z)<\kappa(X)$.
\end{lem}

\begin{proof}
  Because of the symmetry between $Z$ and $\bar Z$, we only have to prove $X\cap
  Z\neq\emptyset$ and $\kappa(X\cap Z)<\kappa(X)$. 

  If $X\cap Z=\emptyset$ then $Y\cap Z\subseteq Z\subseteq\bar X$ and
  $\pi_{\bar X}(Y\cap Z)\le\kappa(Z)<|Y\cap Z|$, which contradicts $Y$ being free in
  $\bar X$. Thus $X\cap
  Z\neq\emptyset$.

  Furthermore, we have $|Y\setminus Z|\le\pi_{\bar X}(Y\setminus Z)\le
  \kappa(\bar X\setminus Z)$, because $Y$ is free in $\bar X$ and by
  the definition of $\pi_{\bar X}$. Thus
  \begin{align*}
    \kappa(X\cap Z)&\le\kappa(X\cap Z)+\kappa(\bar X\setminus
                     Z)-|Y\setminus Z|\\
    &=\kappa(X\cap Z)+\kappa(X\cup Z)-|Y\setminus Z|&\text{by symmetry}\\
    &\le\kappa(X)+\kappa(Z)-|Y\setminus Z|&\text{by submodularity}\\
    &<\kappa(X)&\text{by  \eqref{eq:10}.}
  \end{align*}
\end{proof}

\begin{proof}[Proof of Theorem~\ref{theo:conn-wl}(2)]
 We assume that there is no well-linked set of
  cardinality greater than $\frac{k}{\val(\kappa)}-1$. As every set of cardinality $1$ is well-linked, we
  have $k\ge \val(\kappa)$. We shall construct a directed branch decomposition
  $(T,\gamma)$ of $\kappa$ of width at most $k$. The construction is
  iterative: we define a sequence $(T_1,\gamma_1),\ldots,(T_m,\gamma_m)$ of
  directed decompositions of $\kappa$ of width at most $k$ such that
  $(T,\gamma):=(T_m,\gamma_m)$ is complete.

  We let $T_1$ be the one-node directed tree only consisting of the
  root $r$ and let $\gamma_1(r):={\U}$. Now suppose that $(T_i,\gamma_i)$ is
  defined. If it is complete, we let $m=i$ and stop the
  construction. Otherwise, there is a leaf $t\in L(T_i)$ such that
  $|\gamma_i(t)|\ge 2$. The tree $T_{i+1}$ is obtained from $T_i$ by
  attaching two new children $u_1,u_2$ to $t$. For all nodes $s\in
  V(T_i)$, we let $\gamma_{i+1}(s):=\gamma_i(s)$. It remains to define
  $\gamma_{i+1}(u_1)$ and $\gamma_{i+1}(u_2)$.

  Let $X:=\gamma(t)$. We shall define nonempty disjoint sets $X_1,X_2$ such
  that $X_1\cup X_2=X$ and $\kappa(X_i)\le k$. Then we let
  $\gamma_{i+1}(u_i):=X_i$. If $\kappa(X)\le k-\val(\kappa)$, we pick an
  arbitrary $x\in X$ and let $X_1:=\{x\}$ and
  $X_2:=X\setminus\{x\}$. Then $\kappa(X_1)\le\val(\kappa)\le k$ and
  $\kappa(X_2)\le\kappa(X)+\val(\kappa)\le k$. 

  So suppose that $\kappa(X)>k-\val(\kappa)$. By Lemma~\ref{lem:wl1}, there is a set
  $Y\subseteq\bar X$ that is free in $\bar X$ and of cardinality
  \[
  |Y|\ge\frac{\kappa(\bar
  X)}{\val(\kappa)}>\frac{k-\val(\kappa)}{\val(\kappa)}=\frac{k}{\val(\kappa)}-1.
  \]
  Thus $Y$ is not well-linked, and there is a set $Z\subseteq {\U}$
  such that $\kappa(Z)<\min\{|Y\cap Z|,|Y\setminus Z|\}$. We let
  $X_1:=X\cap Z$ and $X_2:=X\setminus Z$. By
  Lemma~\ref{lem:wl2}, for $i=1,2$ the set $X_i$ is nonempty with $\kappa(X_i)<\kappa(X)$.
\end{proof}

\begin{exe}\label{exe:k-linked}
  Let $\kappa$ be a connectivity function on ${\U}$. A set $V\subseteq {\U}$
  is \emph{$k$-linked} if $|V|\ge 2k$ and for all disjoint
  sets $Y,Z\subseteq V$ of the same cardinality $|Y|=|Z|\le k$ there is no
  $X\subseteq {\U}$ such that $\kappa(X)<|Y|$ and $Y\subseteq X$ and
  $Z\subseteq\bar X$.

  Prove the following:
  \begin{enumerate}[label=(\alph*)]
  \item Let $V$ we a $k$-linked set and $W\subseteq V$ of cardinality $|W|\le
    2k+1$. Then $W$ is well-linked.
  \item Let $W$ be well-linked. Then $W$ is $\floor{|W|/2}$-linked.
  \item If $\bw(\kappa)\le k$, then there is no $(k+1)$-linked set of
    cardinality greater than $3k$.
  \end{enumerate}
\end{exe}

\section{Tangles}
\label{sec:tangles}

Similarly to well-linked sets, tangles describe highly connected ``regions''
of a connectivity system. However, tangles are more
elusive than well-linked sets. The region described by a tangle
may not be a subset of the universe ${\U}$. The tangle only describes the
region in a dual way, by ``pointing to it''.

\begin{figure}
  \centering
  \input{tangleidea}
  \caption{The idea of a tangle (of order $k$)}
  \label{fig:tangleidea}
\end{figure}

To understand the idea, consider Figure~\ref{fig:tangleidea}. Part (a)
shows the universe of a connectivity function, and in part (b) we
highlight what might intuitively be a ``$k$-connected region'' (for
some parameter $k$ that is irrelevant here). We make no effort to fit
this region exactly, because this would be futile anyway.\footnote{Not only
because of my limited latex-drawing abilities.} Part (c) explains why:
we display all separations of order less than $k$. Our region may be
viewed as ``$k$-connected'', because none of these separations splits the
region in a substantial way, only small parts at the boundary may be
sliced off. The region is approximately maximal with this property. It
would be hard, however, to fix the boundary of the region in a
definite way because of the ``crossing'' separations we see at the
``north-east exit'' and ``north west exit'' of the region. There is no
good way of deciding which of these separations we should take to
determine the boundary of the region. Instead, for each of these
separations we can say on which side (most of) the region is (see part
(d) of the figure). This way, we describe our region unambiguously and
without making arbitrary choices at the boundary; we simply leave the
precise boundary unspecified. A \emph{tangle} (of order $k$) does
precisely this: it gives an orientation to the separations of order
less than $k$. It is convenient to think of the part of a
separation the tangle points to, that is, the side where the presumed
``region'' described by the tangle is supposed to be, as the ``big
side''. Part (e) of the figure gives another example of a tangle of
the same order $k$ that describes a different region. Of course, in order to actually describe a region the
orientation a tangle assigns to the separations has to be
``consistent''. In particular, the intersection of two big sides should
not be empty (see part (f) of the figure).

Formally, a tangle of order $k$ is a set system (intuitively
consisting of the ``big sides'' of the separations of order less than
$k$) satisfying four axioms, of which the first and second ensure that the
tangle is indeed an orientation of the separations of order less than $k$, the
third ensures consistency, and the fourth
rules out ``trivial tangles''.

\begin{defn}
  Let $\kappa$ be a connectivity function on a set ${\U}$.  A
  \emph{$\kappa$-tangle} of order $k\ge0$ is a set $\CT\subseteq 2^{\U}$
  satisfying the following conditions.\footnote{Our definition of
    tangle differs from the one mostly found in the literature
    (e.g.~\cite{geegerwhi09,hun11}). In our definition, the ``big side''
    of a separation belongs to the tangle, which seems natural if one
    thinks of a tangle as ``pointing to a region'' (as described
    above), whereas in the definition of \cite{geegerwhi09,hun11} the
    ``small side'' of a separation belongs to the tangle. But of course both definitions yield equivalent
    theories.}
  \begin{nlist}{T}
  \setcounter{nlistcounter}{-1}
  \item\label{li:t0}
    $\kappa(X)<k$ for all $X\in\CT$, 
  \item\label{li:t1}
    For all $X\subseteq {\U}$ with $\kappa(X)<k$, either $X\in\CT$ or
    $\bar X\in\CT$.
  \item\label{li:t2}
    $X_1\cap X_2\cap X_3\neq\emptyset$ for all $X_1,X_2,X_3\in\CT$.
  \item\label{li:t3}
    $\CT$ does not contain any singletons, that is, $\{x\}\not\in\CT$ for all $x\in {\U}$.\uend
\end{nlist}
\end{defn}

We denote the order of a $\kappa$-tangle $\CT$ by $\ord(\CT)$. 

\begin{exa}\label{exa:tangle1}
  Let $G$ be a graph and $C\subseteq G$ a cycle of length at least $4$. We let
  \[
  \CT_C:=\{X\subseteq V(G)\mid \mu_G(X)<2,|V(C)\setminus X|\le 1\}.
  \]
  Then $\CT_C$ is a $\mu_G$-tangle of order $2$. 

  To see this, note that $\CT_C$ trivially satisfies \ref{li:t0}. It
  satisfies \ref{li:t1}, because for every $X\subseteq V(G)$ with
  $\mu_G(X)<2$, either $|V(C)\cap X|\le 1$ or
  $|V(C)\setminus X|\le 1$. It satisfies \ref{li:t2}, because if
  $X_1,X_2,X_3\in\CT_C$ then
  $|V(C)\setminus(X_1\cup X_2\cup X_3)|\le 3$ and thus
  $X_1\cap X_2\cap X_3\supseteq V(C)\cap X_1\cap X_2\cap
  X_3\neq\emptyset$.
  Finally, it satisfies \ref{li:t3}, because if $X\in\CT_C$ then
  $|X|\ge |V(C)\cap X|\ge 3$.

  Essentially the same argument shows that $\CT_C$ is a $\nu_G$-tangle
  (even if the length of $C$ is $3$), and if $C$ is an induced cycle
  then $\CT_C$ is a $\rho_G$-tangle of order $2$. 
  \uend
\end{exa}

\begin{exa}\label{exa:tangle2}
  Let $G$ be a graph and $H\subseteq G$ a 2-connected subgraph. That
  is, $|H|\ge 3$ and for every vertex $v\in V(H)$ the graph
  $H\setminus\{v\}$ is connected.

  Let 
  \[
  \CT_H:=\{Y\subseteq E(G)\mid\kappa_G(Y)<2,E(H)\subseteq Y\}.
  \]
  Then $\CT_H$ is a $\kappa_G$-tangle of order $2$. 
  The crucial observation to prove this is that for every
  $Y\subseteq E(G)$, if $E(H)\cap Y\neq\emptyset$ and
  $E(H)\cap\bar Y\neq\emptyset$ then
  $\partial_G(Y)\ge\partial_H(Y)\ge 2$.

  It is worth noting that the set 
  \[
  \{X\subseteq V(G)\mid\mu_G(X)<2,|V(H)\setminus X|\le 1\}
  \]
  is not necessarily a $\mu_G$ tangle. (Why?)
  \uend
\end{exa}

\begin{exa}[Robertson and Seymour~\cite{gm10}]\label{exa:4}
  Let $G$ be a graph and $H\subseteq G$ a $(k\times k)$-grid. Let
  $\CT$ be the set of all $X\subseteq E(G)$ such that $\kappa_G(X)<k$
  and $X$ contains all edges of some row of
  the grid. Then $\CT$ is a $\kappa_G$-tangle of order $k$. We omit
  the proof, which is not entirely trivial.
  \uend
\end{exa}

\begin{exa}\label{exa:wl-tangle}
  Let $\kappa$ be a connectivity function ${\U}$. Let 
  $W\subseteq
  {\U}$ be a well-linked set of cardinality $|W|\ge 2$. Then 
  \[
  \CT_W:=\Big\{X\subseteq {\U}\Bigmid \kappa(X)<|W|/3,|W\cap X|>(2/3)|W|\Big\}
  \]
  is a tangle of order $\ceil{|W|/3}$.

  To prove this, note that $\CT(W)$ trivially satisfies \ref{li:t0}. To see that it satisfies
  \ref{li:t2}, let $X_1,X_2,X_3\in\CT_W$. Then $|W\setminus
  X_i|<|W|/3$ and thus $W\not\subseteq\bar X_1\cup\bar
  X_2\cup\bar X_3$, which implies $W\cap X_1\cap X_2\cap X_3\neq\emptyset$.
  Furthermore, $\CT(W)$ 
  satisfies \ref{li:t3}, because $|W|\ge 2$ and thus
  $|\{x\}|\le(2/3)|W|$ for all $x\in {\U}$.

  To see that $\CT(W)$ satisfies \ref{li:t1}, let $X\subseteq {\U}$ with
  $\kappa(X)<|W|/3$. Since $W$ is well-linked, we have
  $\kappa(X)\ge\min\{|W\cap X|,|W\setminus X|\}$. Thus either $|W\cap
  X|<|W|/3$ or $|W\setminus X|< |W|/3$. This implies
  $\bar X\in\CT_W$ or $X\in\CT_W$.
  \uend
\end{exa}

\begin{lem}\label{lem:tangle-closure}
  Let $\CT$ be a $\kappa$-tangle of order $k$.
  \begin{enumerate}
  \item For all $X\in\CT$ and $Y\supseteq X$, if $\kappa(Y)<k$ then
    $Y\in\CT$.
   \item For all $X,Y\in\CT$, if $\kappa(X\cap Y)<k$ then
    $X\cap Y\in\CT$.
  \end{enumerate}
\end{lem}

\begin{proof}
 To prove (1), just note that if $Y\not\in\CT$ then $\bar Y\in\CT$ by
 \ref{li:t1}. But as $X\cap \bar Y=\emptyset$, this contradicts
 \ref{li:t2} (with $X_1=X_2=X$ and $X_3=\bar Y$).

 To prove (2), note that if $X\cap Y\not\in\CT$ then $\bar{X\cap
   Y}\in\CT$, and again this contradicts \ref{li:t2}, because
 \[
 X\cap Y\cap \bar{(X\cap Y)}=\emptyset.
 \]
\end{proof}

\begin{rem}
  The reader may wonder why in \ref{li:t2} we require the intersection
  of \emph{three} sets in $\CT$ to be nonempty. Why not the
  intersection of seventeen sets or just two sets in $\CT$?

  We need three sets to guarantee the important property of
  Lemma~\ref{lem:tangle-closure}(2), which may be viewed as a weak
  form of closure of a tangle under intersections.

  However, requiring the intersection
  of three sets in $\CT$ to be nonempty is sufficient for all
  arguments, so there is no reason to require more.
  \uend
\end{rem}

\begin{rem}
  There is a certain similarity between tangles and
  \emph{ultrafilters}, that is, families $\CF$ of nonempty subsets of a
  (usually infinite) set ${\U}$ that are closed under extensions and
  finite intersections. Certainly, this is a fairly superficial
  similarity. But I do feel that the way we view tangles as describing
  ``regions'' of a connectivity system, that is, new somewhat blurry structures
  derived from from the original system, is reminiscent of the
  construction of ultrapowers in model theory.

  Incidentally, tangle axiom \ref{li:t3} corresponds to the
  ultrafilters being \emph{non-principal}, that is, not just families
  of all sets that contain one specific element of the universe.
  \uend
\end{rem}

Let us now give an example that shows how to prove the non-existence
of tangles (of a certain order).

\begin{exa}\label{exa:cycle}
  Let $G$ be a cycle of length $n$. We claim that there is
  no $\mu_G$-tangle of order greater than $2$.  

  To prove this, suppose for contradiction that $\CT$ is a
  $\mu_G$-tangle of order at least $3$.  We call a subset
  $X\subseteq V(C)$ a \emph{segment} if it induces a path in
  $C$. Observe that for every segment $X$ we have $\mu_G(X)\le 2$ and
  thus either $X\in\CT$ or $\bar X\in\CT$. Let us assume that we have
  fixed an orientation of the cycle $C$ and orient the segments
  accordingly. The \emph{first half} and the \emph{second half} of a
  segment $X$ of cardinality $|X|\ge 2$ are the subsegments of cardinalities
  $\ceil{X/2}$ and $\floor{X/2}$, respectively, defined in the obvious
  way. We shall define a sequence $X_1,X_2,\ldots$ of segments such
  that for all $i$ we have $X_i\in\CT$, and $X_{i+1}$ is either the
  first or the second half of $X_i$. We continue the construction
  until $|X_i|=1$. Then $X_i\in\CT$ contradicts \ref{li:t3}.

  To start the construction, we let $X$ be an arbitrary segment of
  cardinality $\ceil{n/2}$. If $X\in\CT$ we let $X_1:=X$ and otherwise we let
  $X_1:=\bar X$. Now suppose we have defined $X_1,\ldots,X_i$ and
  $|X_i|\ge 2$. We let $Y$ be the first half of $X_i$. If $Y\in\CT$ we
  let $X_{i+1}:=Y$. Otherwise, $\bar Y\in\CT$, and we let
  $X_{i+1}:=X_i\cap\bar Y$, that is, $X_{i+1}$ is the second half of
  $X_i$. By Lemma~\ref{lem:tangle-closure}(2) we have $X_{i+1}\in\CT$.

  Again, essentially the same arguments show that $G$ has no
  $\rho_G$-tangle or $\kappa_G$-tangle of order greater than $2$.
  \uend
\end{exa}

\begin{lem}\label{lem:smallset}
  Let $\CT$ be a $\kappa$-tangle of order $k$, and let $X\subseteq {\U}$
  such that $\kappa(X')<k$ for all $X'\subseteq X$. Then $\bar X\in\CT$. 
\end{lem}

\begin{proof}
  We prove that $\bar X'\in\CT$ for all $X'\subseteq X$ by induction
  on $|X'|$. The assumption $\kappa(X')<k$ implies that either
  $X'\in\CT$ or $\bar X'\in\CT$.

  For the base step,  note that if $|X'|=0$ then $\bar X'\in\CT$ by
  \ref{li:t2}.

  For the inductive step, let $|X'|>0$ and $x\in X'$. Then
  $\bar{X'\setminus\{x\}}\in\CT$ by the induction hypothesis and
  $\bar{\{x\}}\in\CT$ by \ref{li:t3}. Thus $\bar X'=\bar{X\setminus\{x\}}\cap \bar{\{x\}}\in\CT$ by Lemma~\ref{lem:tangle-closure}(2).
\end{proof}

\begin{cor}\label{cor:smallset1}
   Let $\CT$ be a $\kappa$-tangle of order $k$ and $X\subseteq {\U}$ such
   that $|X|<k/\val(\kappa)$. Then $\bar X\in\CT$.
\end{cor}

\begin{cor}\label{cor:smallset2}
   Let $G$ be a graph and $\CT$ be a $\kappa_G$-tangle of order
   $k$. Let $X\subseteq E(G)$ such that $|V(X)|<k$. Then $\bar X\in\CT$.
\end{cor}

\subsection{Extensions, Truncations, and Separations}

Let
$\CT,\CT'$ be $\kappa$-tangles. If $\CT'\subseteq\CT$, we say that
$\CT$ is an \emph{extension} of $\CT'$ and $\CT'$ a \emph{truncation}
of $\CT$. The tangles $\CT$ and $\CT'$
are \emph{incomparable} (we write $\CT\bot\CT'$) if neither is an
extension of the other. 
A tangle is \emph{maximal} if it has no proper
extension. 
The \emph{truncation} 
of $\CT$ to order $k\le\ord(\CT)$ is the set
$
\{X\in\CT\mid\kappa(X)<k\},
$
which is obviously a tangle of order $k$. Observe that if $\CT$ is
an extension of $\CT'$, then $\ord(\CT')\le\ord(\CT)$, and $\CT'$ is
the truncation of $\CT$ to order $\ord(\CT')$. 

\begin{rem}\label{rem:tangle-order1}
There is a small technical issue that one needs to be aware of, but
that never causes any real problems: if we view tangles as families of
sets, then their order is not always well-defined. Indeed, if there is
no set $X$ of order $\kappa(X)=k-1$, then a tangle of order $k$
contains exactly the same sets as its truncation to order $k-1$. In such a situation, we
have to explicitly annotate a tangle with its order, formally viewing a
tangle as a pair $(\CT,k)$ where $\CT\subseteq 2^{\U}$ and $k\ge 0$. We always view a
tangle of order $k$ and its truncation to order $k-1$ as distinct
tangles, even if they contain exactly the same sets.
\uend
\end{rem}

Let $\kappa$ be a connectivity function on a set ${\U}$, and let
$\CT,\CT'$ be $\kappa$-tangles. A \emph{$(\CT,\CT')$-separation} is a set $X\subseteq {\U}$ such that
$X\in\CT$ and $\bar X\in\CT'$. Obviously, if $X$ is a
$(\CT,\CT')$-separation then $\bar X$ is a
$(\CT',\CT)$-separation. Observe that there is a
$(\CT,\CT')$-separation if and only if $\CT$ and $\CT'$ are
incomparable. The \emph{order} of a $(\CT,\CT')$-separation $X$ is
$\kappa(X)$. A $(\CT,\CT')$-separation is \emph{minimum} if its
order is minimum.

\begin{lem}\label{lem:lm-tansep}
  Let $\kappa$ be a connectivity function on a set ${\U}$, and let
  $\CT,\CT'$ be incomparable tangles. Then there is a (unique) minimum
  $(\CT,\CT')$-separation $X$ such that
  $X\subseteq X'$ for all minimum $(\CT,\CT')$-separations
  $X'$.

  We call $X$ the \emph{leftmost minimum $(\CT,\CT')$-separation}.
\end{lem}

\begin{proof}
    Let $X$ be a minimum $(\CT,\CT')$-separation of minimum cardinality $|X|$,
  and let $X'$ be another minimum $(\CT,\CT')$-separation. We shall
  prove that $X\subseteq X'$. 

  Let $k:=\kappa(X)=\kappa(X')<\min\{\ord(\CT),\ord(\CT')\}$. We claim
  that
  \begin{equation}
    \label{eq:cantan1}
    \kappa(X\cup X')\ge k.
  \end{equation}
  Suppose for contradiction that $\kappa(X\cup X')<k$. Then $X\cup
  X'\in\CT$ by Lemma~\ref{lem:tangle-closure}(1). Furthermore,
  $\bar{X\cup X'}=\bar X\cap\bar X'\in\CT'$ by Lemma~\ref{lem:tangle-closure}(2).
  Thus $X\cup X'$ is a
  $(\CT,\CT')$-separation of order less than $k$. This
  contradicts the minimality of $k=\kappa(X)$ and proves \eqref{eq:cantan1}.

  By submodularity, 
  \begin{equation}
    \label{eq:cantan2}
    \kappa(X\cap X')\le k.
  \end{equation}
  Then $X\cap X'\in\CT$ by Lemma~\ref{lem:tangle-closure}(2) and $\bar
  X\cap\bar X'=\bar X\cup\bar X'\in\CT'$ by by
  Lemma~\ref{lem:tangle-closure}(1). Thus $X\cap X'$ is a
  $(\CT,\CT')$-separation. By the minimality of $k$, we have
  $\kappa(X\cap X')=k$, and by the minimality of $|X|$ we have
  $|X|\le|X\cap X'|$. This implies $X=X\cap X'$ and thus $X\subseteq X'$.
\end{proof}

\subsection{Covers}

A \emph{cover} of a $\kappa$-tangle $\CT$ is a set $S\subseteq {\U}$ such
that $S\cap X\neq\emptyset$ for all $X\in\CT$. 

\begin{lem}\label{lem:cover}%
  Every tangle of order $k$ has a cover of cardinality at most $k$.
\end{lem}

\begin{proof}
  Let $\CT$ be a $\kappa$-tangle of order $k$.
  By induction on $i\ge 0$ we construct sets $S_i$ of cardinality $|S_i|\le i$ such that for all
  $X\in\CT$, if $S_i\cap X=\emptyset$ then
  $\kappa(X)\ge i$. Then $S_k$ is a cover
  of $\CT$.

  We let $S_0:=\emptyset$. For the inductive step, suppose that $S_i$
  is defined. If $X\cap S_i\neq\emptyset$ for all $X\in\CT$ with
  $\kappa(X)< i+1$, we let $S_{i+1}:=S_i$. Otherwise, let $X\in\CT$  such that
  \begin{eroman}
  \item $S_i\cap X=\emptyset$;
  \item subject to (i), $\kappa(X)$ is minimum;
  \item subject to (i) and (ii), $|X|$ is minimum.
  \end{eroman}
  By the induction hypothesis and our assumption that there be some
  $X'\in\CT$ such that $\kappa(X')< i+1$ and $X'\cap S_i=\emptyset$, we
  have $\kappa(X)=i$.  Let $x\in X$ and $S_{i+1}:=S_i\cup\{x\}$. 

  Let $Y\in\CT$ with $Y\cap S_{i+1}=\emptyset$. Suppose for
  contradiction that $\kappa(Y)< i+1$. If
  $\kappa(X\cap Y)\le i=\kappa(X)$, then $X\cap Y\in\CT$, and as
  $X\cap Y\subseteq X\setminus\{x\}\subset X$, this contradicts (ii) or (iii). Thus
  $\kappa(X\cap Y)>\kappa(X)$ and, by submodularity,
  $\kappa(X\cup Y)<\kappa(Y)\le i$. However,
  $S_i\cap (X\cup Y)=\emptyset$, so this contradicts the induction hypothesis.
\end{proof}

\begin{lem}
  Let $S$ be a cover of a $\kappa$-tangle $\CT$ of order $k$. Then
  $|S|\ge k/\val(\kappa)$.
\end{lem}

\begin{proof}
  If $|S|<k/\val(\kappa)$, then $\bar S\in\CT$ by
  Corollary~\ref{cor:smallset1}. As $S\cap\bar S=\emptyset$, this
  contradicts $S$ being a cover of $\CT$.
\end{proof}

The following theorem is a generalisation of a result for graphs due to \cite{ree97}.
A cover $S$ of a tangle $\CT$ is \emph{minimum} if its cardinality
$|S|$ is minimal.

\begin{theo}\label{theo:cover-wl}
  Let $\CT$ be a $\kappa$-tangle of order $k$, and let $S$ be a
  minimum cover of $\CT$. Then $S$ is well-linked. 
\end{theo}

\begin{proof}
  Suppose for contradiction that $S$ is not well linked. Then there is
  a set $X$ such that $\kappa(X)<\min\{|S\cap X|,|S\setminus
  X|\}$. As $\kappa(X)<|S|\le k$, either $X\in\CT$ or $\bar
  X\in\CT$. Without loss of generality we assume that $X\in\CT$.

  By Lemma~\ref{lem:wl1} there is a set $Y$ that is free in $\bar X$
  such that $\pi_X(Y)=\kappa(X)$. We claim that $(S\cap X)\cup Y$ is a cover of
  $\CT$. 
  To see this, let $X'\in\CT$ such that $S\cap X\cap X'=\emptyset$. Then $X\cap
  X'\not\in\CT$, because $S$ is a cover of $\CT$. By
  Lemma~\ref{lem:tangle-closure}(2), this implies
  $\kappa(X\cap X')\ge k>\kappa(X')$. By symmetry and submodularity, 
  \[
  \kappa(\bar X\cap\bar X')=\kappa(X\cup
  X')<\kappa(X).
  \]
  If $X'\cap Y=\emptyset$, then $Y\subseteq\bar X\cap\bar
  X'\subseteq\bar X$ and thus $\kappa(X)=\pi_X(Y)\le\kappa(\bar X\cap\bar
  X')<\kappa(X)$, which is a contradiction. Hence $X'\cap
  Y\neq\emptyset$. This proves that $(S\cap X)\cup Y$ is a cover of $\CT$.

  However, 
  \[
  |(S\cap X)\cup Y|=|S\cap X|+|Y|\le|S\cap X|+\kappa(X)<|S\cap
  X|+|S\setminus X|=|S|, 
  \]
  where $|Y|\le\pi_{\bar X}(Y)\le\kappa(\bar X)=\kappa(X)$ holds
  because $Y$ is free in $\bar X$. This contradicts the minimality of
  $|S|$.
\end{proof}

Combining this lemma with Example~\ref{exa:wl-tangle}, we thus have.

\begin{cor}
  \begin{enumerate}
  \item If $\kappa$ has a well-linked set of cardinality at least $3k$, then
    there is a $\kappa$-tangle of order $k$.
  \item If there is a $\kappa$-tangle of order $k$, then $\kappa$ has
    a well-linked set of cardinality at least $k/\val(\kappa)$.
  \end{enumerate}
 \end{cor}
 
 Lemma~\ref{lem:cover} has the following powerful generalisation,
 which sometimes allows us to control a tangle of order $k$ by a set
 of cardinality bounded in terms of $k$. A \emph{triple cover} of a
 $\kappa$-tangle $\CT$ is a set $S\subseteq {\U}$ such that
 $S\cap X_1\cap X_2\cap X_3\neq\emptyset$ for all $X_1,X_2,X_3\in\CT$.

\begin{theo}[\cite{groschwe15b}]\label{lem:triple-cover}
  There is a function $f:\NN\to\NN$ such that every tangle order $k$ has a triple
  cover of cardinality at most $f(k)$.
\end{theo}

We omit the proof.

\subsection{Tangles in Graphs}

Robertson and Seymour originally defined tangles for graphs rather
than general connectivity functions. We will see that a tangle of a
graph $G$, which we will call a \emph{$G$-tangle}, is almost the same as
$\kappa_G$-tangle, up to small issues regarding isolated
vertices, isolated edges, and pendant edges.

A \emph{separation} of a graph $G$ is a pair $(A,B)$ of subgraphs of
$G$ such that $A\cup B:=(V(A)\cup V(B),E(A)\cup E(B))=G$ and
$E(A)\cap E(B)=\emptyset$. The \emph{order} of the separation $(A,B)$
is $\ord(A,B):=|V(A)\cap V(B)|$. Note that a separation $(A,B)$ is essentially,
but not exactly, the same as the partition $(E(A),E(B))$ of $E(G)$. A
separation $(A,B)$ is \emph{trivial} if $A=G$ or $B=G$.

A \emph{$G$-tangle} of order $k$ is a family $\CS$ of separations of
$G$ satisfying the following conditions.
  \begin{nlist}{GT}
  \setcounter{nlistcounter}{-1}
  \item\label{li:gt0}
    The order of all separations $(A,B)\in\CS$ is less than $k$.
  \item\label{li:gt1}
    For all separations $(A,B)$ of $G$ of order less than $k$, either
    $(A,B)\in\CS$ or $(B,A)\in\CS$.
  \item\label{li:gt2}
    If $(A_1,B_1), (A_2,B_2),(A_3,B_3)\in\CS$ then $A_1\cup A_2\cup
    A_3\neq G$.
  \item\label{li:gt3}
    $V(A)\neq V(G)$ for all $(A,B)\in\CS$.
\end{nlist}

\begin{exa}\label{exa:graph-tangle}
  Every graph $G$ with $E(G)\neq\emptyset$ has a $G$-tangle of order
  $2$. 

  Indeed, let $e\in E(G)$ and 
  \[
  \CT_e:=\{(A,B)\mid(A,B)\text{ separation of $G$ of order less than
    $2$ with }e\in E(B)\}.
  \]
  We claim that $\CT$ is a tangle of order $2$. It trivially satisfies
  \ref{li:gt0}, \ref{li:gt1}, and \ref{li:gt2}. It satisfies
  \ref{li:gt3}, because if $(A,B)\in \CT_e$ then $|V(B)|\ge 2$ and
  $|V(A)\cap V(B)|\le 1$.

  This illustrates the difference between $G$-tangles and
  $\kappa_G$-tangle, because if $G$ is a path of length $1$ then it does not even have
  a $\kappa_G$ tangle of order $1$, and if $G$ is a path of length
  $2$, it has a
  $\kappa_G$ tangle of order $1$, but not a $\kappa_G$-tangle of order
  $2$.
  \uend
\end{exa}

We call an edge of a graph \emph{isolated} if both of its endvertices
have degree $1$. We call an edge \emph{pendant} if it
is not isolated and has one endvertex of degree $1$. 

\begin{prop}\label{prop:graph-tangle}
  Let $G$ be a graph and $k\ge 0$.
  \begin{enumerate}
  \item If $\CT$ is a $\kappa_G$-tangle of order $k$, then 
    \[
    \CS:=\big\{(A,B)\bigmid(A,B)\text{ separation of $G$ of order $<k$
      with
    }E(B)\in\CT\big\}
    \]
    is a $G$-tangle of order $k$.
  \item If $\CS$ is a $G$-tangle of order $k$, then
    \[
    \CT:=\big\{ E(B)\bigmid (A,B)\in\CS\big\}
    \]
    is a $\kappa_G$-tangle of order $k$, unless 
    \begin{eroman}
    \item either $k=1$ and there is an isolated vertex $v\in V(G)$ 
      such that $\CS$ is the set of all separations $(A,B)$ of order
      $0$ with with $v\in V(B)\setminus V(A)$,
    \item or $k=1$ and there is an isolated edge $e\in E(G)$ such that
      and $\CS$ is the set of all separations $(A,B)$ of order $0$
      with $e\in E(B)$,
    \item or $k=2$ and there is an isolated or pendant edge
      $e=vw\in E(G)$ and $\CS$ is the set of all separations $(A,B)$
      of order at most $1$ with $e\in E(B)$.
\end{eroman}
  \end{enumerate}
\end{prop}

I omit the simple, but tedious proof.

A \emph{star} is a connected graph in which at most $1$ vertex has
degree greater than $1$. Note that we admit degenerate stars
consisting of a single vertex or a single edge.

\begin{cor}\label{cor:graph-tangle}
  Let $G$ be a graph that has a $G$-tangle of order $k$. Then $G$ has
 a $\kappa_G$-tangle of order $k$, unless $k=1$ and $G$ only has isolated edges
  or $k=2$ and all connected components of $G$ are stars.
\end{cor}

We now turn to a different characterisation of graph tangles due to Reed~\cite{ree97}.
Let $G$ be a graph. We say that subgraphs $H_1,\ldots,H_m\subseteq G$
\emph{touch} if there is a vertex $v\in\bigcap_{i=1}^m V(H_i)$ or an edge
$e\in E(G)$ such that each $H_i$ contains at least one endvertex of
$e$. A family $\CH$ of subgraphs of $G$ \emph{touches pairwise} if all
$H_1,H_2\in\CH$ touch, and it \emph{touches triplewise} if all
$H_1,H_2,H_3\in\CH$ touch. A \emph{vertex cover} (or \emph{hitting set})
for $\CH$ is a set $S\subseteq V(G)$ such that $S\cap
V(H)\neq\emptyset$ for all $H\in\CH$.

\begin{theo}[\cite{ree97}]\label{theo:reed}
  A graph $G$ has a $G$-tangle of order $k$ if and only if there is a
  family $\CH$ of connected subgraphs of $G$ that touches triplewise
  and has no vertex cover of cardinality less than $k$.
\end{theo}

In fact, Reed~\cite{ree97} defines a tangle of a graph $G$ to be a family
$\CH$  of connected subgraphs of $G$ that touches triplewise and its
order to be the cardinality of a minimum vertex cover.

\begin{proof}
  For the forward direction, let $\CT$ be a $G$-tangle of order $k$.

  \begin{claim}
    For every set $S\subseteq V(G)$ of cardinality
    $|S|<k$ there is a (nonempty) connected component $C_S$ of $G\setminus S$ such
  that for all separations $(A,B)$ of $G$ with
  $V(A)\cap V(B)\subseteq S$ we have $(A,B)\in\CT\iff C_S\subseteq
  B$.

  \proof
  Let $C_1,\ldots,C_m$ be the set of all connected
  components of $G\setminus S$.  For every $I\subseteq [m]$, we define
  a separation $(A_I,B_I)$ of $G$ as follows:
 \begin{itemize}
 \item $V(B_I):=S\cup\bigcup_{i\in I}V(C_i)$ and
   $E(B_I):=\bigcup_{i\in I}E\big(V(C_i),V(C_i)\cup S\big)$;
 \item $V(A_I):=S\cup\bigcup_{i\in [m]\setminus I}V(C_i)$ and
   $E(A_I):=E(G)\setminus E(B_I)=\bigcup_{i\in[m]\setminus I}E\big(V(C_i)\cup S,V(C_i)\cup S\big)$.
 \end{itemize}
 Note that $V(A_I)\cap V(B_I)=S$ and thus $\ord(A_I,B_I)<k$. Thus for
 all $I$, either $(A_I,B_I)\in\CT$ or $(B_I,A_I)\in\CT$.
 By \ref{li:gt2}, if $(A_I,B_I),(A_J,B_J)\in\CT$ then
 $(A_{I\cap J},B_{I\cap J})\in\CT$, because 
 \[
 A_I\cup A_J\cup B_{I\cap J}=G.
 \]
 Moreover, by \ref{li:gt3} we have
 $(A_\emptyset,B_\emptyset)\not\in\CT$.
 
 Now a straightforward
 interval-halving argument (similar to the one used in
 Example~\ref{exa:cycle}) shows that there is an $i\in[m]$ such that
 $(A_{\{i\}},B_{\{i\}})\in\CT$. We let $C_S:=C_i$.
 \uend
 \end{claim}

  We let
  \[
  \CH:=\{C_S\mid S\subseteq V(G)\text{ with }|S|<k\}.
  \]
  $\CH$ has no vertex cover of cardinality less than $k$, because if
  $S\subseteq V(G)$ with $|S|<k$ then $S\cap V(C_S)=\emptyset$. It
  remains to prove that $\CH$ touches triplewise. For $i=1,2,3$, let
  $H_i\in\CH$ and $S_i\subseteq V(G)$ with $|S_i|<k$ such that
  $H_i=C_{S_i}$. Let 
  \[
  B_i:=\big(V(C_i)\cup S_i, E(V(C_i),V(C_i)\cup S_i)\big)
  \]
  and $A_i:=\big(V(G)\setminus V(C_i), E(G)\setminus
  E(B_i)\big)$. By Claim~1, $(A_i,B_i)\in\CT$. Hence $A_1\cup A_2\cup
  A_3\neq G$ by \ref{li:gt2}. If $V(A_i)\cup V(A_2)\cup V(A_3)\neq
  V(G)$ then $V(C_1)\cap V(C_2)\cap V(C_3)\neq \emptyset$ and hence
  $C_1,C_2,C_3$ touch. Otherwise, $E(A_i)\cup E(A_2)\cup E(A_3)\neq
  E(G)$. Hence there is an edge $e\in E(B_1)\cap E(B_2)\cap
  E(B_3)$. As every edge in $E(B_i)$ has an endvertex in $V(C_i)$, this
  shows that $C_1,C_2,C_3$ touch.

  \medskip
  For the backward direction, let $\CH$ be a family of connected
  subgraphs of $G$ that touches triplewise
  and has no vertex cover of cardinality less than $k$.
  We let $\CT$ be the set of all separations of $G$ of order less
  than $k$ such that $H\subseteq B\setminus V(A)$ for some $H\in\CH$.
  
  $\CT$ trivially satisfies \ref{li:gt0}. To see that it satisfies
  \ref{li:gt1}, let $(A,B)$ be a separation of $G$ of order less than
  $k$ and $S:=V(A)\cap V(B)$. Then $S$ is no a vertex cover of $\CH$,
  and hence there is a $H\in\CH$ such that $S\cap V(H)=\emptyset$. As
  $H$ is connected, either $H\subseteq B\setminus V(A)$ or $H\subseteq
  A\setminus V(B)$ and thus
  either $(A,B)\in\CT$ or $(B,A)\in\CT$.

  To see that $\CT$ satisfies \ref{li:gt2}, let $(A_i,B_i)\in\CT$
  and $H_i\in\CH$ such that $H_i\subseteq B_i\setminus V(A_i)$, for
  $i=1,2,3$. If $V(H_1)\cap V(H_2)\cap V(H_3)\neq\emptyset$ then
  $V(A_1)\cup V(A_2)\cup V(A_3)\neq V(G)$. If there is an edge $e$
  that has an endvertex in $V(H_i)$ for $i=1,2,3$, then $e\in E(B_i)$
  and thus $E(A_1)\cup E(A_2)\cup E(A_3)\neq E(G)$.

  Finally, $\CT$ satisfies \ref{li:t3}, because if $(A,B)\in\CT$ and
  $H\in\CH$ with $H\subseteq B\setminus V(A)$ then $V(A)\subseteq
  V(G)\setminus V(H)\neq V(G)$.
\end{proof}

$G$-tangles of orders $1$, $2$, $3$ are in one-to-one correspondence
with the connected, biconnected, and triconnected components of a
graph. A proof of this fact can be found in \cite{gro16a}. Let us call
a graph $G$ \emph{quasi-4-connected} if it is 3-connected and for
every separation $(A,B)$ of order $3$, either
$|V(A)\setminus V(B)|\le 1$ or $|V(B)\setminus V(A)|\le 1$. In
\cite{gro16b}, I proved that every graph $G$ has a tree decomposition into
quasi-4-connected components and that these quasi-4-connected
components correspond to the $G$-tangles of order $4$.

Let me mention that recently Carmesin, Diestel, Hundertmark, and Stein
\cite{cardiehun+14} gave a decomposition of graphs into a different form of
$k$-connected regions (called $(k-1)$-blocks there), which is based on
$k$-linked sets (see Exercise~\ref{exe:k-linked}).

\section{The Duality Theorem}
\label{sec:duality}

Let $\kappa$ be a connectivity function on ${\U}$ and $\CA\subseteq 2^{\U}$.
A pre-decomposition $(T,\gamma)$ of $\kappa$ is \emph{over} $\CA$
if all its atoms are in $\CA$, that is, $\gamma(t)\in\CA$ for all
$t\in L(T)$.  A $\kappa$-tangle $\CT$ \emph{avoids} $\CA$ if
$\CT\cap\CA=\emptyset$. Note that, by \ref{li:t3}, every tangle avoids
the set 
\[
\Sing{{\U}}:=\big\{\{x\}\bigmid x\in {\U}\big\}
\] 
of all singletons. %

\begin{theo}[Duality Theorem, \cite{gm10}]\label{theo:duality}
  Let $\kappa$ be a connectivity function on ${\U}$. Let
  $\CA\subseteq 2^{\U}$ such that $\CA$ is closed under taking subsets
  and $\Sing\U\subseteq\CA$.  Then there is a decomposition of width
  less than $k$ over $\CA$ if and only if there is no $\kappa$-tangle
  of order $k$ that avoids $\CA$.
\end{theo}

\begin{proof}
  For the forward direction, let $(T,\gamma)$ be decomposition of
  $\kappa$ over $\CA$ of width less than $k$. Suppose for
  contradiction that $\CT$ is a $\kappa$-tangle of order $k$ that
  avoids $\CA$. For every edge $st\in E(T)$, we orient $st$ towards
  $t$ if $\gamma(s,t)\in\CT$ and towards $s$ if
  $\bar{\gamma(s,t)}=\gamma(t,s)\in\CT$. As $\CT$ is a tangle of
  order $k$ and $\kappa(\gamma(s,t))<k$ for all
  $(s,t)\in\vec E(T)$, every edge gets an orientation. As $T$ is a
  tree, there is a node $s\in V(T)$ such that all edges incident with
  $s$ are oriented towards $s$. If $s$ is a leaf, then $\gamma(s)\in\CT$
  and thus $\gamma(s)\not\in\CA$, because
  $\CT$ avoids $\CA$. This contradicts $(T,\gamma)$ being a
  decomposition over $\CA$. Thus $s$ is an internal node, say,
  with neighbours $t_1,t_2,t_3$. Then $\gamma(t_i,s)\in\CT$ and
  thus
  $\gamma(t_1,s)\cap \gamma(t_2,s)\cap
  \gamma(t_3,s)\neq\emptyset$.
  This implies
  $\gamma(s,t_1)\cup \gamma(s,t_2)\cup \gamma(s,t_3)\neq {\U}$,
  which contradicts $(T,\gamma)$ being a decomposition.

  \medskip For the proof of the backward direction, suppose that there
  is no $\kappa$-tangle $\CT$ of order $k$ that avoids $\CA$. We shall
  prove that there is a pre-decomposition of $\kappa$ over $\CA$ of
  width less than $k$. By the Exactness Lemma (for Undirected
  Decompositions, Corollary~\ref{cor:exact}), and since $\CA$ is
  closed under taking subsets, we obtain a decomposition of $\kappa$
  over $\CA$ of width at less than $k$.

  The proof is by induction on the
  number of sets $X\subseteq {\U}$ with $\kappa(X)<k$ such that
  neither $X\in\CA$ nor $\bar X\in\CA$.

  For the base step, let us assume that for all
  $X\subseteq {\U}$ with $\kappa(X)<k$ either $X\in\CA$ or $\bar
  X\in\CA$. Let
  \[
  \CY=\{\bar X\mid X\in\CA\text{ with }\kappa(X)<k\}.
  \]
  Then
  $\CY$ trivially satisfies the tangle axiom \ref{li:t0}. It satisfies
  \ref{li:t1} by our assumption that either $X\in\CA$ or $\bar
  X\in\CA$ for all $X\subseteq {\U}$ with
  $\kappa(X)<k$.

  If
  $\CY$ violates \ref{li:t2}, then there are sets $Y_1,Y_2,Y_3\in\CY$
  with $Y_1\cap Y_2\cap Y_3=\emptyset$. We let $T$ be the tree
  with vertex set $V(T)=\{s,t_1,t_2,t_3\}$ and edge set
  $\{st_1,st_2,st_3\}$, and we define $\gamma(t_i,s):=Y_i$ and
  $\gamma(s,t_i):=\bar Y_i\in\CA$. Then $(T,\gamma)$ is a
  pre-decomposition of $\kappa$ over $\CA$ of width less than $k$.
  
  So let us assume that $\CY$ satisfies \ref{li:t2}. Then it must
  violate \ref{li:t3}, because there is no
  tangle of order $k$ that avoids $\CA$. Thus for some
  $x\in {\U}$ we have $\{x\}\in\CY$ and thus $\bar{\{x\}}\in\CA$ and
  $\kappa(\bar{\{x\}})=\kappa(\{x\})<k$. Note that $\{x\}\in\CA$,
  because $\Sing\U\subseteq\CA$. We
  let $T$ be the tree consisting of just one edge $st$, and we define
  $\gamma$ by $\gamma(s,t)=\{x\}$,
  $\gamma(s,t)=\bar{\{x\}}$. Then $(T,\gamma)$ is a
  pre-decomposition of $\kappa$ over $\CA$ of width less than $k$.

  \medskip For the inductive step, let $X\subseteq {\U}$ such that
  $\kappa(X)<k$ and neither $X\in\CA$ nor $\bar X\in\CA$ and such that
  $|X|$ is minimum subject to these conditions. Let
  $\CA^1:=\CA\cup 2^X$ and
  $\CA^2:=\CA\cup 2^{\bar X}$. Then by
  the inductive hypothesis, for $i=1,2$ there is a pre-decomposition
  $(T^i,\gamma^i)$ of $\kappa$ over $\CA^i$ of width less than $k$. 
  If there is no leaf $t^i$ of $T^i$ with $\gamma(t^i)\not\in\CA$, then
  $(T^i,\gamma^i)$ is a pre-decomposition of $\kappa$ over $\CA$ of
  width less than $k$, and we are done. So let us assume that for
  $i=1,2$ there is a leaf $t^i$ of $T^i$ with $\gamma(t^i)\not\in\CA$.

  Consider $(T^1,\gamma^1)$. By the Exactness Lemma and since
  $\CA^1$ is closed under taking subgraphs, we
  may assume that $(T^1,\gamma^1)$ is exact. This implies that the
  atoms $\gamma^1(t)$ for the leaves $t\in L(T^1)$ are mutually
  disjoint. Let $X':=\gamma^1(t^1)\not\in\CA$. Then $X'\subseteq X$ and $\bar
  X\subseteq\bar X'$, and as $\bar X\not\in\CA$ and $\CA$ is closed
  under taking subsets, it follows that $\bar
  X'\not\in\CA$. By the minimality of $|X|$, this implies
  $X'=X$. Furthermore, as the decomposition $(T^1,\gamma^1)$ is
  exact, $t^1$ is the only leaf of $T^1$ with
  $\gamma^1(t^1)=X$, and for all other leaves $t$ we have
  $\gamma^1(t)\in\CA$. Let $s^1$ be the neighbour of $t^1\in T^1$.

  Now consider $(T^2,\gamma^2)$. Let $t^2_1,\ldots,t^2_m$ be an
  enumeration of all leaves $t$ of $T^2$ with
  $\gamma^2(t)\not\in\CA$. Then $\gamma(t^2_i)\subseteq \bar X$ for all
  $i\in[m]$. Let $s^2_i$ be the neighbour of $t^2_i$ in $T^2$. Without
  loss of generality we may assume that
  $\gamma^2(t^2_i)=\gamma^2(s_i^2,t_i^2)=\bar X$ and
  $\gamma^2(t_i^2,s_i^2)=X$, because increasing a set $\gamma^2(t)$
  for a leaf $t$
  preserves the property of being a pre-decomposition. 

  To construct a pre-decomposition $(T,\gamma)$ of $\kappa$
  over $\CA$, we take $m$ disjoint copies 
  \[
  (T^1_1,\gamma^1_1),\ldots,(T^1_m,\gamma^1_m) 
  \]
  of
  $(T^1,\gamma^1)$. For each node $t\in V(T^1)$, we denote its copy
  in $T^1_i$ by $t_i$. Then for every edge $st\in E(T^1)$ we have
  $\gamma^1_i(s_i,t_i)=\gamma^1(s,t)$. In particular,
  $\gamma^1_i(s^1_i,t^1_i)=\gamma^1(s^1,t^1)=X$. We let $T$ be
  the tree obtained from the disjoint union of
  $T^1_1,\ldots,T^1_m,T^2$ by deleting the nodes $t^1_i,t^2_i$ and
  adding edges $s^1_is^2_i$ for all $i\in[m]$. We define
  $\gamma:V(T)\to 2^{\U}$ by
  \[
  \gamma(s,t):=
  \begin{cases}
    X&\text{if }(s,t)=(s^1_i,s^2_i)\text{ for some }i\in[m],\\
    \bar X&\text{if }(s,t)=(s^2_i,s^1_i)\text{ for some }i\in[m],\\
    \gamma^1_i(s,t)&\text{if }st\in E(T^1_i),\\
    \gamma^2(s,t)&\text{if }st\in E(T^2).
  \end{cases}
  \]
  It is easy to see that $(T,\gamma)$ is a pre-decomposition
  $\kappa$ over $\CA$ of
  width less than $k$.
\end{proof}

The following corollary states that the branch width of a connectivity
function $\kappa$ equals the maximum order of a $\kappa$-tangle.

\begin{cor}\label{cor:duality}
 There is a $\kappa$-tangle of order $k$ if and only if $\bw(\kappa)\ge k$.
\end{cor}

\begin{proof}
  We apply the Duality Theorem with
  $\CA:=\Sing{{\U}}\cup\{\emptyset\}$. As we can eliminate ``empty leaves''
  in a decomposition by Lemma~\ref{lem:empty-leaves}, $\kappa$ has a
  branch decomposition of width $k$ if and only if it has a
  decomposition over $\CA$ of width $k$. Furthermore, by \ref{li:t2}
  and \ref{li:t3}, every
  $\kappa$-tangle avoids $\CA$.
\end{proof}

Recall that the \emph{branch width} of a graph $G$ is $\bw(\kappa_G)$.

\begin{cor}\label{cor:graph-duality}
  Let $k\ge 3$. Then for every graph $G$, there is a $G$-tangle of
  order $k$ if and only if the branch width of $G$ is at least $k$.
\end{cor}

This follows from Corollary~\ref{cor:duality} and
Corollary~\ref{cor:graph-tangle}. Note that the assertion of the
corollary fails for $k\le 2$, because, by
Example~\ref{exa:graph-tangle}, every graph $G$ with at least one edge
has a $G$-tangle of order $2$, but its branch width may be $0$.

\begin{rem}
  Recall the characterisation of $G$-tangles by triplewise touching
  families of connected subgraphs of $G$
  (Theorem~\ref{theo:reed}). Phrased using this characterisation, the
  Duality Theorem (Corollary~\ref{cor:graph-duality}, to be precise)
  says that, the branch width of a graph $G$ is the
  maximum $k$ such that there is a family $\CH$ of connected subgraphs
  of $G$ that touches triplewise and has no vertex cover of
  cardinality less than $k$ (provided the branch width is at least $k$).

  There is a similar duality for tree width, due to Seymour and
  Thomas~\cite{seytho93} (also see \cite{ree97}). A \emph{bramble of
    order $k$} of
  a graph $G$ is family $\CH$ of connected subgraphs
  of $G$ that touches pairwise and has no vertex cover of
  cardinality less than $k$. Then the tree width
  of a graph is the maximum $k$ such that $G$ has a bramble of order
  $k+1$.

In \cite{dieoum14a,dieoum14b}, Diestel and Oum develop a general duality theory for
width parameters like branch width (in the previous theorem) and tree width.
  \uend
\end{rem}

Recall
Theorem~\ref{theo:mu_vs_kappa}, stating that for every graph $G$ with at least
one vertex of degree $2$ we have
\[
\bw(\mu_G)\le\bw(\kappa_G)\le 2\bw(\mu_G).
\]
We have already proved the first inequality in
Section~\ref{sec:graphdec}. We use the Duality Theorem to prove the
second inequality.

\begin{proof}[Proof of the second inequality of
  Theorem~\ref{theo:mu_vs_kappa}]
  Let $G$ be a graph and $k\ge 0$.
  By the Duality Theorem, it suffices to prove that if there is a
  $\kappa_G$-tangle of order $2k+1$ then there is a $\mu_G$-tangle of
  order $k+1$. So let $\CT$ be a $\kappa_G$-tangle of order $2k+1$.
  We shall define a $\mu_G$-tangle $\CS$ of order $k+1$. 

  For every $X\subseteq V(G)$, we fix a minimum vertex cover $S_X$ of $E(X,\bar X)$
  in such a way that $S_X=S_{\bar X}$. 
  We let
  \[
  \CS:=\big\{ X\subseteq V(G)\bigmid \mu_G(X)\le k,E(X,X\cup
  S_X)\in\CT\big\}.
  \]
  We shall prove that $\CS$ is a $\mu_G$-tangle of order $k+1$. It trivially
  satisfies \ref{li:t0}. To see that it satisfies \ref{li:t1}, let
  $X\subseteq V(G)$ with $\mu_G(X)\le k$. Let $Y:=E(X,X\cup S_X)$. Then $\partial(Y)\subseteq
  S_X$ and thus $\kappa_G(Y)\le k$.  Therefore, either $Y\in \CT$ or
  $\bar Y\in\CT$. If $Y\in\CT$ then $X\in\CS$. So suppose that $\bar
  Y\in\CT$.
  We have
  \[
  \bar Y=E(\bar X,\bar X\cup S_X)\setminus E(\bar X\cap S_X,X\cap S_X).
  \]
  It follows that $E(\bar X,\bar X\cup S_X)\in\CT$, because
  $\kappa_G(E(\bar X,\bar X\cup S_X))\le|S_X|\le k<\ord(\CT)$ and
  $\bar Y\subseteq E(\bar X,\bar X\cup S_X)$.
  Thus $\bar X\in\CS$.

  To prove that $\CS$ satisfies \ref{li:t2}, let $X_0,X_1,X_2\in\CS$.
  Let $S_i:=S_{X_i}$ and $Y_i:=E(X_i,X_i\cup S_i)$. Then $Y_i\in\CT$. 
  
  Let $Y_i':=Y_i\setminus E(S_i,S_{i+1})$, where the sum is taken
  modulo $3$. Then $\partial(Y_i')\subseteq S_i\cup S_{i+1}$ and thus
  $\kappa_G(Y_i')\le 2k$. Furthermore, $\bar{E(S_i,S_{i+1})}\in\CT$ by
  Corollary~\ref{cor:smallset2}. Hence 
  \[
  Y_i'=Y_i\cap \bar{E(S_i,S_{i+1})}\in\CT
  \]
  by Lemma~\ref{lem:tangle-closure}(2).

  By \ref{li:t2} there is an edge $e=vw\in Y_0'\cap Y_1'\cap
  Y_2'$. Suppose that neither $v\in X_0\cap X_1\cap X_2$ nor $w\in
  X_0\cap X_1\cap X_2$. Say, $v\not\in X_i$ and $w\not\in X_j$. Then
  $i\neq j$, because $e\in Y_i=E(X_i,X_i\cup S_i)$ has at least one
  endvertex in $X_i$. Without loss of generality
  we assume that $j=i+1\bmod3$. As $e\in E(X_i,X_i\cup S_i)$ and
  $v\not\in X_i$ we
  have $v\in S_i$. Similarly, $w\in S_j=S_{i+1}$. Thus $e\in
  E(S_i,S_{i+1})$ and therefore $e\not\in Y_i'$, which is a
  contradiction.

  Finally, to prove \ref{li:t3}, let $X\in\CS$. Suppose for
  contradiction that $|X|=1$, say, $X=\{x\}$. Let $Y:=E(X,X\cup
  S_X)$. Then $Y\in\CT$. Observe that either
  $S_X=\emptyset$ or $S_X=\{x\}$
  or $S_X$ consists of a single vertex in $x'\in V(G)\setminus\{x\}$. In the first
  two cases, we have $Y=\emptyset$ and in the third
  case we have $Y=\{xx'\}$ and hence $|Y|=1$. Either way, this contradicts $Y\in\CT$.
\end{proof}

\section{The Canonical Decomposition Theorem}
\label{sec:can}

In this section, we shall prove that every connectivity system can be
decomposed into its maximal tangles in a canonical way. In view of the
correspondence of between $G$-tangles of order $2$ and the biconnected
components of a graph $G$ and $G$-tangles of order $3$ and the triconnected
components of $G$, this may be seen as a generalisation of the
decompositions of a graph into its biconnected and triconnected
components. 

Recall that a
$\kappa$-tangle $\CT$ is \emph{maximal} if it has no proper extension
and that we call a construction that associates a decomposition with every
connectivity system \emph{canonical} if every isomorphism between two
connectivity systems extends to an isomorphism between the
corresponding decompositions. We will discuss the form of
decomposition we will use later. Let us start with an example that
illustrates some of the issues arising.

\begin{figure}
  \centering
  \input{exadec}
  \caption{Graph $G$ of Example~\ref{exa:dec} and a decomposition}
  \label{fig:exadec}
\end{figure}
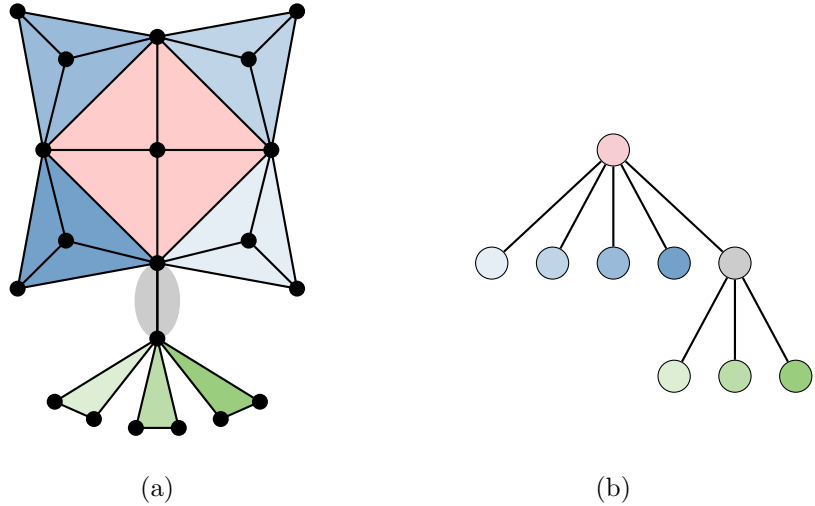

\begin{exa}\label{exa:dec}
  Consider the graph $G$ in Figure~\ref{fig:exadec}(a). The coloured
  regions correspond to the maximal $\kappa_{G}$-tangles. The order of
  of all four blue tangles is $4$, the order of the red tangle is $3$,
  the order of the three green tangles and the grey tangle is $2$.
  For example, the red tangle consist of all $Y\subseteq E(G)$ such
  that $\kappa_G(Y)\le 2$ and $V(Y)$ contains all vertices in the red region,
  including those on the boundary.

  Intuitively, a decomposition of $G$ ``into its maximal tangles''
  might look as indicated in  Figure~\ref{fig:exadec}(b), where the
  colour of a node indicates which tangle is associated with it.
  
  It is not clear how exactly this decomposition is defined. If it is
  supposed to partition the elements of the universe (that is,
  $E(G)$), what do we do with the edges on the boundary of two
  coloured regions describing the tangles? We will answer these
  questions soon. 

  Also note that the decomposition tree is not cubic,
  and if we want the decomposition tree to be canonical, there is no
  way to achieve this with a cubic tree. Specifically, in a canonical
  decomposition we cannot partition the three green tangles in any
  other way than into singletons, and for this we need a node of
  degree $4$: three outgoing edges to the nodes representing the green
  tangles, and one outgoing edge connecting it to the rest of the
  decomposition. 

  Now consider the graph $G'$ in Figure~\ref{fig:exadec2}(a), which is
  obtained from $G$ by deleting the middle vertex. Now there is no
  longer a tangle corresponding to the middle square, that is, the red
  tangle of $G$. (The other maximal tangles remain the same.)
  Nevertheless, the ``best'' decomposition, displayed in
  Figure~\ref{fig:exadec2}(b),  is still the same, except that the
  root node no longer corresponds to a maximal tangle. Indeed, it will
  be necessary to allow such nodes that do not correspond to any tangle
  in our decompositions. We will call such nodes ``hub nodes''.
 \uend
\end{exa}

\begin{figure}
  \centering
  \input{exadec2}
  \caption{Graph $G'$ of Example~\ref{exa:dec} and a decomposition}
  \label{fig:exadec2}
\end{figure}
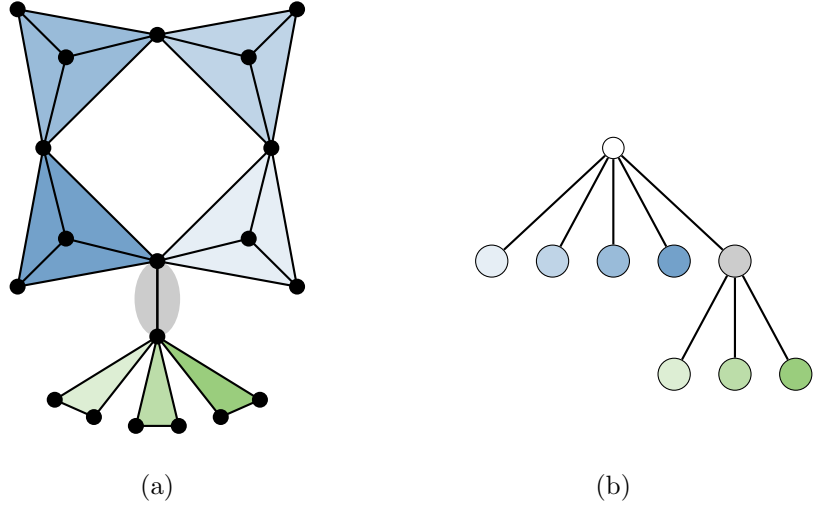

In \cite{gm10}, Robertson and Seymour proved that every graph has a
tree decomposition into parts corresponding to its maximal
tangles. Geelen, Gerards, and Whittle \cite{geegerwhi09} generalised
this to arbitrary connectivity systems. However, these decompositions
are not canonical.\footnote{Incidentally, Reed~\cite{ree97} calls
  Robertson and Seymour's decomposition ``canonical'', but he uses the
  term ``canonical'' with a different meaning. (It is not clear to me
  which.)} Carmesin, Diestel, Hamann, and
Hundertmark~\cite{cardiehamhun16} proved that every graph has a
canonical tree decomposition into parts corresponding to its maximal
tangles, and Hundertmark~\cite{hun11} (also see \cite{diehunlem16}) generalised this to arbitrary
connectivity systems. Our presentation of Hundermark's result
follows \cite{groschwe15a}.

\subsection{Tree Decomposition and Nested Separations}
\label{sec:tree1}

The type of decomposition we use here differs from the decompositions
introduced in Section~\ref{sec:dec} in two ways: the decomposition
trees are no longer cubic, and the pieces (or atoms) of the decomposition
are not only located at the leaves of the tree, but also at internal
nodes. Such decompositions are called \emph{tree decompositions},
which is a bit unfortunate, because tree decompositions of the connectivity function
$\kappa_G$ of a graph $G$ are not the same as the tree decompositions
of $G$ (introduced in Section~\ref{sec:graphdec}), whereas branch
decompositions of $\kappa_G$ are the same as branch decompositions of
$G$. But, in particular in lack of a better term, I do not want to
change the established terminology.

As usual, let $\kappa$ be a connectivity function on a set ${\U}$.

\begin{defn}
  A \emph{tree decomposition} of $\kappa$ is a pair $(T,\beta)$
  consisting of a tree $T$ and a function $\beta\colon V(T)\to 2^{\U}$
  such that the sets $\beta(t)$ for $t\in V(T)$ are mutually disjoint
  and their union is ${\U}$.
  \uend
\end{defn}

We introduce additional terminology and notation.  Let $(T,\beta)$ be
a tree decomposition of $\kappa$. We call the sets $\beta(t)$ the
\emph{bags} of the decomposition.
For every oriented edge $(s,t)\in\vec E(T)$ we
let $\gamma(s,t)$ be the union of the sets $\beta(t')$ for all
nodes $t'$ in the connected component of $T-st$ that contains
$t$. Note that $\gamma(s,t)=\bar{\gamma(t,s)}$. We call $\gamma(s,t)$
the \emph{cone} or the \emph{separation} of the decomposition at
$(s,t)$ and let
\[
\Sep(T,\beta):=\big\{\gamma(s,t)\bigmid (s,t)\in\vec E(T)\big\}.
\]
We always denote the cone mapping of a tree decomposition $(T,\beta)$
by $\gamma$, and we use implicit naming conventions such as denoting the
cone mapping of $(T',\beta')$ by $\gamma'$.

We could define tree decompositions based on their cones: let us
call a pair $(T,\gamma)$ consisting of a tree $T$ and a mapping
$\gamma:\vec E(T)\to 2^{\U}$ a \emph{weak decomposition} of $\kappa$
if $\gamma(s,t)=\bar{\gamma(t,s)}$ for all $(s,t)\in\vec E(T)$ and
$\gamma(s,t)\cap\gamma(s,t')=\emptyset$ for all $s\in V(T)$ and
$t,t'\in N(s)$. Then if we let 
\[
\beta(s):={\U}\setminus\bigcup_{t\in N(s)}\gamma(s,t),
\]
the pair $(T,\beta)$ is a tree decomposition of $\kappa$ with cone
mapping $\gamma$. Conversely, if
$(T,\beta)$ is a tree decomposition with cone mapping $\gamma$ then $(T,\gamma)$ is a weak
decomposition. 

In particular, every decomposition $(T,\gamma)$ is a weak
decomposition. If we define $\beta:V(T)\to 2^{\U}$ by
$\beta(t):=\gamma(s,t)$ for all leaves $t$ with $N(t)=\{s\}$ and $\beta(t):=\emptyset$ for
all internal nodes $t$, then $(T,\beta)$ is the tree decomposition
corresponding to the weak decomposition $(T,\gamma)$.
A weak decomposition is ``weaker'' than a
decomposition in two ways: the tree is not necessarily cubic, and the
union of the separations of the outgoing edges of a node is not
necessarily ${\U}$ as in a decomposition or pre-decomposition. However,
as opposed to a pre-decomposition, a weak decomposition is exact at
every node.

It is not necessary for us to define the width of a tree
decomposition. Nevertheless, the following exercise shows how it can
be done.

\begin{exe}
  Let $(T,\beta)$ be a tree decomposition of $\kappa$.
  \begin{enumerate}[label=(\alph*)]
  \item The a \emph{adhesion} of $(T,\beta)$ is
    \[
    \ad(T,\beta):=\max\big\{\kappa(\gamma(s,t))\bigmid
    (s,t)\in\vec E(T)\big\}
    \]
    if $E(T)\neq\emptyset$ and $\ad(T,\beta):= 0$ otherwise.

    The \emph{width} of $(T,\beta)$ \emph{at} a node $t\in V(T)$
    is
    \[
    \width(T,\beta,t):=\max_{X\subseteq\beta(t),\;{\U}\subseteq
      N(t)}\kappa\left(X\cup\bigcup_{u\in {\U}}\gamma(t,u)\right),
    \]
    and the \emph{width} of $(T,\beta)$ is
    $\width(T,\beta):=\max\{\width(T,\beta,t)\mid t\in V(T)\}$.

    Prove that if $(T,\gamma)$ is a decomposition (and not just a
    weak decomposition) then $\width(T,\beta)=\ad(T,\beta)$.
  \item Prove that the branch width of $\kappa$ is the minimum of the
    widths of its tree decompositions, that is, every branch
    decomposition can be transformed into a tree decomposition of at
    most the same width.\uend
  \end{enumerate}
\end{exe}

The following example shows that we cannot always transform a tree
decomposition into a decomposition with the same separations.

\begin{figure}
  \centering
  \input{triangles}
  \caption{Graph $G$ of Example~\ref{exa:triangles}}
  \label{fig:triangles}
\end{figure}
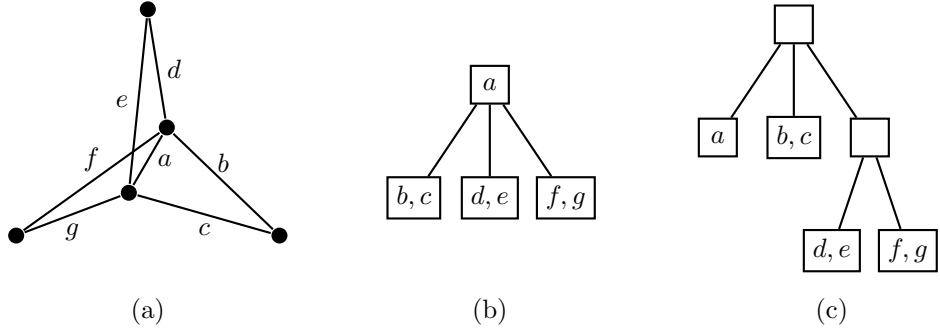

\begin{exa}\label{exa:triangles}
  Consider the graph $G$ shown in
  Figure~\ref{fig:triangles}(a). Figure~\ref{fig:triangles}(b) shows a
  tree decomposition of $\kappa_G$ and Figure~\ref{fig:triangles}(c) a
  somewhat similar decomposition. It is not hard to show that there is
  no decomposition that has the same separations as the tree
  decomposition in (b).

  Also note that the tree decomposition in (b) is invariant under
  automorphisms of the graph $G$, whereas the decomposition in (c) is
  not. In fact, if our goal is to decompose $G$ into the three
  triangles, as the tree decomposition in (b) does, then we cannot do
  this with a cubic tree. Intuitively, this should be clear; we can
  prove it by an exhaustive (and exhausting) case distinction.
  \uend
\end{exa}

\begin{rem}\label{rem:treedec}
  Let $(T,\beta)$ be a tree decomposition of a graph $G$ (in the usual
  sense defined in Section~\ref{sec:graphdec}). It yields a tree decomposition $(T,\beta')$ of $\kappa_G$ (in
  the sense defined above) as follows: for
  every edge $e\in E(G)$, we arbitrarily choose a node $t_e\in V(T)$
  that covers $e$. Then for every $t\in V(T)$ we let
  $\beta'(t):=\{e\in E(G)\mid t=t_e\}$.

  Conversely, if we have a tree decomposition $(T,\beta')$ of $\kappa_G$,
  then we can define a tree decomposition $(T,\beta)$ of $G$ as
  follows. For every node $v\in V(G)$ we let ${\U}_v$ be the set of all
  nodes $t\in V(T)$ such that $v$ is incident with an edge $e\in
  \beta'(t)$. We let $\hat {\U}_v$ be the union of ${\U}_v$ with all nodes
  $t\in V(T)$ appearing on a path between two nodes in ${\U}_v$. Now we
  let $\beta(t)=\{v\in V(G)\mid t\in\hat {\U}_v\}$. We call $(T,\beta)$
  the \emph{tree decomposition of $G$ corresponding to} $(T,\beta')$.

  Note that the construction of a tree decomposition of $\kappa_G$ from a
  tree decomposition of $G$ involves arbitrary choices, whereas the
  construction of a tree decomposition of $G$ from a
  tree decomposition of $\kappa_G$ is canonical. Thus the ``tree
  decomposition of a graph corresponding to a tree decomposition of
  its edge set'' is well-defined.
  \uend
\end{rem}

We will now characterise tree decompositions in terms of the structure
of their separations. Two sets $X,Y\subseteq {\U}$ are \emph{nested} if
either $X\subseteq Y$ or $X\subseteq\bar Y$ or $\bar X\subseteq Y$ or
$\bar X\subseteq\bar Y$; otherwise $X$ and $Y$ \emph{cross}. Note that
$X$ and $Y$ cross if and only if the four sets $X\cap Y$, $X\cap\bar
Y$, $\bar X\cap Y$, and $\bar X\cap\bar Y$ are all nonempty.  A family
$\CS\subseteq 2^{\U}$ is \emph{nested} if all $X,Y\in\CS$ are nested.
Observe that for every tree decomposition $(T,\beta)$ of $\kappa$ the set
$\Sep(T,\beta)$ is nested and closed under complementation.

The following converse of this observation goes back (at least) to \cite{gm10}.

\begin{lem}\label{lem:nested}
  Let $\CS\subseteq 2^{\U}$. Then $\CS=\Sep(T,\beta)$ for a tree
  decomposition $(T,\beta)$ of $\kappa$ if and only if $\CS$ is nested
  and closed under complementation.

  Furthermore, there is a canonical construction that associates with
  every set $\CS\subseteq 2^{\U}$ that is nested and closed under
  complementation a
  tree decomposition $(T_\CS,\beta_\CS)$ of $\kappa$ with $\Sep(T_\CS,\beta_\CS)=\CS$.
\end{lem}

Recall that a construction is canonical if every isomorphism between
two inputs of the construction commutes with an isomorphism between the
outputs.  Let us explain what this means for the construction of our
lemma. Let $\kappa$ be a connectivity function on on ${\U}$ and $\kappa'$
a connectivity function on ${\U}'$, and let $\CS\subseteq 2^{\U}$,
$\CS'\subseteq 2^{{\U}'}$ be nested and closed under complementation. 
Let
$f$ be an isomorphism $({\U},\kappa,\CS)$ to $({\U}',\kappa',\CS')$, that is, a
bijective $f:{\U}\to {\U}'$ such that $\kappa(X)=\kappa'(f(X))$ and
$X\in\CS\iff f(X)\in\CS'$ for all
$X\subseteq {\U}$. The canonicity of the construction means that for
every such $f$
there is an isomorphism $g$ from $T_{\CS}$ to $T_{\CS'}$
such that $f(\beta_{\CS}(t))=\beta_{\CS'}(g(t))$ for all $t\in
V(T_{\CS})$.

\begin{proof}[Proof of Lemma~\ref{lem:nested}]
  We have already noted that the set of separations of a tree
  decomposition is nested and closed under complementation. 

  To prove the backward direction, we describe the construction of a
  tree decomposition $(T,\beta)=(T_{\CS},\beta_{\CS})$ with
  $\Sep(T,\beta)=\CS$ from a set $\CS\subseteq 2^{\U}$ that is nested
  and closed under complementation. Canonicity will be obvious from
  the construction.

  By induction on $|\CS|$ we construct a rooted tree $(T,r)$ and a
   mapping $\beta:V(T)\to 2^{\U}$ such that $(T,\beta)$ is a tree
   decomposition with $\Sep(T,\beta)=\CS$.
   
   In the base step $\CS=\emptyset$, we let $T$ be a tree with one
   node $r$ and we define $\beta$ by $\beta(r):={\U}$.

   In the inductive step $\CS\neq\emptyset$, let $X_1,\ldots,X_m$ be a
   list of all inclusion-wise minimal elements of $\CS$. 
As $\CS$ is
   nested, for all $i\neq j$ we have $X_i\subseteq \bar X_j$. This
   implies that the sets $X_i$ are mutually disjoint. %
   Let
   \[
   \CS':=\CS\setminus\{X_i,\bar X_i\mid i\in[m]\}.
   \]
   By the induction hypothesis, there is a rooted tree $(T',r')$ and a
   mapping $\beta':V(T')\to 2^{\U}$ such that $(T',\beta')$ is a tree
   decomposition with $\Sep(T',\beta')=\CS'$. Let $r'$ be the root of
   $T'$. 

   If $X_1=\emptyset$ and $m=1$, we construct $T$ from $T'$ by adding a fresh
   child $t_0$ to the root $r'$. We let $r:=r'$ be the root of the new
   tree and define $\beta$ by $\beta(t'):=\beta'(t')$ for all $t'\in
   V(T')$ and $\beta(t_0)=\emptyset$.

   Otherwise, all the $X_i$ are nonempty. For every $i\in[m]$, let
   $t_i$ be a node of maximum
  depth such that $X_i\subseteq\gamma'(s_i,t_i)$ for the parent $s_i$ of $t_i$,
  or $t_i:=r'$ if no such node exists.\footnote{The \emph{depth} of a node in
  a rooted tree is its distance from the root.} Observe that there is only one
  such node $t_i$. Indeed, if $t\neq t_i$ has the same depth as
  $t_i$, then neither $t_i = r'$ nor $t=r'$. Let $s$ be the parent of
  $t$. Then the edges $(s_i,t_i)$ and
  $(s,t)$ are pointing away from each other and thus
  $\gamma'(s_i,t_i)\cap \gamma'(s,t)=\emptyset$. As
  $X_i\neq\emptyset$, this contradicts 
  $X_i\subseteq\gamma'(s_i,t_i)\cap\gamma'(s,t)$.

  We define a new tree $T$ from $T'$ by attaching a fresh leaf $u_i$
  to $t_i$ for every $i\in[m]$. We let $r:=r'$ be the root of $T$. We
  define $\beta\colon V(T)\to2^{\U}$ by
  \[
  \beta(t):=
  \begin{cases}
    X_i&\text{if }t=u_i,\\
    \beta'(t)\setminus\bigcup_{i=1}^mX_i&\text{if }t\in V(T').
  \end{cases}
  \]
  As the sets $X_i$ are mutually disjoint, $(T,\beta)$ is a tree
  decomposition of ${\U}$.  We need to prove that
  $\Sep(T,\beta)=\CS$.

  \begin{claim}[1]
    For all oriented edges $(s,t)\in\vec E(T')$ we have
    $\gamma(s,t)=\gamma'(s,t)$.

    \proof Let $(s,t)\in\vec E(T')$. Without loss of generality we
    assume that $t$ is a child of $s$. 
    We need to prove that 
    \[
    X_i\subseteq\gamma(s,t)\iff
    X_i\subseteq\gamma'(s,t)
    \]
    for all $i\in[m]$.

    If $X_i\subseteq\gamma(s,t)$, then $u_i$ is a descendant of
    $t$ in $T$ and
    thus $t=t_i$ or $t_i$ is a descendant of $t$ in $T'$. But as
    $X_i\subseteq\gamma'(s_i,t_i)$, this implies
    $X_i\subseteq\gamma'(s,t)$.

    For the backward direction, suppose that
    $X_i\subseteq\gamma'(s,t)$. Then $t_i=t$ or $t_i$ is a
    descendant of $t$ in $T'$, and thus $u_i$ is a descendant of $t$
    in $T$. This implies $X_i\subseteq\gamma(s,t)$.
    \uend
  \end{claim}

  To prove that $\Sep(T,\beta)\subseteq\CS$, let $X\in
  \Sep(T,\beta)$. Say, $X=\gamma(s,t)$ for some oriented edge
  $(s,t)\in\vec E(T)$.  If $(s,t)=(t_i,u_i)$ for some $i\in[m]$,
  then $X=X_i\in\CS$, and if $(s,t)=(u_i,t_i)$ then $X=\bar
  X_i\in\CS$, because $\CS$ is closed under
  complementation. Otherwise, $(s,t)\in\vec E(T')$. Then by Claim~1 we
  have $X=\gamma'(s,t)\in\CS'\subseteq\CS$.

  To prove the converse inclusion, let $X\in\CS$. If $X=X_i$ for some
  $i\in[m]$, then $X=\gamma(t_i,u_i)$, and if $X=\bar X_i$, then
  $X=\gamma(u_i,t_i)$. Otherwise, $X\in\CS'$, and thus by
  Claim~1, $X=\gamma'(s,t)=\gamma(s,t)$ for some
  $(s,t)\in\vec E(T')$.
\end{proof}

\subsection{Tangle Decompositions}

It is our goal to construct tree decompositions whose parts correspond
to tangles and whose separations separate these tangles. It will be
convenient to define such decompositions through their families of
separations.
Let $\KT$ be a
family of mutually incomparable $\kappa$-tangles. Then a \emph{nested
  set of separations} for $\KT$ is a set $\CS\subseteq 2^{\U}$ that is nested and closed
under complementation and satisfies the following two conditions.
\begin{nlist}{TN}
\item\label{li:tn1} For all $\CT,\CT'\in\KT$ with $\CT\bot\CT'$ there is a
  $Z\in\CS$ such that $Z$ is a minimum $(\CT,\CT')$-separation.
\item\label{li:tn2} For all $Z\in\CS$ there are tangles $\CT,\CT'\in\KT$ with
  $\CT\bot\CT'$ such that $Z$ is a minimum $(\CT,\CT')$-separation.
\end{nlist}
A \emph{nested set of separations} for an arbitrary (not necessarily
mutually incomparable) family $\KT$ of
$\kappa$-tangles is a nested set of separations for the family
$\KT_{\max}\subseteq\KT$ consisting of all inclusion-wise maximal tangles
in $\KT$.

The following theorem from \cite{groschwe15a} shows that nested sets
of separations for a family of tangles correspond to tree
decompositions ``displaying'' these tangles in a nice way.

\begin{theo}\label{theo:td}
  Let $\KT$ be a nonempty family of mutually incomparable
  $\kappa$-tangles.  Let $\CS$ be a nested set of separations for $\KT$. Then
  for every tree decomposition $(T,\beta)$ of $\kappa$ with
  $\Sep(T,\beta)=\CS$ there is a unique injective mapping $\tau:\KT\to V(T)$
  satisfying the following conditions.
  \begin{eroman}
  \item\label{li:td1} For all distinct $\CT,\CT'\in\KT$ there is an oriented edge
    $(t,t')\in\vec E(T)$ on the oriented path from $\tau(\CT)$ to
    $\tau(\CT')$ in $T$ such that $\gamma(t',t)$ is a minimum
    $(\CT,\CT')$-separation.
  \item\label{li:td2} For every oriented edge $(t,t')\in\vec E(T)$ there
    are tangles $\CT,\CT'\in\KT$ such that $(t,t')$ appears on the
    oriented path from $\tau(\CT)$ to $\tau(\CT')$ and $\gamma(t',t)$ is a
    minimum $(\CT,\CT')$-separation.
  \item\label{li:td3} For every tangle $\CT\in\KT$ and every neighbour
    $t'$ of $t:=\tau(\CT)$ in $T$ it holds that $\gamma(t',t)\in\CT$.
  \item\label{li:td4} For every tangle $\CT\in\KT$ and every oriented edge
    $(u',u)$ pointing towards $t:=\tau(\CT)$, if
    $\kappa(\gamma(u',u))<\ord(\CT)$ then $\gamma(u',u)\in\CT$.
  \item\label{li:td5} For all leaves $t\in L(T)$ there is a $\CT\in\KT$ such that $t=\tau(\CT)$.
  \end{eroman}
  Furthermore, there is a canonical construction that associates with
  $\CS$ a tree decomposition $(T,\beta)$ with
  $\Sep(T,\beta)=\CS$ and the unique injective mapping $\tau:\KT\to
  V(T)$ satisfying conditions (i)--(v).
\end{theo}

\begin{proof}
  Without loss of generality we assume that $|\KT|\ge 2$. Otherwise,
  by \ref{li:tn2} we have $\CS=\emptyset$ and the trivial one-node
  tree decomposition together with the unique mapping $\tau$ from
  $\KT$ to this one-node tree satisfies (i)--(v).

  By \ref{li:tn1}, $|\KT|\ge 2$ implies 
  $\CS\neq\emptyset$. Furthermore, $\ord(\CT)\ge 1$
  for all $\CT\in\KT$, because the unique tangle of order $0$ is the
  empty tangle, which is comparable with all other tangles. Let
  $(T,\beta)$ be a tree decomposition of $\kappa$ with
  $\Sep(T,\beta)=\CS$. 

  For every $k\ge 1$, we let $E_k$ be the set of all edges
  $e=tt'\in E(T)$ with $\kappa(\gamma(t,t'))<k$.

  For every tangle $\CT\in\KT$ of order $k$ we construct a connected subset
  $C_{\CT}\subseteq V(T)$ as follows: we orient all edges $e=tt'\in
  E_k$ in such a way that they
  point towards $\CT$, that is, if $\gamma(t,t')\in\CT$ then
  the orientation of $e$ is $(t,t')$ and otherwise the orientation is
  $(t',t)$. Then there is a unique connected component of $T-E_k$ (the
  forest obtained from $T$ by deleting all edges in $E_k$) such
  that all oriented edges point towards this component. We let
  $C_{\CT}$ be the node set of this connected component.

  It follows from \ref{li:tn1} that the sets $C_{\CT}$ are mutually
  vertex disjoint. To see this, consider distinct
  $\CT,\CT'\in\KT$. Let $Z\in\CS$ be a minimum $(\CT,\CT')$ separation
  and $(t,t')\in\vec E(T)$ such that $\gamma(t',t)=Z$. Such an edge
  exists, because $\Sep(T,\beta)=\CS$. Then
  $C_{\CT}$ is contained in the connected component of $T-tt'$ that
  contains $t$ and $C_{\CT'}$ is contained in the connected component
  of $T-tt'$ that contains $t'$. Hence $C_{\CT}\cap
  C_{\CT'}=\emptyset$.

  \begin{claim}[1]
    Let $\CT,\CT'\in\KT$ be distinct, and let $Z\in\CS$ be a minimum
    $(\CT,\CT')$-separation. Then every oriented edge $(t,t') \in \vec E(T)$ such that
    $\gamma(t',t)=Z$ appears on the oriented path from $C_{\CT}$
    to $C_{\CT'}$.

    \proof
    Let $(t,t')\in\vec E(T)$ such that
    $Z=\gamma(t',t)$. 
    As $Z\in\CT$
    the oriented edge $(t',t)$ points towards $C_{\CT}$, and as $\bar
    Z=\gamma(t,t')\in\CT'$ the oriented edge $(t,t')$ points
    towards $C_{\CT'}$. It follows that the oriented edge $(t,t')$
    appears on the oriented path $\vec P$ from $C_{\CT}$ to $C_{\CT'}$
    in $T$.
    \uend
  \end{claim}

  \begin{claim}[2]
    For all $\CT\in\KT$ it holds that $|C_{\CT}|=1$.

    \proof Suppose for contradiction that $|C_{\CT}|>1$ for some
    $\CT\in\KT$. Let $C:=C_{\CT}$. As $C$ is connected, there is an
    edge $e=t_1t_2\in E(T)$ with both endvertices in $C$. Then
    $e\not\in E_{\ord(\CT)}$ and thus
    $\kappa(\gamma(t_1,t_2))\ge\ord(\CT)$.

    Let $\CT_1,\CT_2\in\KT$ such that $Z:=\gamma(t_2,t_1)\in\CS$
    is a minimum $(\CT_1,\CT_2)$-separation. Such tangles exist by
    \ref{li:tn2}. For $i=1,2$, let $C_i:=C_{\CT_i}$. By Claim~1, the
    oriented edge $(t_1,t_2)$ appears on the oriented path $\vec P$
    from $C_1$ to $C_2$ in $T$.

  We have
  \[
  \ord(\CT)\le\kappa(\gamma(t_1,t_2))
  =\kappa(Z)<\min\{\ord(\CT_1),\ord(\CT_2)\}.
  \]
  Let $Z_1\in\CS$ be a minimum $(\CT_1,\CT)$-separation. Then
  $\kappa(Z_1)<\ord(\CT)\le\kappa(Z)$. Moreover, by Claim~1, there is an oriented edge $(u_1,u)$ on the oriented path
  $\vec Q$ from $C_1$ to $C$ such that
  $\gamma(u,u_1)=Z_1$. 

  We have $Z_1\in\CT_1$, because $Z_1$ is a
  $(\CT_1,\CT)$-separation. Since $t_1\in C$, the path $\vec Q$ is an
  initial segment of the path $\vec P$, and therefore $(u_1,u)$ is
  also an edge of $\vec P$. The edge $(u_1,u)$ occurs before $(t_1,t_2)$ on the path
  $\vec P$. Thus $\bar Z_1=\gamma(u_1,u)\supseteq\gamma(t_1,t_2)=\bar
  Z$, and as $\bar Z\in\CT_2$, this implies $\bar Z_1\in\CT_2$. Hence
  $Z_1$ is a $(\CT_1,\CT_2)$-separation. As $\kappa(Z_1)<\kappa(Z)$,
  this contradicts the minimality of $Z$. 
  \uend
  \end{claim}

  We define $\tau:\KT\to V(T)$ by letting $\tau(\CT)$ be the unique
  node in $C_{\CT}$, for all $\CT\in\KT$. This mapping is well-defined
  by Claim~2, and injective, because the sets $C_{\CT}$ are mutually
  disjoint.

  It follows from \ref{li:tn1} and Claim~1 that $(T,\beta,\tau)$
  satisfies \ref{li:td1}. It follows from \ref{li:tn2} and
  $\Sep(T,\beta)=\CS$ and Claim~1 that $(T,\beta,\tau)$ satisfies
  \ref{li:td2}. By the construction of $C_{\CT}$, for all oriented
  edges $(t',t)\in E(T)$ with $t\in C_{\CT}$ and $t'\not\in C_{\CT}$ it holds
  that $\gamma(t',t)\in\CT$. This implies
  that $(T,\beta,\tau)$ satisfies \ref{li:td3}. Furthermore, for all
  edges $(u',u)\in E(T)$ pointing towards $t:=\tau(\CT)$ there is a
  neighbour $t'$ of $t$ such that either $(u',u)=(t',t)$ or $(u',u)$
  points towards $t'$. In both cases,
  $\gamma(u',u)\supseteq\gamma(t',t)$. As $\gamma(t',t)\in\CT$, if
  $\kappa(\gamma(u',u'))<\ord(\CT)$, by Lemma~\ref{lem:tangle-closure}(1) we have
  $\gamma(u',u)\in\CT$. This proves \ref{li:td4}. Finally,
  \ref{li:td5} follows from \ref{li:td2}.

  To prove the uniqueness of $\tau$, suppose for contradiction that
  $\tau':\KT\to V(T)$ is another injective mapping satisfying
  \ref{li:td1}--\ref{li:td5}. Let
  $\CT\in\KT$ such that $t:=\tau(\CT)\neq\tau'(\CT)=:t'$. Let $u,u'$
  be the neighbours of $t,t'$, respectively, on the path from $t$ to
  $t'$ in $T$. Then by \ref{li:td3},
  $\gamma(t',t),\gamma(u',u)\in\CT$. However,
  $\gamma(t',t)\cap\gamma(u',u)=\emptyset$. This contradicts
  $\CT$ being a tangle.

  To construct a  tree decomposition $(T,\beta)$ with
  $\Sep(T,\beta)=\CS$, we apply Lemma~\ref{lem:nested}.
\end{proof}

\begin{cor}
  Let $\KT$ be a nonempty family of mutually incomparable
  $\kappa$-tangles. Let $(T,\beta)$ is a tree decomposition of
  $\kappa$ and $\tau:\KT\to V(T)$ a mapping satisfying conditions \ref{li:td1}
  and \ref{li:td2} of Theorem~\ref{theo:td}. Then $\tau$ is injective and satisfies \ref{li:td3}--\ref{li:td5}, and
  $\CS:=\Sep(T,\beta)$ is a nested set of separations for $\KT$.
\end{cor}

We call a triple $(T,\beta,\tau)$ where $(T,\beta)$ is a tree
decomposition for $\kappa$ and $\tau:\KT\to V(T)$ a mapping satisfying conditions \ref{li:td1}
  and \ref{li:td2} (and hence \ref{li:td3}--\ref{li:td5}) of
  Theorem~\ref{theo:td} a \emph{tree
  decomposition for $\KT$}.
Nodes $t\in \tau(\KT)$ are called \emph{tangle nodes} and the
remaining nodes $t\in V(T)\setminus\tau(\KT)$ are called \emph{hub
  nodes}.

\subsection{Decomposing Coherent Families}
\label{sec:tree2}

Let us call a family $\KT$ of $\kappa$-tangles of order $k+1$
\emph{coherent} if all elements of $\KT$ have the same truncation to order
$k$. Observe that this condition implies, and is in fact equivalent to, the
condition that for distinct $\CT,\CT'\in\KT$ the order of a minimum
$(\CT,\CT')$-separation is $k$.
The main result of this section, Lemma~\ref{lem:coherent}, shows how to
compute  a tree decomposition for a
coherent family
of tangles of order $k+1$.

We call set $Z\subseteq {\U}$ a \emph{$\KT$-separation}
if there are $\CT,\CT'\in\KT$ such that $Z$ is a
$(\CT,\CT')$-separation. $Z$ is a \emph{minimum $\KT$-separation} if
it is a $\KT$-separation of minimal order.

\begin{lem}\label{lem:coherent0}
  Let $\KT$ be a coherent family of $\kappa$-tangles of order $k+1$, and
  let $Z_0$ be an inclusion-wise minimal minimum $\KT$-separation. Then for all
  minimum $\KT$-separations $Z$, either $Z_0\subseteq Z$ or $Z_0\subseteq \bar Z$.
\end{lem}

\begin{proof}
  Let $\CT_0,\CT_0'\in\KT$ such
  that $Z_0$ is a minimum $(\CT_0,\CT_0')$-separation. Moreover, let $\CT,\CT'\in\KT$ be
  distinct, and let $Z$ be a minimum $(\CT,\CT')$-separation. Then $\kappa(Z_0)=\kappa(Z)=k$.
  Without loss of generality, we may assume that
  $    
  Z\in\CT_0.
  $
  Otherwise, we swap $\CT$ and $\CT'$ and take $\bar Z$ instead of
  $Z$.
  
  If $\kappa(Z_0\cap Z)\le k$, then $Z_0\cap
  Z\in\CT_0$ by
  Lemma~\ref{lem:tangle-closure}(2) and $\bar{Z_0\cap Z}=\bar
  Z_0\cup\bar Z\in\CT_0'$ by
  Lemma~\ref{lem:tangle-closure}(1). Thus $Z_0\cap Z$
  is a minimum $(\CT_0,\CT_0')$-separation, and by the inclusionwise minimality of
  $Z_0$ it
  follows that $Z_0\subseteq Z_0\cap Z$ and thus $Z_0\subseteq Z$.

  So let us assume $\kappa(Z_0\cap Z)> k$. By
  submodularity, $ \kappa(Z_0\cup Z)<k$. We have $Z_0\cup
  Z\in\CT_0\cap\CT$ Lemma~\ref{lem:tangle-closure}(1). Thus $\bar{Z_0\cup
    Z}=\bar Z_0\cap\bar Z\not\in\CT_0'\cup\CT'$, because otherwise $Z_0\cup Z$ is a
  $(\CT_0,\CT_0')$-separation or a $(\CT,\CT')$-separation of order
  strictly less than $k$, which is impossible because $\KT$ is a
  coherent family. By
  Lemma~\ref{lem:tangle-closure}(2), this implies $\bar
  Z_0\not\in\CT'$, $\bar Z\not\in\CT_0'$ and thus $Z_0\in\CT'$ and $Z\in\CT_0'$.

   Then both $Z_0$ and $\bar Z$ are minimum
   $(\CT',\CT_0')$-separations. As $Z_0$ is an inclusionwise minimal
   $\KT$-separation, it is a leftmost minimum
   $(\CT',\CT_0')$-separation, and hence $Z_0\subseteq\bar Z$.
\end{proof}

\begin{lem}\label{lem:coherent}
  Let $\KT$ be
  a coherent family of $\kappa$-tangles of order $k+1$. Then there is
  a nested set of separations for $\KT$.
\end{lem}

\begin{proof}
  By induction on $i\ge0$ we  define sets
  $\CS_i$ of minimum $\KT$-separations and families
  $\KT_i\subseteq\KT$ as follows.
  \begin{itemize}
  \item $\CS_0:=\emptyset$ and $\KT_0:=\emptyset$.
  \item     $\CS_{i+1}$ is the union of $\CS_i$ with all
    inclusion-wise minimal minimum $\KT\setminus\KT_i$-separations, and
    $\KT_{i+1}$ is the set of all tangles $\CT\in\KT$ such that for
    some $\CT'\in\KT$ the set
    $\CS_{i+1}$ contains a minimum $(\CT,\CT')$ separation.
  \end{itemize}
  Observe that 
  \[
  \Big|\KT\setminus\bigcup_{i\ge0}\KT_i\Big|\le
  1.
  \]
  We let $\CS$ be the closure of  $\bigcup_{i\ge0}\CS_i$ under
  complementation. We claim that $\CS$ is a nested set of separations for $\KT$.

  It follows from Lemma~\ref{lem:coherent0} that $\CS$ is nested:
  when we add  a $Z_0$ to $\CS_{i+1}$, it is nested with all
  minimum $\KT\setminus\KT_i$-separations and thus with all $Z\in\bigcup_{j\ge
    i+1}\CS_j$. 

  $\CS$ trivially satisfies \ref{li:tn2}, because each element of each
  $\CS_i$ is a minimum $\KT$-separation.

  To prove that $\CS$ satisfies \ref{li:tn1}, for all $i\ge0$
  we prove that for all $\CT\in\KT_{i+1}\setminus\KT_i$,
  $\CT'\in\KT\setminus\KT_i$ there is a $Z\in\CS_{i+1}$ such
  that $Z$ 
  is a minimum $(\CT,\CT')$-separation. As $\Big|\KT\setminus\bigcup_{i\ge0}\KT_i\Big|\le
  1$, this implies \ref{li:tn1}

Let $\CT\in\KT_{i+1}\setminus\KT_i$,
  $\CT'\in\KT\setminus\KT_i$. By the definition of $\KT_{i+1}$, there is a $Z\in\CS_{i+1}$ and a
  $\CT''\in\KT$ such that $Z$ is a minimum $(\CT,\CT'')$
  separation. By the definition of $\CS_{i+1}$, the set $Z$ is an inclusion-wise
  minimal $\KT\setminus\KT_i$-separation. Let
  $Z'$ be a minimum $(\CT,\CT')$-separation. By Lemma~\ref{lem:coherent0}, either $Z\subseteq
  Z'$ or $Z\subseteq \bar Z'$. If $Z\subseteq \bar Z'$, then $Z\cap
  Z'=\emptyset$, which contradicts $Z,Z'\in\CT$. Thus $Z\subseteq
  Z'$. By Lemma~\ref{lem:tangle-closure}(1), we have $\bar
  Z\in\CT'$, because $\bar Z\supseteq\bar Z'\in\CT'$. Thus $Z\in\CS_{i+1}$ is
  a minimum $(\CT,\CT')$-separation. 
\end{proof}

\subsection{Decomposing Arbitrary Families}
\label{sec:tree3}

In this section, we will describe how to build a ``global'' tree
decomposition of all tangles of order at most $k+1$ from ``local''
decompositions for coherent families of tangles. Let $\KT$ be a family
of $\kappa$-tangles that is closed under taking truncations. For every
$k\ge0$, we let $\KT^{\le k}$ be the tangles of order at most $k$ in
$\CT$, and we let $\KT^{\le k}_{\max}$ be the inclusionwise maximal
tangles in $\KT^{\le k}$. We call a tangle $\CT\in \KT^{\le k}_{\max}$
\emph{extendible} if there is a $\CT'\in\KT\setminus \KT^{\le k}$ such that $\CT\subseteq\CT'$.

Suppose that, for some $k\ge 0$, we have a nested family $\CS^{\le k}$
of separations for $\KT^{\le k}_{\max}$. Let $\CT^*\in \KT^{\le
  k}_{\max}$ be extendible. Then $\ord(\CT^*)=k$. Let $\KT^*$ be the set of all
$\CT\in\KT$ such that $\ord(\CT)=k+1$ and $\CT^*\subseteq\CT$.
Let $Z_1,\ldots,Z_m$ be a list of all inclusionwise minimal sets in
$\CT^*\cap\CS^{\le k}$. %

In Example~\ref{exa:contraction}, we introduced the \emph{contraction} 
$({\U}\contract_{\CA},\kappa\contract_{\CA})$ of a connectivity system
$({\U},\kappa)$  to a set $\CA$ of \emph{atoms}. Here we take the sets
$\bar Z_i$ as atoms. For the readers convenience, let me repeat the
necessary definitions, adapting them to our context and simplifying
the notation by dropping the index ${}_\CA$ everywhere in the notation. We let
\[ 
V^*:={\U}\setminus\bigcup_{i=1}^m\bar Z_i.
\]
We take fresh
  elements $z_1,\ldots,z_m\not\in {\U}$ and let
\[
{\U}\contract:=V^*\cup\{z_1,\ldots,z_m\}.
\]
For every set $X\subseteq {\U}\contract$, we define the \emph{expansion} of $X$
to be the set 
\[
X\expand:=(X\cap V^*)\cup\bigcup_{z_i\in X}\bar Z_i.
\]
Now we define $\kappa\contract:2^{{\U}\contract}\to\ZZ$ by 
\[
\kappa\contract(X):=\kappa(X\expand).
\]
Then $\kappa\contract$ is a connectivity function on ${\U}\contract$.
For every $\CT\in\{\CT^*\}\cup\KT^*$ we let
\[
\CT\contract:=\{X\subseteq {\U}\contract\mid X\expand\in\CT\}.
\]
Using the fact that $\bar Z_i\not\in\CT$ for all $i$, it is easy to
see that $\CT\contract$ is a $\kappa\contract$-tangle with
$\ord(\CT\contract)=\ord(\CT)$.

\begin{lem}\label{lem:intersection}
  Let $X\subseteq {\U}$ such that there are $\kappa$-tangles $\CT,\CT'$
  for which $X$ is a minimum $(\CT,\CT')$-separation. Then for every
  $Y\subseteq {\U}$, either $\kappa(X\cap Y)\le\kappa(Y)$ or
  $\kappa(X\setminus Y)\le\kappa(Y)$.
\end{lem}

\begin{proof}
  Suppose for contradiction that $\kappa(X\cap Y)>\kappa(Y)$ and
  $\kappa(X\setminus Y)>\kappa(Y)$. Then by submodularity,
  $\kappa(X\cup Y)<\kappa(X)$ and $\kappa(X\cup\bar Y)<\kappa(X)$. 

  Now let $\CT,\CT'$ be tangles $\CT,\CT'$
  such that $X$ is a minimum $(\CT,\CT')$-separation.
  As $X\subseteq X\cup Y,X\cup\bar Y$ and $X\in\CT$, we have $X\cup
  Y,X\cup\bar Y\in\CT$. 
  Furthermore, either $\bar{X\cup Y}\in\CT'$ or $\bar{X\cup\bar
    Y}\in\CT'$, because 
  $
  \bar X\cap(X\cup Y)\cap(X\cup\bar Y)=\emptyset.
  $
  Thus either $X\cup Y$ or $X\cup\bar Y$ is a $(\CT,\CT')$-separation
  of order less than $\kappa(X)$. This contradicts the minimality
  of $X$.
\end{proof}

\begin{lem}\label{lem:injective}
  Let $\CT,\CT'\in\KT^*$ be distinct. Then $\CT\contract$ and
  $\CT'\contract$ are distinct, and for every 
  minimum $(\CT\contract,\CT'\contract)$-separation $X$ the expansion
  $X\expand$ is a minimum $(\CT,\CT')$-separation.
\end{lem}

Note that there is a $(\CT\contract,\CT'\contract)$-separation,
because distinct tangles of the same order are incomparable.

\begin{proof}
  We choose a minimum $(\CT,\CT')$-separation $Y$ in such a way that
  it maximises the number of $i\in[m]$ with $Y\cap \bar Z_i=\emptyset$
  or $\bar Z_i\subseteq Y$. Then $\kappa(Y)=k$.

  \begin{claim}[1]
    For all $i\in[m]$, either $Y\cap\bar Z_i=\emptyset$ or $\bar
    Z_i\subseteq Y$.

    \proof Suppose for contradiction that there is some $i\in[m]$ such
    that $\emptyset\subset \bar Z_i\cap Y\subset \bar Z_i$. By
    \ref{li:tn2} for $\CS^{\le k}$, there are tangles $\CT_i,\CT_i'$ such that
    $Z_i$ is a minimum
    $(\CT_i,\CT_i')$-separation.  By Lemma~\ref{lem:intersection}
    (applied to $X:=Z_i$ and $Y$), either $\kappa(Y\cap
    Z_i)\le\kappa(Y)$ or $\kappa(\bar Y\cap Z_i)\le\kappa(Y)$. 
    
    Suppose first that $\kappa(Y\cap
    Z_i)\le k$. Then by Lemma~\ref{lem:tangle-closure}(2) we
    have $Y\cap
    Z_i\in\CT$, because $Y\in\CT$ and
    $Z_i\in\CT^*\subseteq\CT$. Furthermore, by
    Lemma~\ref{lem:tangle-closure}(1) we have $\bar{Y\cap
      Z_i}=\bar Y\cup\bar Z_i\in\CT'$, because $\bar Y\in\CT'$. Thus
    $(Y\cap Z_i)$ is a minimum $(\CT,\CT')$-separation as
    well. Furthermore, $(Y\cap Z_i)\cap\bar Z_i=\emptyset$, and for
    all $j\neq i$, if $Y\cap\bar Z_j=\emptyset$ then $(Y\cap
    Z_i)\cap\bar Z_j=\emptyset$, and if $\bar Z_j\subseteq Y$, then
    $\bar Z_j\subseteq (Y\cup Z_i)$, because $\bar
    Z_j\subseteq Z_i$. This
    contradicts the choice of $Y$.

    Suppose next that $\kappa(\bar Y\cap
    Z_i)\le k$. Arguing as above with $Y,\bar Y$ and $\CT,\CT'$
    swapped, we see that $\bar Y\cap Z_i$ is a minimum
    $(\CT',\CT)$-separation. Thus $Y\cup\bar Z_i$ is a minimum
    $(\CT,\CT')$-separation. We have $\bar Z_i\subseteq Y\cup\bar
    Z_i$, and for all $j\neq i$, if $\bar Z_j\subseteq Y$ then $\bar
    Z_j\subseteq \bar Z_i\cup Y$, and if $\bar Z_j\cap Y=\emptyset$
    then $\bar Z_j\cap(\bar Z_i\cup Y)=\bar Z_j\cap \bar
    Z_i=\emptyset$. Again, this contradicts the choice of $Y$.
    \uend
  \end{claim}

  It follows from Claim~1 that there is a $Y'\subseteq {\U}\contract$
  such that $Y=Y'\expand$.
   This set $Y'$ is a
  $(\CT\contract,\CT'\contract)$-separation. Thus the order $k'$ of a minimum
  $(\CT\contract,\CT'\contract)$-separation is at most
  $\kappa\contract(Y')=\kappa(Y)=k$.
  Now let $X'\subseteq {\U}\contract$ be a minimum
  $(\CT\contract,\CT'\contract)$-separation. Then the expansion
  $X'\expand$ is a
  $(\CT,\CT')$-separation, and this implies that
  \[
  k\le\kappa(X'\expand)=\kappa\contract(X')=k'\le k.
  \]
 Hence $k=k'$, and $X'\expand$ is a
  minimum $(\CT,\CT')$-separation.
\end{proof}

We let
\[
\KT^*\contract:=\{\CT\contract\mid\CT\in\KT^*\}.
\]
Observe that, $\KT^*\contract$ is a coherent family of
$\kappa\contract$-tangles of order $k+1$, because the truncation to
order $k$ of all elements of $\KT^*\contract$ is $\CT^*\contract$.

\begin{cor}\label{cor:injective}
  Let $\CS$ be a nested set of separations for
  $\KT^*\contract$. Then $\CS\expand:=\{ X\expand\mid X\in\CS\}$ is a
  nested set of separations for $\KT^*$.
\end{cor}

Finally, we are ready to prove the main theorem this section.

\begin{theo}[Canonical Decomposition Theorem \cite{hun11}]\label{theo:candec}
  There is a canonical construction that associates with every finite
  set
  $\KT$ of $\kappa$-tangles a nested set
  of separations for $\KT$.
\end{theo}

\begin{proof}
  Without loss of generality we assume that $\KT$ is closed under
  taking truncations. By induction on $k\ge 0$ we construct a nested
  set $\CS^{\le k}$ of separations for $\KT^{\le k}$.

  We let $\CS^{\le 0}:=\emptyset$. 

  Now let $k\ge 0$, and suppose that $\CS^{\le k}$ is a nested set of
  separations for $\KT^{\le k}$. Let $\CT_1,\ldots,\CT_n$ be a list of
  all extendible tangles in $\KT^{\le k}_{\max}$. For every $i\in[n]$,
  we let $\KT_i$ be the set of all $\CT\in\KT$ such that
  $\ord(\CT)=k+1$ and $\CT_i\subseteq\CT$. Recall that $\KT_i$ is a
  coherent family of $\kappa$-tangles of order $k+1$ and observe that
  \[
  \KT^{\le k+1}=\KT^{\le k}\cup\bigcup_{i=1}^n\KT_i.
  \]
  We apply the contraction construction described above
  with $\CT^*:=\CT_i$. We let $({\U}_i,\kappa_i)$ be the contraction of
  $({\U},\kappa)$ at $\CT_i$ with respect to $\CS^{\le k}$. 

  We let $\KT_i\contract:=\{\CT\contract\mid\CT\in\KT_i\}$. We
  construct a canonical nested set $\CS_i$ of separations for
  $\KT_i\contract$ and let
  $\CS_i\expand:=\{ Z\expand \mid Z\in\CS_i\contract\}$. Now we let
  \[
  \CS^{\le k+1}:=\CS^{\le k}\cup\bigcup_{i=1}^n\CS_i\expand.
  \]

  \begin{claim}[1]
   $\CS^{\le k+1}$ is nested.

    \proof
    We already know that the family $\CS^{\le k}$ is
    nested. Furthermore, by Corollary~\ref{cor:injective}, 
   the family $\CS_i\expand$ is nested  for every
    $i\in[n]$.

    Thus we need to show that the sets in $\CS_i\expand$ are nested with all
    sets in $\CS^{\le k}$ as well as all sets in $\CS_j\expand$ for
    $j\neq i$. So let
    $X\expand\in\CS_i\expand$. 

    First consider a $Y\in\CS^{\le k}$. We have $\ord(\CT_i)=k$,
    because $\CT_i$ is extendible. Thus either $Y\in\CT_i$ or
    $\bar Y_i\in\CT_i$. Without loss of generality we assume that
    $Y\in\CT_i$. Let $Z\subseteq Y$ be inclusionwise minimal such that
    $Z\in\CT_i$. Then either $\bar Z\subseteq X\expand$ or $\bar Z\cap
    X\expand=\emptyset$, because $\bar Z$ is one of the sets (called
    $Z_i$ above) that are contracted to a
    vertex $z$ in the construction of $({\U}_i,\kappa_i)$, and either
    $z\in X$ or $z\not\in X$.

    If $\bar Z\subseteq X\expand$ then $\bar Y\subseteq X\expand$, and
    if $\bar Z\cap
    X\expand=\emptyset$ then $\bar Y\subseteq\bar Z\subseteq\bar
    X\expand$. Thus $Y$ and $X\expand$ are nested.

    Now consider a set $X'\expand\in\CS_j\expand$ for some $j\neq i$. Let
    $Y\in\CS^{\le k}$ be a minimum $(\CT_i,\CT_j)$-separation. Then
    $Y\in\CT_i$, and by the
    argument above, either $\bar Y\subseteq X\expand$ or $\bar
    Y\subseteq\bar X\expand$. Similarly, $\bar Y\in\CT_j$ and thus
    $Y\subseteq X'\expand$ or $Y\subseteq \bar X'\expand$. This
    implies that $X\expand$ and $X'\expand$ are nested. For example,
    if $\bar Y\subseteq X\expand$ and $Y\subseteq \bar X'\expand$ then
    $X'\expand\subseteq \bar Y\subseteq X\expand$.
    \uend
  \end{claim}

  \begin{claim}[2]
   $\CS^{\le k+1}$ is a nested set of separations for $\KT^{\le k+1}$.

    \proof
    By Claim~1, $\CS^{\le k+1}$ is nested. By construction,
    it is closed under complementation. 

    $\CS^{\le k+1}$ satisfies \ref{li:tn2}, because $\CS^{\le k}$ does for all $i\in[m]$, all
    $Z\in\CS_i\expand$ are minimum separations for tangles in
    $\KT_i\subseteq\KT^{\le k+1}_{\max}$.
   
    To see that $\CS^{\le k+1}$ satisfies \ref{li:tn1}, let
    $\CT,\CT'\in\KT^{\le k+1}_{\max}$ be distinct. Let $\CT^*$ be the
    truncation of $\CT$ to order $k$ if $\ord(\CT)=k+1$ and
    $\CT^*:=\CT$ otherwise, and let $\CT^{**}$ be defined similarly
    from $\CT'$. If $\CT^*\neq\CT^{**}$, there is a $Z\in\CS^{\le k}$
    that is a minimum $(\CT^*,\CT^{**})$-separation, and this $Z$ is
    also a minimum $(\CT,\CT')$-separation. Otherwise,
    $\ord(\CT)=\ord(\CT')=k+1$ and $\CT^*=\CT^{**}=\CT_i$ for some
    $i\in[n]$. Then $\CT,\CT'\in\KT_i$, and $\CS_i\expand$ contains a minimum
    $(\CT,\CT')$-separation.  \uend
\end{claim}

Now let $\ell$ be the maximum order of a tangle in $\KT$. Then
$\CS:=\CS^{\le \ell}$ is a nested set of separations for
$\KT$. Clearly, the construction of $\CS$ is canonical.
\end{proof}

\begin{cor}
  There is a canonical construction of a nested set of separations for
  the set of all $\kappa$-tangles.
\end{cor}

\begin{cor}
  There is a canonical construction that associates with every 
  set
  $\KT$ of $\kappa$-tangles a tree decomposition for $\KT$.
\end{cor}

\begin{cor}
  Suppose that $n:=|{\U}|\ge 2$.
  Then there are at most $n-1$ maximal $\kappa$-tangles.
\end{cor}

Note that if $|{\U}|\le 1$, then the empty tangle is the unique maximal $\kappa$-tangle.

\begin{proof}
  Let $\KT_{\max}$ be the family of all maximal $\kappa$-tangles, and
  let $(T,\beta,\tau)$ be a tree decomposition for
  $\KT_{\max}$. We assume without loss of generality that
   $|\KT_{\max}|\ge 2$. Then $|T|\ge 2$. By Theorem~\ref{theo:td}\ref{li:td5}, all leaves of $T$ are tangle nodes. Let $t$ be
   a leaf, $s$ the neighbour of $t$, and let $\CT_t\in
   \KT_{\max}$ be the tangle with $\tau(\CT_t)=t$. By
   Theorem~\ref{theo:td}\ref{li:td3} we have $\beta(t)=\gamma(s,t)\in\CT_t$. By
   \ref{li:t2} and \ref{li:t3}, this implies $|\beta(t)|>1$.

   Now let $t\in V(T)$ be a tangle node of degree $2$, say, with
   neighbours $s$ and $u$. Let $\CT_t\in \KT_{\max}$
   such that $\tau(\CT_t)=t$.  By Theorem~\ref{theo:td}\ref{li:td3}, we have
   $\gamma(u,t),\gamma(s,t)\in\CT_t$, which implies
   $\beta(t)=\gamma(u,t)\cap\gamma(s,t)\neq\emptyset$.

   Let $n_1,n_2,n_{\ge 3}$ be the numbers of tangle nodes of degree
   $1$, $2$, at least $3$, respectively. We have $2n_1+n_2\le |{\U}|$.
   Furthermore, $n_{\ge 3}<n_1$, because a tree with $n_1$ leaves has
   less than $n_1$ nodes of degree at least $3$. Thus
   \[
   |\KT_{\max}|=n_1+n_2+n_{\ge 3}<2n_1+n_2\le|{\U}|.\qedhere
   \]
\end{proof}

\section{Algorithmic Aspects}
\label{sec:alg}

In this section, we will briefly cover the main algorithmic aspects of
the theory. There are essentially three types of algorithmic results.
\begin{enumerate}
\item Algorithms for computing decompositions of bounded width. We
  will focus on branch decompositions for general connectivity systems
  below. A few references to algorithms for computing tree
  decompositions are \cite{arncorpro87,bod96,bodgrodre+13,feihajlee08,gm13}.
\item Algorithms for computing tangles and the canonical tangle tree
  decompositions.
\item Algorithm for solving otherwise hard algorithmic problems
  efficiently on structures of bounded branch or tree width. We will
  not consider these here. Some pointers to the rich literature are
  \cite{bod97,bodkos08,cou90,cygfomkow+15,coumakrot00,dowfel13,flufrigro02,flugro06,ggmss05,nie06,oum16}.
\end{enumerate}
Before we state any results, we need to describe the computation
model, which is not obvious for algorithms on connectivity
functions. Specifying a connectivity function explicitly requires exponential
space in
the size of the universe, and this is usually not what we
want. Specific connectivity functions are usually given implicitly. For example,
an algorithm that takes one of the connectivity functions
$\kappa_G,\nu_G,\mu_G,\rho_G$ as its input will usually just be given
the graph $G$, and an algorithm taking the connectivity function of a
representable matroid as its input will be given a representation of
the matroid. However, for the general theory we assume a more abstract
computation model that applies to all connectivity functions in the
same way.

In this model, algorithms expecting a connectivity function or any
other set function $\kappa:2^{\U}\to\NN$ as input are
given the universe ${\U}$ as actual input (say, as a list of objects),
and they are given an oracle that returns for~$X\subseteq {\U}$ the value
of~$\kappa({\U})$. The running time of such algorithms is measured in
terms of the size $|{\U}|$ of the universe, which in the following we
denote by $n$. We assume this computation
model whenever we say that an algorithm is given \emph{oracle access} to a
set function $\kappa$. 

A  fact underlying most of the following algorithms is that, under this model of
computation, submodular functions can be efficiently minimised \cite{iwaflefuj01,schri00}.

Seymour and Oum where the first to consider the computation of branch
decompositions in this abstract setting. In \cite{oumsey06a}, they
showed that given oracle access to a univalent connectivity function $\kappa$ of
branch width $k$, there is an algorithm computing a branch
decomposition of $\kappa$ of width at most $3k+1$ in time
$f(k)n^{O(1)}$ for some function $f$. At the price of an increase of
the approximation ratio by a factor of $\val(\kappa)$, this can be
generalised to connectivity functions that are not necessarily
univalent. Later, Oum and Seymour~\cite{oumsey07} showed that branch
decomposition of width exactly $k$ can be computed in time
$n^{O(k)}$. Hlinen{\'{y}} and Oum~\cite{hlioum08} showed that for
certain well-behaved connectivity functions, among them the
connectivity functions of matroids representable over a finite field
and the cut-rank function of graphs, this can be improved to
$f(k)n^3$.

To discuss algorithms for computing the tangle tree decomposition, we
first need to know how to deal with tangles
algorithmically. In~\cite{groschwe15a}, Schweitzer and I designed a data structure representing
all $\kappa$-tangles of order at most $k$ and providing basic
functionalities for these tangles, such as a membership oracle or a
function returning a leftmost minimum separation for two
tangles. Given oracle access to $\kappa$, this data structure can be
computed in time $n^{O(k)}$. Using this data structure, we showed
that the canonical tree decomposition for the family $\KT^{\le k}$ of
all $\kappa$-tangles of order at most $k$ can be computed in time $n^{O(k)}$.

\end{document}

%% file: separation.tex
\begin{tikzpicture}
  \begin{scope}
  \clip (0,0) ellipse (3cm and 2cm);
  \draw[blau,thick,fill=blue!10] (0,0) ellipse (3cm and 2cm);
  \draw[very thick] (-5.5,0) circle (5cm);
  \end{scope}
  \path (-3,2) node[anchor=north west] {${\U}$}
        (-1.75,0) node {$X$}
        (1.25,0) node {$\bar X={\U}\setminus X$}
        ;
  \path (0,-2.5) node {$\kappa(X)=\kappa(\bar X)=$ order of the separation};
\end{tikzpicture}


%% file: hexgrid-sep.tex
\begin{tikzpicture}
  \begin{scope}
    
  \path (0,0) node {$\includegraphics[width=4cm]{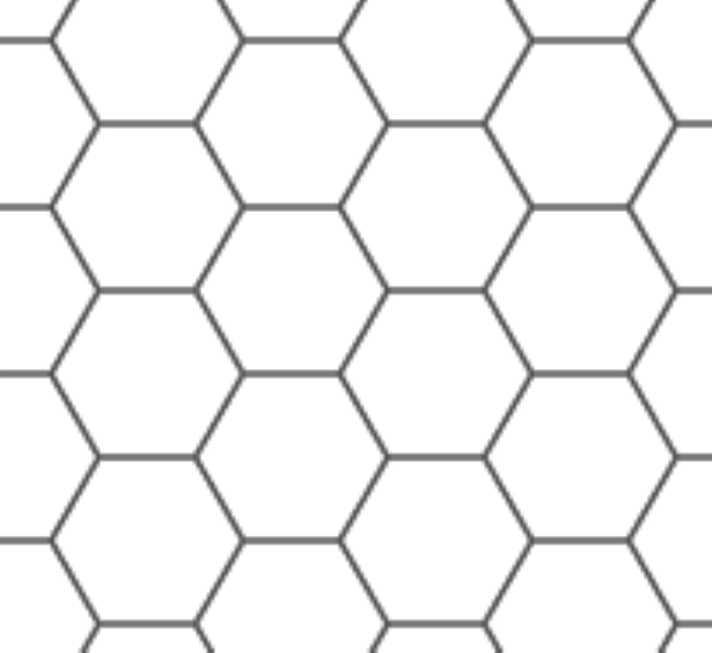}$};
  
  \draw[red,very thick] (-0.09,-0.28) circle (5.5mm);

  \path (0,-2.6) node[anchor=south] {(a)};
  \end{scope}

  \begin{scope}[xshift=6cm]    
  \path (0,0) node {$\includegraphics[width=4cm]{hexgrid-smallpart.pdf}$};
  
  \draw[red,very thick] (-0.09,-0.28) circle (5.5mm);
  \draw[gruen,very thick] (0.18,0.2) circle (5.5mm);

  \path (0,-2.6) node[anchor=south] {(b)};
  \end{scope}

\end{tikzpicture}


%% file: cutrank.tex
  \begin{tikzpicture}[
      vertex/.style={draw,circle,minimum size=6mm,inner sep=0pt},
      every edge/.style={draw,thick},
      ]

       \draw[gruen, fill=gruen!20] (-0.5,0) ellipse (1cm and 2cm);
      \draw[rot, fill=rot!20] (2,0) ellipse (10mm and 2.8cm);

      \node[vertex] (a) at (-0.5,1.5) {$a$};
      \node[vertex] (b) at (0,0) {$b$};
      \node[vertex] (c) at (-0.5,-1.5) {$c$};
      \node[vertex] (d) at (2,2.25) {$d$};
     \node[vertex] (e) at (2.5,0.75) {$e$};
     \node[vertex] (f) at (2.5,-0.75) {$f$};
     \node[vertex] (g) at (2,-2.25) {$g$};
 
     \path (a) edge (b) edge[bend right] (c) (b) edge (c)
               (e) edge (f)
               (a) edge (d) edge (e) edge (f)
               (b) edge (d) edge (g)
               (c) edge (e) edge (f) edge (g);

     \path (-0.5,-2.6) node[anchor=south,gruen] {$X$};          
     \path (2,-3.4) node[anchor=south,rot] {$\bar X$};      

     \path (5,1.6) node[anchor=north west] {$
        M(X,\bar X)=\begin{pmatrix}
         1&1&1&0\\
         1&0&0&1\\
         0&1&1&1
       \end{pmatrix}
       $
     };
     \path (7,-0.2) node [anchor=west] {rows labelled $a,b,c$,};
     \path (7,-0.6) node [anchor=west] {columns labelled $d,e,f,g$};
     \path (5,-0.9) node[anchor=north west] {$\rho_G(X)=2$};
  \end{tikzpicture}


%% file: htw.tex
    \begin{tikzpicture}
      \path (-2,0) node {$1$} (0,0) node {$4$} (2,0) node {$2$}
      (-2,-1.5) node {$3$} (2,-1.5) node {$5$};

      \draw (-2,-1.5) circle (3mm)
            (0,0) circle (3mm)
            (2,-1.5) circle (3mm)
            (0,0) ellipse (25mm and 7mm)
            (-2,-0.75) ellipse (6mm and 14mm)
            (2,-0.75) ellipse (6mm and 14mm);

      \path (0,-2.8) node {(a)};   

      \begin{scope}
      [
      xshift=5cm,
      ]
      \node (a) at (0,0) {$a$};
      \node (b) at (2,0) {$b$};
      \node (c) at (4,0) {$c$};
      \node (d) at (0,-2) {$d$};
      \node (e) at (2,-2) {$e$};
      \node (f) at (4,-2) {$f$};

      \draw (1,0) ellipse (14mm and 5mm)
            (3,0) ellipse (14mm and 5mm)
            (0,-1) ellipse (5mm and 14mm) 
            (2,-1) ellipse (5mm and 14mm) 
            (4,-1) ellipse (5mm and 14mm);

      \path (2,-2.8) node {(b)};      
      \end{scope}
 
      
    \end{tikzpicture}


%% file: dbranchdec1.tex
  \begin{tikzpicture}[
      leaf/.style={draw,inner sep=0.5mm},
      inner/.style={circle,fill=black,inner sep=0mm,minimum size=2mm},
      every edge/.style={draw,thick,->},
    ]

    \begin{scope}
    \footnotesize
    \node[leaf] (L1) at (0,0) {$\begin{pmatrix}
        1\\0\\0\\0
      \end{pmatrix}$};
    \node[leaf] (L2) at (1,0) {$\begin{pmatrix}
        0\\1\\0\\0
      \end{pmatrix}$};
    \node[leaf] (L3) at (2,0) {$\begin{pmatrix}
        1\\1\\0\\0
      \end{pmatrix}$};
    \node[leaf] (L4) at (3,0) {$\begin{pmatrix}
      1\\1\\1\\1
    \end{pmatrix}$};
    \node[leaf] (L5) at (4,0) {$\begin{pmatrix}
        0\\0\\1\\0
      \end{pmatrix}$};
    \node[leaf] (L6) at (5,0) {$\begin{pmatrix}
        0\\0\\0\\1
      \end{pmatrix}$};
    
    \small
    \node[inner] (I1) at (0.5,1.2) {};
    \draw[->] (I1) edge (L1) edge (L2);

    \node[inner] (I2) at (1,2) {};
    \draw[->] (I2) edge (I1) edge (L3);

    \node[inner] (I3) at (2,2.8) {};
    \draw[->] (I3) edge (I2) edge (L4);

    \node[inner] (I4) at (4.5,1.2) {};
    \draw[->] (I4) edge (L5) edge (L6);

    \node[inner] (I5) at (3,3.6) {};
    \draw[->] (I5) edge (I3) edge (I4);

    \node[anchor=south east] at (I2) {\normalsize$t$};

    \normalsize
    \node at (2.5,-1.2) {(a)};
    \end{scope}

    \begin{scope}[xshift=\textwidth/2]
    \footnotesize
    \node[leaf] (L7) at (0,0) {$\begin{pmatrix}
        1\\0\\0\\0
      \end{pmatrix}$};
    \node[leaf] (L8) at (1,0) {$\begin{pmatrix}
        0\\0\\1\\0
      \end{pmatrix}$};
    \node[leaf] (L9) at (2,0) {$\begin{pmatrix}
        1\\1\\0\\0
      \end{pmatrix}$};
    \node[leaf] (L10) at (3,0) {$\begin{pmatrix}
      1\\1\\1\\1
    \end{pmatrix}$};
    \node[leaf] (L11) at (4,0) {$\begin{pmatrix}
        0\\1\\0\\0
      \end{pmatrix}$};
    \node[leaf] (L12) at (5,0) {$\begin{pmatrix}
        0\\0\\0\\1
      \end{pmatrix}$};
    
    \small
    \node[inner] (I9) at (0.5,1.2) {};
    \draw (I9) edge (L7) edge (L8);

    \node[inner] (I6) at (1,2) {};
    \draw (I6) edge (I9) edge (L9);

    \node[inner] (I7) at (4.5,1.2) {};
    \draw (I7) edge (L11) edge (L12);

     \node[inner] (I8) at (4,2) {};
    \draw (I8) edge (L10) edge (I7);
   
     \node[inner] (I10) at (2.5,3.5) {};
     \draw[->] (I10) edge (I6) edge (I8);

     \normalsize
    \node at (2.5,-1.2) {(b)};
    \end{scope}
  \end{tikzpicture}


%% file: rankdec.tex
  \begin{tikzpicture}[
      vertex/.style={draw,circle,minimum size=6mm,inner sep=0pt},
      leaf/.style={draw,minimum size=6mm,inner sep=0pt},
      inner/.style={fill=black,circle,minimum size=2mm,inner sep=0pt},
      every edge/.style={draw,thick},
      ]

      \begin{scope}
      \node[vertex] (a) at (-0.5,1.5) {$a$};
      \node[vertex] (b) at (0,0) {$b$};
      \node[vertex] (c) at (-0.5,-1.5) {$c$};
      \node[vertex] (d) at (2,2.25) {$d$};
     \node[vertex] (e) at (2.5,0.75) {$e$};
     \node[vertex] (f) at (2.5,-0.75) {$f$};
     \node[vertex] (g) at (2,-2.25) {$g$};
 
     \path (a) edge (b) edge[bend right] (c) (b) edge (c)
               (e) edge (f)
               (a) edge (d) edge (e) edge (f)
               (b) edge (d) edge (g)
               (c) edge (e) edge (f) edge (g);

     \path (0.75,-3) node {(a) a graph $G$};
     \end{scope}

     \begin{scope}[xshift=5cm,yshift=-1.5cm]
       \node[leaf] (Lb) at (0,0) {$b$};
       \node[leaf] (La) at (1,0) {$a$};
       \node[leaf] (Lc) at (2,0) {$c$};
       \node[leaf] (Ld) at (3,0) {$d$};
       \node[leaf] (Lg) at (4,0) {$g$};
       \node[leaf] (Le) at (5,0) {$e$};
       \node[leaf] (Lf) at (6,0) {$f$};

       \node[inner] (I1) at (1.5,1) {};
       \draw[->] (I1) edge (La) edge (Lc);

       \node[inner] (I2) at (1,2) {};
       \draw[->] (I2) edge (Lb) edge (I1);

      \node[inner] (I3) at (3.5,1) {};
       \draw[->] (I3) edge (Ld) edge (Lg);

      \node[inner] (I4) at (5.5,1) {};
       \draw[->] (I4) edge (Le) edge (Lf);

      \node[inner] (I5) at (4.5,2) {};
       \draw[->] (I5) edge (I3) edge (I4);

      \node[inner] (I6) at (2.75,3) {};
       \draw[->] (I6) edge (I2) edge (I5);

     \path (3,-1.5) node {(b) a branch decomposition of $V(G)$};

       \end{scope}
   \end{tikzpicture}


%% file: branchdec1.tex
  \begin{tikzpicture}[
      leaf/.style={draw,inner sep=0.5mm},
      inner/.style={circle,fill=black,inner sep=0mm,minimum size=2mm},
      every edge/.style={draw,thick},
    ]

    \begin{scope}
    \footnotesize
    \node[leaf] (L1) at (0,0) {$\begin{pmatrix}
        1\\0\\0\\0
      \end{pmatrix}$};
    \node[leaf] (L2) at (1,0) {$\begin{pmatrix}
        0\\1\\0\\0
      \end{pmatrix}$};
    \node[leaf] (L3) at (2,0) {$\begin{pmatrix}
        1\\1\\0\\0
      \end{pmatrix}$};
    \node[leaf] (L4) at (3,0) {$\begin{pmatrix}
      1\\1\\1\\1
    \end{pmatrix}$};
    \node[leaf] (L5) at (4,0) {$\begin{pmatrix}
        0\\0\\1\\0
      \end{pmatrix}$};
    \node[leaf] (L6) at (5,0) {$\begin{pmatrix}
        0\\0\\0\\1
      \end{pmatrix}$};
    
    \small
    \node[inner] (I1) at (0.5,1.2) {};
    \draw (I1) edge (L1) edge (L2);

    \node[inner] (I2) at (1,2) {};
    \draw (I2) edge (I1) edge (L3);

    \node[inner] (I3) at (2,2.8) {};
    \draw (I3) edge (I2) edge (L4);

    \node[inner] (I4) at (4.5,1.2) {};
    \draw (I4) edge (L5) edge (L6);

    \draw (I3) edge (I4);

    \normalsize
    \node at (2.5,-1.2) {(a)};
    \end{scope}

    \begin{scope}[xshift=\textwidth/2]
    \footnotesize
    \node[leaf] (L7) at (0,0) {$\begin{pmatrix}
        1\\0\\0\\0
      \end{pmatrix}$};
    \node[leaf] (L8) at (1,0) {$\begin{pmatrix}
        0\\0\\1\\0
      \end{pmatrix}$};
    \node[leaf] (L9) at (2,0) {$\begin{pmatrix}
        1\\1\\0\\0
      \end{pmatrix}$};
    \node[leaf] (L10) at (3,0) {$\begin{pmatrix}
      1\\1\\1\\1
    \end{pmatrix}$};
    \node[leaf] (L11) at (4,0) {$\begin{pmatrix}
        0\\1\\0\\0
      \end{pmatrix}$};
    \node[leaf] (L12) at (5,0) {$\begin{pmatrix}
        0\\0\\0\\1
      \end{pmatrix}$};
    
    \small
    \node[inner] (I9) at (0.5,1.2) {};
    \draw (I9) edge (L7) edge (L8);

    \node[inner] (I6) at (1,2) {};
    \draw (I6) edge (I9) edge (L9);

    \node[inner] (I7) at (4.5,1.2) {};
    \draw (I7) edge (L11) edge (L12);

     \node[inner] (I8) at (4,2) {};
    \draw (I8) edge (L10) edge (I7);
   
     \draw (I6) edge[bend left] (I8);

    \node[anchor=south east] at (I6) {\normalsize$t$};
    \node[anchor=south west] at (I8) {\normalsize$u$};

     \normalsize
    \node at (2.5,-1.2) {(b)};
    \end{scope}
  \end{tikzpicture}


%% file: tangleidea.tex
  \begin{tikzpicture}
    \begin{scope}[scale=0.8]

    \path (0,0) node {\includegraphics[height=5.6cm]{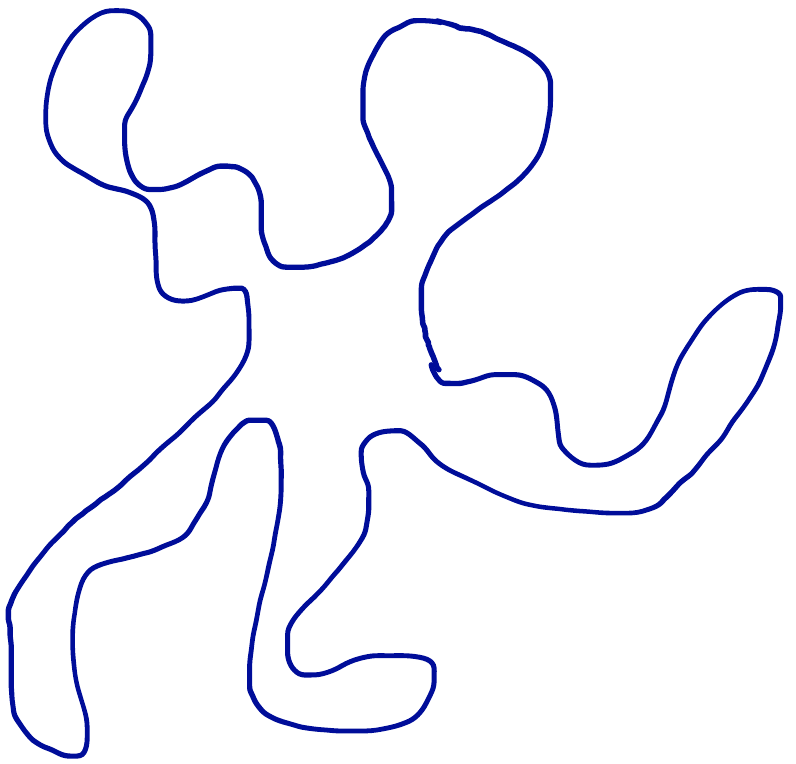}};

      \path (-0.5,-4.2) node {(a) Universe of a connectivity function};
  \end{scope}


  \begin{scope}[scale=0.8,xshift=8cm]

    \draw[dashed,blau,fill=blau!20] (-0.6,0.3) circle (1.1cm);
    \path (0,0) node {\includegraphics[height=5.6cm]{tangle.pdf}};

      \path (-0.5,-4.2) node {(b) A $k$-connected region};
  \end{scope}


  \begin{scope}[scale=0.8,yshift=-9cm]

    \draw[dashed,blau,fill=blau!20] (-0.6,0.3) circle (1.1cm);
    \path (0,0) node {\includegraphics[height=5.6cm]{tangle.pdf}};


      \coordinate (gruen1) at (-0.5,1.9);
      \coordinate (gruen2) at -(0.8,0.7);
      \draw[gruen,very thick] (gruen1)--(gruen2); 

      \coordinate (rot1) at (-0.6,1.4);
      \coordinate (rot2) at (1,1.2) ;
      \draw[rot,very thick] (rot1)--(rot2);

      \coordinate (mittelblau1) at (-0.55,1.6);
      \coordinate (mittelblau2) at -(0.9,1) ;
      \draw[mittelblau,very thick] (mittelblau1)--(mittelblau2); 
      
      \coordinate (orange1) at (-1.1,1.3);
      \coordinate (orange2) at (-1.5,0.5);
      \draw[orange,very thick] (orange1)--(orange2); 
      
      \coordinate (gelb1) at (-0.9,1.2);
      \coordinate (gelb2) at (-1.6,0.6);
      \draw[gelb,very thick] (gelb1)--(gelb2); 
      
      \coordinate (blau1) at (-2.1,2.1);
      \coordinate (blau2) at (-2.6,1.5);
      \draw[blau,very thick] (blau1)--(blau2); 
      
      \coordinate (violett1) at (-1.2,-0.7);
      \coordinate (violett2) at (-1.8,0);
      \draw[violett,very thick] (violett1)--(violett2); 
      
      \coordinate (tuerkis1) at (-3.5,-1.5);
      \coordinate (tuerkis2) at (-2.6,-2.1);
      \draw[tuerkis,very thick] (tuerkis1)--(tuerkis2); 
      
      \coordinate (rosa1) at (-3.4,-3.6);
      \coordinate (rosa2) at (-2.6,-3);
      \draw[rosa,very thick] (rosa1)--(rosa2); 
      
      \coordinate (hellblau1) at (-1.3,-0.7);
      \coordinate (hellblau2) at (-0.1,-0.6);
      \draw[hellblau,very thick] (hellblau1)--(hellblau2); 
      
      \coordinate (black1) at (-1.6,-2.4);
      \coordinate (black2) at (-0.7,-2.4);
      \draw[black,very thick] (black1)--(black2); 
      
      \coordinate (bordeaux1) at (0,-0.7);
      \coordinate (bordeaux2) at (0.6,0.3);
      \draw[bordeaux,very thick] (bordeaux1)--(bordeaux2); 
      
      \coordinate (maigruen1) at (1.7,-0.5);
      \coordinate (maigruen2) at (2,-1.5);
      \draw[maigruen,very thick] (maigruen1)--(maigruen2); 
      
      \coordinate (magenta1) at (1.9,-0.5);
      \coordinate (magenta2) at (1.7,-1.5);
      \draw[magenta,very thick] (magenta1)--(magenta2);

      \path (-0.5,-4.2) node {(c) Separations of order $<k$};
  \end{scope}


  \begin{scope}[scale=0.8,xshift=8cm,yshift=-9cm]

    \draw[dashed,blau,fill=blau!20] (-0.6,0.3) circle (1.1cm);
    \path (0,0) node {\includegraphics[height=5.6cm]{tangle.pdf}};


      \coordinate (gruen1) at (-0.5,1.9);
      \coordinate (gruen2) at -(0.8,0.7);
      \draw[gruen,very thick] (gruen1)--(gruen2); 

      \coordinate (rot1) at (-0.6,1.4);
      \coordinate (rot2) at (1,1.2) ;
      \draw[rot,very thick] (rot1)--(rot2);

      \coordinate (mittelblau1) at (-0.55,1.6);
      \coordinate (mittelblau2) at -(0.9,1) ;
      \draw[mittelblau,very thick] (mittelblau1)--(mittelblau2); 
      
      \coordinate (orange1) at (-1.1,1.3);
      \coordinate (orange2) at (-1.5,0.5);
      \draw[orange,very thick] (orange1)--(orange2); 
      
      \coordinate (gelb1) at (-0.9,1.2);
      \coordinate (gelb2) at (-1.6,0.6);
      \draw[gelb,very thick] (gelb1)--(gelb2); 
      
      \coordinate (blau1) at (-2.1,2.1);
      \coordinate (blau2) at (-2.6,1.5);
      \draw[blau,very thick] (blau1)--(blau2); 
      
      \coordinate (violett1) at (-1.2,-0.7);
      \coordinate (violett2) at (-1.8,0);
      \draw[violett,very thick] (violett1)--(violett2); 
      
      \coordinate (tuerkis1) at (-3.5,-1.5);
      \coordinate (tuerkis2) at (-2.6,-2.1);
      \draw[tuerkis,very thick] (tuerkis1)--(tuerkis2); 
      
      \coordinate (rosa1) at (-3.4,-3.6);
      \coordinate (rosa2) at (-2.6,-3);
      \draw[rosa,very thick] (rosa1)--(rosa2); 
      
      \coordinate (hellblau1) at (-1.3,-0.7);
      \coordinate (hellblau2) at (-0.1,-0.6);
      \draw[hellblau,very thick] (hellblau1)--(hellblau2); 
      
      \coordinate (black1) at (-1.6,-2.4);
      \coordinate (black2) at (-0.7,-2.4);
      \draw[black,very thick] (black1)--(black2); 
      
      \coordinate (bordeaux1) at (0,-0.7);
      \coordinate (bordeaux2) at (0.6,0.3);
      \draw[bordeaux,very thick] (bordeaux1)--(bordeaux2); 
      
      \coordinate (maigruen1) at (1.7,-0.5);
      \coordinate (maigruen2) at (2,-1.5);
      \draw[maigruen,very thick] (maigruen1)--(maigruen2); 
      
      \coordinate (magenta1) at (1.9,-0.5);
      \coordinate (magenta2) at (1.7,-1.5);
      \draw[magenta,very thick] (magenta1)--(magenta2);
    
      \coordinate (gruen3) at ( $ (gruen1) ! .45 ! (gruen2) $ );
      \coordinate (gruen4) at ( $(gruen3)!6mm!270:(gruen2)$);
      \draw[gruen,ultra thick,->] (gruen3)--(gruen4);
      
      \coordinate (rot3) at ( $ (rot1) ! .4 ! (rot2) $ );
      \coordinate (rot4) at ( $(rot3)!6mm!270:(rot2)$);
      \draw[rot,ultra thick,->] (rot3)--(rot4);
      
      \coordinate (mittelblau3) at ( $ (mittelblau1) ! .4 ! (mittelblau2) $ );
      \coordinate (mittelblau4) at ( $(mittelblau3)!6mm!270:(mittelblau2)$);
      \draw[mittelblau,ultra thick,->] (mittelblau3)--(mittelblau4);
      
      \coordinate (orange3) at ( $ (orange1) ! .4! (orange2) $ );
      \coordinate (orange4) at ( $(orange3)!6mm!90:(orange2)$);
      \draw[orange,ultra thick,->] (orange3)--(orange4);

      \coordinate (gelb3) at ( $ (gelb1) ! .5! (gelb2) $ );
      \coordinate (gelb4) at ( $(gelb3)!6mm!90:(gelb2)$);
      \draw[gelb,ultra thick,->] (gelb3)--(gelb4);
      
      \coordinate (blau3) at ( $ (blau1) ! .55! (blau2) $ );
      \coordinate (blau4) at ( $(blau3)!6mm!90:(blau2)$);
      \draw[blau,ultra thick,->] (blau3)--(blau4);
      
      \coordinate (violett3) at ( $ (violett1) ! .5! (violett2) $ );
      \coordinate (violett4) at ( $(violett3)!6mm!270:(violett2)$);
      \draw[violett,ultra thick,->] (violett3)--(violett4);
      
      \coordinate (tuerkis3) at ( $ (tuerkis1) ! .5! (tuerkis2) $ );
      \coordinate (tuerkis4) at ( $(tuerkis3)!6mm!90:(tuerkis2)$);
      \draw[tuerkis,ultra thick,->] (tuerkis3)--(tuerkis4);

      \coordinate (rosa3) at ( $ (rosa1) ! .5! (rosa2) $ );
      \coordinate (rosa4) at ( $(rosa3)!6mm!90:(rosa2)$);
      \draw[rosa,ultra thick,->] (rosa3)--(rosa4);
      
      \coordinate (hellblau3) at ( $ (hellblau1) ! .5! (hellblau2) $ );
      \coordinate (hellblau4) at ( $(hellblau3)!6mm!90:(hellblau2)$);
      \draw[hellblau,ultra thick,->] (hellblau3)--(hellblau4);
      
      \coordinate (black3) at ( $ (black1) ! .59! (black2) $ );
      \coordinate (black4) at ( $(black3)!6mm!90:(black2)$);
      \draw[black,ultra thick,->] (black3)--(black4);
      
      \coordinate (maigruen3) at ( $ (maigruen1) ! .35! (maigruen2) $ );
      \coordinate (maigruen4) at ( $(maigruen3)!6mm!270:(maigruen2)$);
      \draw[maigruen,ultra thick,->] (maigruen3)--(maigruen4);
      
      \coordinate (magenta3) at ( $ (magenta1) ! .4! (magenta2) $ );
      \coordinate (magenta4) at ( $(magenta3)!6mm!270:(magenta2)$);
      \draw[magenta,ultra thick,->] (magenta3)--(magenta4);

      \coordinate (bordeaux3) at ( $ (bordeaux1) ! .59! (bordeaux2) $ );
      \coordinate (bordeaux4) at ( $(bordeaux3)!6mm!90:(bordeaux2)$);
      \draw[bordeaux,ultra thick,->] (bordeaux3)--(bordeaux4);

      \path (-0.5,-4.2) node {(d) Tangle describing the region};
  \end{scope}


  \begin{scope}[scale=0.8,yshift=-18cm]

    \draw[dashed,gruen,fill=gruen!20] (1,-0.6) circle (9mm);
    \path (0,0) node {\includegraphics[height=5.6cm]{tangle.pdf}};


      \coordinate (gruen1) at (-0.5,1.9);
      \coordinate (gruen2) at -(0.8,0.7);
      \draw[gruen,very thick] (gruen1)--(gruen2); 

      \coordinate (rot1) at (-0.6,1.4);
      \coordinate (rot2) at (1,1.2) ;
      \draw[rot,very thick] (rot1)--(rot2);

      \coordinate (mittelblau1) at (-0.55,1.6);
      \coordinate (mittelblau2) at -(0.9,1) ;
      \draw[mittelblau,very thick] (mittelblau1)--(mittelblau2); 
      
      \coordinate (orange1) at (-1.1,1.3);
      \coordinate (orange2) at (-1.5,0.5);
      \draw[orange,very thick] (orange1)--(orange2); 
      
      \coordinate (gelb1) at (-0.9,1.2);
      \coordinate (gelb2) at (-1.6,0.6);
      \draw[gelb,very thick] (gelb1)--(gelb2); 
      
      \coordinate (blau1) at (-2.1,2.1);
      \coordinate (blau2) at (-2.6,1.5);
      \draw[blau,very thick] (blau1)--(blau2); 
      
      \coordinate (violett1) at (-1.2,-0.7);
      \coordinate (violett2) at (-1.8,0);
      \draw[violett,very thick] (violett1)--(violett2); 
      
      \coordinate (tuerkis1) at (-3.5,-1.5);
      \coordinate (tuerkis2) at (-2.6,-2.1);
      \draw[tuerkis,very thick] (tuerkis1)--(tuerkis2); 
      
      \coordinate (rosa1) at (-3.4,-3.6);
      \coordinate (rosa2) at (-2.6,-3);
      \draw[rosa,very thick] (rosa1)--(rosa2); 
      
      \coordinate (hellblau1) at (-1.3,-0.7);
      \coordinate (hellblau2) at (-0.1,-0.6);
      \draw[hellblau,very thick] (hellblau1)--(hellblau2); 
      
      \coordinate (black1) at (-1.6,-2.4);
      \coordinate (black2) at (-0.7,-2.4);
      \draw[black,very thick] (black1)--(black2); 
      
      \coordinate (bordeaux1) at (0,-0.7);
      \coordinate (bordeaux2) at (0.6,0.3);
      \draw[bordeaux,very thick] (bordeaux1)--(bordeaux2); 
      
      \coordinate (maigruen1) at (1.7,-0.5);
      \coordinate (maigruen2) at (2,-1.5);
      \draw[maigruen,very thick] (maigruen1)--(maigruen2); 
      
      \coordinate (magenta1) at (1.9,-0.5);
      \coordinate (magenta2) at (1.7,-1.5);
      \draw[magenta,very thick] (magenta1)--(magenta2);
    
      \coordinate (gruen3) at ( $ (gruen1) ! .45 ! (gruen2) $ );
      \coordinate (gruen4) at ( $(gruen3)!6mm!270:(gruen2)$);
      \draw[gruen,ultra thick,->] (gruen3)--(gruen4);
      
      \coordinate (rot3) at ( $ (rot1) ! .4 ! (rot2) $ );
      \coordinate (rot4) at ( $(rot3)!6mm!270:(rot2)$);
      \draw[rot,ultra thick,->] (rot3)--(rot4);
      
      \coordinate (mittelblau3) at ( $ (mittelblau1) ! .4 ! (mittelblau2) $ );
      \coordinate (mittelblau4) at ( $(mittelblau3)!6mm!270:(mittelblau2)$);
      \draw[mittelblau,ultra thick,->] (mittelblau3)--(mittelblau4);
      
      \coordinate (orange3) at ( $ (orange1) ! .4! (orange2) $ );
      \coordinate (orange4) at ( $(orange3)!6mm!90:(orange2)$);
      \draw[orange,ultra thick,->] (orange3)--(orange4);

      \coordinate (gelb3) at ( $ (gelb1) ! .5! (gelb2) $ );
      \coordinate (gelb4) at ( $(gelb3)!6mm!90:(gelb2)$);
      \draw[gelb,ultra thick,->] (gelb3)--(gelb4);
      
      \coordinate (blau3) at ( $ (blau1) ! .55! (blau2) $ );
      \coordinate (blau4) at ( $(blau3)!6mm!90:(blau2)$);
      \draw[blau,ultra thick,->] (blau3)--(blau4);
      
      \coordinate (violett3) at ( $ (violett1) ! .5! (violett2) $ );
      \coordinate (violett4) at ( $(violett3)!6mm!270:(violett2)$);
      \draw[violett,ultra thick,->] (violett3)--(violett4);
      
      \coordinate (tuerkis3) at ( $ (tuerkis1) ! .5! (tuerkis2) $ );
      \coordinate (tuerkis4) at ( $(tuerkis3)!6mm!90:(tuerkis2)$);
      \draw[tuerkis,ultra thick,->] (tuerkis3)--(tuerkis4);

      \coordinate (rosa3) at ( $ (rosa1) ! .5! (rosa2) $ );
      \coordinate (rosa4) at ( $(rosa3)!6mm!90:(rosa2)$);
      \draw[rosa,ultra thick,->] (rosa3)--(rosa4);
      
      \coordinate (hellblau3) at ( $ (hellblau1) ! .5! (hellblau2) $ );
      \coordinate (hellblau4) at ( $(hellblau3)!6mm!90:(hellblau2)$);
      \draw[hellblau,ultra thick,->] (hellblau3)--(hellblau4);
      
      \coordinate (black3) at ( $ (black1) ! .59! (black2) $ );
      \coordinate (black4) at ( $(black3)!6mm!90:(black2)$);
      \draw[black,ultra thick,->] (black3)--(black4);
      
      \coordinate (maigruen3) at ( $ (maigruen1) ! .35! (maigruen2) $ );
      \coordinate (maigruen4) at ( $(maigruen3)!6mm!270:(maigruen2)$);
      \draw[maigruen,ultra thick,->] (maigruen3)--(maigruen4);
      
      \coordinate (magenta3) at ( $ (magenta1) ! .4! (magenta2) $ );
      \coordinate (magenta4) at ( $(magenta3)!6mm!270:(magenta2)$);
      \draw[magenta,ultra thick,->] (magenta3)--(magenta4);

      \coordinate (bordeaux3) at ( $ (bordeaux1) ! .59! (bordeaux2) $ );
      \coordinate (bordeaux4) at ( $(bordeaux3)!6mm!270:(bordeaux2)$);
      \draw[bordeaux,ultra thick,->] (bordeaux3)--(bordeaux4);

      \path (-0.5,-4.3) node {(e) \parbox[t]{4cm}{Another tangle describing a
        different region}};
  \end{scope}


  \begin{scope}[scale=0.8,xshift=8cm,yshift=-18cm]

    \path (0,0) node {\includegraphics[height=5.6cm]{tangle.pdf}};


      \coordinate (gruen1) at (-0.5,1.9);
      \coordinate (gruen2) at -(0.8,0.7);
      \draw[gruen!20,very thick] (gruen1)--(gruen2); 

      \coordinate (rot1) at (-0.6,1.4);
      \coordinate (rot2) at (1,1.2) ;
      \draw[rot!20,very thick] (rot1)--(rot2);

      \coordinate (mittelblau1) at (-0.55,1.6);
      \coordinate (mittelblau2) at -(0.9,1) ;
      \draw[mittelblau!20,very thick] (mittelblau1)--(mittelblau2); 
      
      \coordinate (orange1) at (-1.1,1.3);
      \coordinate (orange2) at (-1.5,0.5);
      \draw[orange!20,very thick] (orange1)--(orange2); 
      
      \coordinate (gelb1) at (-0.9,1.2);
      \coordinate (gelb2) at (-1.6,0.6);
      \draw[gelb!20,very thick] (gelb1)--(gelb2); 
      
      \coordinate (blau1) at (-2.1,2.1);
      \coordinate (blau2) at (-2.6,1.5);
      \draw[blau!20,very thick] (blau1)--(blau2); 
      
      \coordinate (violett1) at (-1.2,-0.7);
      \coordinate (violett2) at (-1.8,0);
      \draw[violett,very thick] (violett1)--(violett2); 
      
      \coordinate (tuerkis1) at (-3.5,-1.5);
      \coordinate (tuerkis2) at (-2.6,-2.1);
      \draw[tuerkis!20,very thick] (tuerkis1)--(tuerkis2); 
      
      \coordinate (rosa1) at (-3.4,-3.6);
      \coordinate (rosa2) at (-2.6,-3);
      \draw[rosa,very thick] (rosa1)--(rosa2); 
      
      \coordinate (hellblau1) at (-1.3,-0.7);
      \coordinate (hellblau2) at (-0.1,-0.6);
      \draw[hellblau,very thick] (hellblau1)--(hellblau2); 
      
      \coordinate (black1) at (-1.6,-2.4);
      \coordinate (black2) at (-0.7,-2.4);
      \draw[black!20,very thick] (black1)--(black2); 
      
      \coordinate (bordeaux1) at (0,-0.7);
      \coordinate (bordeaux2) at (0.6,0.3);
      \draw[bordeaux,very thick] (bordeaux1)--(bordeaux2); 
      
      \coordinate (maigruen1) at (1.7,-0.5);
      \coordinate (maigruen2) at (2,-1.5);
      \draw[maigruen!20,very thick] (maigruen1)--(maigruen2); 
      
      \coordinate (magenta1) at (1.9,-0.5);
      \coordinate (magenta2) at (1.7,-1.5);
      \draw[magenta!20,very thick] (magenta1)--(magenta2);
    
      \coordinate (gruen3) at ( $ (gruen1) ! .45 ! (gruen2) $ );
      \coordinate (gruen4) at ( $(gruen3)!6mm!270:(gruen2)$);
      \draw[gruen!20,ultra thick,->] (gruen3)--(gruen4);
      
      \coordinate (rot3) at ( $ (rot1) ! .4 ! (rot2) $ );
      \coordinate (rot4) at ( $(rot3)!6mm!270:(rot2)$);
      \draw[rot!20,ultra thick,->] (rot3)--(rot4);
      
      \coordinate (mittelblau3) at ( $ (mittelblau1) ! .4 ! (mittelblau2) $ );
      \coordinate (mittelblau4) at ( $(mittelblau3)!6mm!270:(mittelblau2)$);
      \draw[mittelblau!20,ultra thick,->] (mittelblau3)--(mittelblau4);
      
      \coordinate (orange3) at ( $ (orange1) ! .4! (orange2) $ );
      \coordinate (orange4) at ( $(orange3)!6mm!90:(orange2)$);
      \draw[orange!20,ultra thick,->] (orange3)--(orange4);

      \coordinate (gelb3) at ( $ (gelb1) ! .5! (gelb2) $ );
      \coordinate (gelb4) at ( $(gelb3)!6mm!90:(gelb2)$);
      \draw[gelb!20,ultra thick,->] (gelb3)--(gelb4);
      
      \coordinate (blau3) at ( $ (blau1) ! .55! (blau2) $ );
      \coordinate (blau4) at ( $(blau3)!6mm!90:(blau2)$);
      \draw[blau!20,ultra thick,->] (blau3)--(blau4);
      
      \coordinate (violett3) at ( $ (violett1) ! .5! (violett2) $ );
      \coordinate (violett4) at ( $(violett3)!6mm!90:(violett2)$);
      \draw[violett,ultra thick,->] (violett3)--(violett4);
      
      \coordinate (tuerkis3) at ( $ (tuerkis1) ! .5! (tuerkis2) $ );
      \coordinate (tuerkis4) at ( $(tuerkis3)!6mm!90:(tuerkis2)$);
      \draw[tuerkis!20,ultra thick,->] (tuerkis3)--(tuerkis4);

      \coordinate (rosa3) at ( $ (rosa1) ! .5! (rosa2) $ );
      \coordinate (rosa4) at ( $(rosa3)!6mm!90:(rosa2)$);
      \draw[rosa!20,ultra thick,->] (rosa3)--(rosa4);
      
      \coordinate (hellblau3) at ( $ (hellblau1) ! .5! (hellblau2) $ );
      \coordinate (hellblau4) at ( $(hellblau3)!6mm!90:(hellblau2)$);
      \draw[hellblau!20,ultra thick,->] (hellblau3)--(hellblau4);
      
      \coordinate (black3) at ( $ (black1) ! .59! (black2) $ );
      \coordinate (black4) at ( $(black3)!6mm!90:(black2)$);
      \draw[black!20,ultra thick,->] (black3)--(black4);
      
      \coordinate (maigruen3) at ( $ (maigruen1) ! .35! (maigruen2) $ );
      \coordinate (maigruen4) at ( $(maigruen3)!6mm!270:(maigruen2)$);
      \draw[maigruen!20,ultra thick,->] (maigruen3)--(maigruen4);
      
      \coordinate (magenta3) at ( $ (magenta1) ! .4! (magenta2) $ );
      \coordinate (magenta4) at ( $(magenta3)!6mm!270:(magenta2)$);
      \draw[magenta!20,ultra thick,->] (magenta3)--(magenta4);

      \coordinate (bordeaux3) at ( $ (bordeaux1) ! .59! (bordeaux2) $ );
      \coordinate (bordeaux4) at ( $(bordeaux3)!6mm!270:(bordeaux2)$);
      \draw[bordeaux,ultra thick,->] (bordeaux3)--(bordeaux4);

      \path (-0.5,-4.3) node {(f) \parbox[t]{4cm}{An inconsistent orientation of the separations}};
  \end{scope}

  \end{tikzpicture}

%% file: exadec.tex
    \begin{tikzpicture}[
    thick,
    ]
    

    \foreach \name/\a/\r in {
      a/0/1.5cm, b/90/1.5cm, c/180/1.5cm, d/270/1.5cm,
      e/45/1.7cm, f/45/2.6cm,
      g/135/1.7cm, h/135/2.6cm,
      i/225/1.7cm, j/225/2.6cm,
      k/315/1.7cm, l/315/2.6cm
    }
    \node[coordinate] (\name) at (\a:\r) {};  

    \node[coordinate] (s) at (0,0) {};

    \foreach \name/\a/\r in {
      t/270/1.5cm,
      m/240/2.7cm, n/252/2.7cm,
      o/264/2.7cm, p/276/2.7cm,
      q/288/2.7cm, r/300/2.7cm
    }
    \node[coordinate,yshift=-1cm] (\name) at (\a:\r) {};

    \fill[red!20] (a)--(b)--(c)--(d);
    \fill[blau!25] (a)--(b)--(f);
    \fill[blau!40] (b)--(c)--(h);
    \fill[blau!55] (c)--(d)--(j);
    \fill[blau!10] (d)--(a)--(l);
    \fill[gruen!20] (t)--(m)--(n);
    \fill[gruen!40] (t)--(o)--(p);
    \fill[gruen!60] (t)--(q)--(r);

    \fill[black!20] (0,-2) ellipse (3mm and 5mm);

    \foreach \v in {a,b,...,t}
    \fill (\v) circle (3pt);

    \foreach \v/\w in {
      a/b, b/c, c/d, d/a, 
      a/e, a/f, b/e, b/f, e/f,
      b/g, b/h, c/g, c/h, g/h,
      c/i, c/j, d/i, d/j, i/j,
      d/k, d/l, a/k, a/l, k/l,
      t/m, t/n, m/n,
      t/o, t/p, o/p,
      t/q, t/r, q/r,
      s/a, s/b, s/c, s/d,
      d/t%
    }
    \draw[thick] (\v)--(\w);
    
    \draw[thick] (d)--(t);

    \begin{scope}[
      xshift=6cm,
      tn/.style={circle,draw,thin,fill,inner sep=1.5mm},
      level/.style={sibling distance=8mm},
      ]
      \node[tn,fill=rot!20] {}
        child { node[tn,fill=blau!10] {} }
        child { node[tn,fill=blau!25] {} }
        child { node[tn,fill=blau!40] {} }
        child { node[tn,fill=blau!55] {} }
        child { node[tn,fill=black!20] {} 
          child { node[tn,fill=gruen!20] {} } 
          child { node[tn,fill=gruen!40] {} } 
          child { node[tn,fill=gruen!60] {} } 
        }
      ;
    \end{scope}
    
    \path (0,-4.5) node {(a)} (6,-4.5) node {(b)};
  \end{tikzpicture}


%% file: exadec2.tex
    \begin{tikzpicture}[
    thick,
    ]
    

    \foreach \name/\a/\r in {
      a/0/1.5cm, b/90/1.5cm, c/180/1.5cm, d/270/1.5cm,
      e/45/1.7cm, f/45/2.6cm,
      g/135/1.7cm, h/135/2.6cm,
      i/225/1.7cm, j/225/2.6cm,
      k/315/1.7cm, l/315/2.6cm
    }
    \node[coordinate] (\name) at (\a:\r) {};  


    \foreach \name/\a/\r in {
      t/270/1.5cm,
      m/240/2.7cm, n/252/2.7cm,
      o/264/2.7cm, p/276/2.7cm,
      q/288/2.7cm, r/300/2.7cm
    }
    \node[coordinate,yshift=-1cm] (\name) at (\a:\r) {};

    \fill[blau!25] (a)--(b)--(f);
    \fill[blau!40] (b)--(c)--(h);
    \fill[blau!55] (c)--(d)--(j);
    \fill[blau!10] (d)--(a)--(l);
    \fill[gruen!20] (t)--(m)--(n);
    \fill[gruen!40] (t)--(o)--(p);
    \fill[gruen!60] (t)--(q)--(r);

    \fill[black!20] (0,-2) ellipse (3mm and 5mm);

    \foreach \v in {a,b,...,r,t}
    \fill (\v) circle (3pt);

    \foreach \v/\w in {
      a/b, b/c, c/d, d/a, 
      a/e, a/f, b/e, b/f, e/f,
      b/g, b/h, c/g, c/h, g/h,
      c/i, c/j, d/i, d/j, i/j,
      d/k, d/l, a/k, a/l, k/l,
      t/m, t/n, m/n,
      t/o, t/p, o/p,
      t/q, t/r, q/r,
      d/t%
    }
    \draw[thick] (\v)--(\w);
    
    \draw[thick] (d)--(t);

    \begin{scope}[
      xshift=6cm,
      tn/.style={circle,draw,thin,fill,inner sep=1.5mm},
      level/.style={sibling distance=8mm},
      ]
      \node[tn,inner sep=1mm,fill=white] {}
        child { node[tn,fill=blau!10] {} }
        child { node[tn,fill=blau!25] {} }
        child { node[tn,fill=blau!40] {} }
        child { node[tn,fill=blau!55] {} }
        child { node[tn,fill=black!20] {} 
          child { node[tn,fill=gruen!20] {} } 
          child { node[tn,fill=gruen!40] {} } 
          child { node[tn,fill=gruen!60] {} } 
        }
      ;
    \end{scope}
    
    \path (0,-4.5) node {(a)} (6,-4.5) node {(b)};
  \end{tikzpicture}


%% file: triangles.tex
    \begin{tikzpicture}[
    thick,
    vertex/.style={fill=black,circle,inner sep=0pt,minimum width=6pt},
    tn/.style={draw,minimum height=5mm,minimum width=5mm},
    ]
    
    \node[vertex] (v) at (60:5mm) {};
    \node[vertex] (w) at (240:5mm) {};
    \node[vertex] (x) at (-30:20mm) {};
    \node[vertex] (y) at (90:20mm) {};
    \node[vertex] (z) at (210:20mm) {};
   
    \draw[thick] (v) edge node[right] {$a$} (w);
    \draw[thick] (v) edge node[above] {$b$} (x);
    \draw[thick] (w) edge node[below] {$c$} (x);
    \draw[thick] (v) edge node[right] {$d$} (y);
    \draw[thick] (w) edge node[left] {$e$} (y);
    \draw[thick] (v) edge node[above] {$f$} (z);
    \draw[thick] (w) edge node[below] {$g$} (z);

    \begin{scope}[
      xshift=4.5cm,
      yshift=1cm,
      level/.style={sibling distance=10mm},
      ]
      \node[tn] {$a$}
        child { node[tn] {$b,c$} }
        child { node[tn] {$d,e$} }
        child { node[tn] {$f,g$} }
      ;
    \end{scope}

    \begin{scope}[
      xshift=8.5cm,
      yshift=1.8cm,
      level/.style={sibling distance=10mm},
      ]
      \node[tn] {}
        child { node[tn] {$a$} }
        child { node[tn] {$b,c$} }
        child {node [tn] {}
          child { node[tn] {$d,e$} }
          child { node[tn] {$f,g$} }
        }
      ;
    \end{scope}
    
    \path (0,-2) node {(a)} (4.5,-2) node {(b)} (9,-2) node {(c)};

  \end{tikzpicture}
